\documentclass[letterpaper,11pt]{article}
\usepackage{palatino}
\usepackage[scaled]{helvet} 
\usepackage{courier} 
\normalfont
\usepackage[T1]{fontenc}
\usepackage{verbatim}
\usepackage{xspace}
\usepackage{amsmath}
\usepackage{amsthm}
\usepackage{amssymb}
\usepackage{graphicx}
\usepackage{epstopdf}
\usepackage{url}
\usepackage{fontenc}
\usepackage[utf8]{inputenc}
\usepackage{mathtools}
\usepackage{fullpage}
\usepackage{verbatim}
\usepackage{paralist}
\usepackage{amsfonts}
\usepackage{amsmath}
\usepackage{amsthm}
\usepackage{latexsym}
\usepackage{url}
\usepackage{algorithm,algorithmic}
\usepackage{thmtools}
\usepackage{thm-restate}
\usepackage{hyperref}
\usepackage{cleveref}
\usepackage{complexity}
\usepackage{color}
\usepackage[dvipsnames]{xcolor}

\declaretheorem[name=Theorem,numberwithin=section]{thm}

\hypersetup{colorlinks=true,linkcolor=blue,citecolor=blue,urlcolor=blue}
\usepackage[top=0.8in, bottom=0.8in, left=1in, right=1in]{geometry}

\newtheorem{theorem}{Theorem}[section]
\newtheorem{cor}[theorem]{Corollary}
\newtheorem{lemma}[theorem]{Lemma}

\newtheorem{claim}[theorem]{Claim}

\newtheorem{definition}[theorem]{Definition}

\newtheorem*{theorem*}{Theorem}
\newtheorem*{corollary*}{Corollary}
\newtheorem*{conjecture*}{Conjecture}
\newtheorem*{lemma*}{Lemma}
\newtheorem*{thm*}{Theorem}
\newtheorem*{prop*}{Proposition}
\newtheorem*{claim*}{Claim}
\newtheorem*{cor*}{Corollary}
\newtheorem*{obs*}{Observation}
\newtheorem*{rem*}{Remark}
\newtheorem*{definition*}{Definition}

\newtheorem*{rec*}{Recommendation}

\DeclareMathOperator{\DD}{{\cal D}}
\DeclareMathOperator{\eps}{\gamma}

\addtocounter{page}{-2}

\newcommand{\INPUT}{\item[{\bf Input:}]}

\newcommand{\ignore}[1]{}%
\newcommand{\Ai}{A_{\text{-}i}}
\newcommand{\score}{F}
\newcommand{\MS}{\hat{A}}
\DeclareMathOperator{\good}{\texttt{good}(A)}
\DeclareMathOperator{\bad}{\texttt{bad}(A)}
\DeclareMathOperator{\pdf}{\rho}
\DeclareMathOperator{\mg}{\theta_g}
\DeclareMathOperator{\mb}{\theta_b}

\advance \topmargin by -\headheight
\advance \topmargin by -\headsep
\evensidemargin \oddsidemargin
\newcommand{\rmk}[1]{&&\color{lightgray}{\textit{#1}}}
\newcommand{\OPT}{\texttt{OPT}}

\newcommand{\prob}[2][]{\text{\bf P}\ifthenelse{\not\equal{}{#1}}{_{#1}}{}\!\left(#2\right)}
\newcommand{\expect}[2][]{\text{\bf E}\ifthenelse{\not\equal{}{#1}}{_{#1}}{}\!\left[#2\right]}
\newcommand{\var}[2][]{\text{\bf Var}\ifthenelse{\not\equal{}{#1}}{_{#1}}{}\!\left[#2\right]}

\DeclareMathOperator{\argmax}{argmax}
\DeclareMathOperator{\argmin}{argmin}

\title{Submodular Optimization under Noise}

\author{
  Avinatan Hassidim\footnote{Supported by ISF 1241/12;}\\
 Bar Ilan University\\
  \texttt{avinatan@cs.biu.ac.il}
    \and
  Yaron Singer\footnote{Supported by NSF grant CCF-1301976, CAREER CCF-1452961, Google Faculty Research Award, Facebook Faculty Award.}\\
  Harvard University\\
  \texttt{yaron@seas.harvard.edu}
}

\begin{document}

\frenchspacing
\sloppy
\widowpenalty10000
\clubpenalty10000
\date{}
\maketitle

\begin{abstract}
We consider the problem of maximizing a monotone submodular function under noise.  There has been a great deal of work on optimization of submodular functions under various constraints, resulting in algorithms that provide desirable approximation guarantees.  In many applications, however, we do not have access to the submodular function we aim to optimize, but rather to some erroneous or noisy version of it.  This raises the question of whether provable guarantees are obtainable in presence of error and noise. We provide initial answers, by focusing on the question of maximizing a monotone submodular function under a cardinality constraint when given access to a noisy oracle of the function.  We show that:
\begin{itemize}
  \item For a cardinality constraint $k \geq 2$, there is an approximation algorithm whose approximation ratio is arbitrarily close to $1-1/e$;
  \item For $k=1$ there is an algorithm whose approximation ratio is arbitrarily close to $1/2$.  No randomized algorithm can obtain an approximation ratio better than $1/2+o(1)$;
  \item If the noise is adversarial, no non-trivial approximation guarantee can be obtained.
\end{itemize}
\end{abstract}

\newpage
\tableofcontents
\clearpage
\addtocounter{page}{-1}
\pagenumbering{arabic}

\newpage \section{Introduction}
In this paper we study the effects of error and noise on submodular optimization.
A function $f:2^{N} \to \mathbb{R}$ defined on a ground set $N$ of size $n$ is submodular if for any $S,T\subseteq N$:
$$f(S \cup T) \leq f(S) + f(T) - f(S \cap T)$$
Equivalently, submodularity can be defined in terms of a natural diminishing returns property.  For any $A,B \subseteq N$ let $f_{A}(B) = f(A \cup B) - f(A)$, then $f$ is submodular if $\forall S \subseteq T \subseteq N, a \in N\setminus T$:
$$f_{S}(a) \geq f_{T}(a).$$
In general, submodular functions may require a representation that is exponential in the size of the ground set and the assumption is that we are given access to a \emph{value oracle} which given a set $S$ returns $f(S)$.  It is well known that submodular functions admit desirable approximation guarantees and are heavily used in applications such as market design, data mining, and machine learning (see related work).  For the classic problem of maximizing a monotone (i.e. $S\subseteq T \implies f(S) \leq f(T)$) submodular function under a cardinality constraint, the greedy algorithm which iteratively adds the element with largest marginal contribution into the solution obtains a $1-1/e$ approximation~\cite{nemhauser1978} which is optimal unless using exponentially-many queries~\cite{nemhauser1978best} or P=NP~\cite{feige1998threshold}.

Since submodular functions can be exponentially representative, it may be reasonable to assume that there are cases where one faces some error in their evaluation.  In market design where submodular functions often model agents' valuations for goods, it seems reasonable to assume that agents do not precisely know their valuations.
Even with compact representation, evaluation of a submodular function may be prone to error.
In learning and sketching submodular functions, the algorithms produce an approximate version of the function~\cite{GHIM09,DBH11,BCIW12-valuations,BDFKNR12-sketches,FKV13-submodular_trees,FV13-submodular_juntas,DLBS14_a,DLBS14_b,FK14-coverage,FV15-submodular_polynomials,Balcan15-AAMAS}.
\begin{center}
\emph{Can we retain desirable approximation guarantees in the presence of error?}
\end{center}
For $f:2^{N} \to \mathbb{R}$ and $\epsilon>0$ we say that $\widetilde{f}:2^{N} \to \mathbb{R}$ is $\epsilon$-erroneous if for every set $S \subseteq N$, it respects:
$$(1-\epsilon)f(S) \leq \widetilde{f}(S) \leq (1+\epsilon)f(S)$$
For the canonical problem of $\max_{S:|S|\leq k}f(S)$, one can trivially approximate the solution within a factor of $\frac{1-\epsilon}{1+\epsilon}$ using ${n}\choose{k}$ queries with an $\epsilon$-erroneous oracle by simply evaluating all possible subsets and returning the best solution (according to the erroneous oracle).  Is there a polynomial-time algorithm that can obtain desirable approximation guarantees for maximizing a monotone submodular function under a cardinality constraint given access to $\epsilon$-erroneous oracles? In Appendix~\ref{sec:examples} we sketch an example showing that the celebrated greedy algorithm fails to obtain an approximation strictly better than $O(1/k)$ for any constant $\epsilon>0$ when given access to an $\epsilon$-erroneous oracle $\widetilde{f}$ instead of $f$.  It turns out that this is not intrinsic to greedy. No algorithm is robust to small errors.

\begin{theorem*}[\ref{thm:adversarial}]
No randomized algorithm can obtain an approximation strictly better than $O(n^{-1/2+\delta})$ to maximizing monotone submodular functions under a cardinality constraint using $e^{n^{\delta}}/n$ queries to an $\epsilon$-erroneous oracle, for any
fixed $\epsilon,\delta<1/2$, with high probability.
\end{theorem*}

Since desirable guarantees are generally impossible with erroneous oracles, we seek natural relaxations of the problem.  The first could be to consider stricter classes of functions.  It is trivial to show for example, that \emph{additive} functions (i.e. $f(S)=\sum_{a\in S}f(a)$) allow us to obtain a $\frac{1-\epsilon}{1+\epsilon}$ approximation when given access to $\epsilon$-erroneous oracles.
Unfortunately, it seems like there are not many interesting classes of submodular functions that enjoy these properties.  In fact,
our impossibility result applies to very simple affine functions, and even coverage functions like the example in Appendix~\ref{sec:examples}.  An alternative relaxation is to consider error models that are not necessarily adversarial.

\paragraph{Noisy oracles.}  We can equivalently say that ${\widetilde{f}:2^{N} \to \mathbb{R}}$ is $\epsilon$-erroneous if for every $S \subseteq N$ we have that $\widetilde{f}(S) = \xi_{S}f(S)$ for some $\xi_{S} \in [1-\epsilon,1+\epsilon]$.  The lower bound stated above applies to the case in which the error multipliers $\xi_{S}$  are adversarially chosen.  A natural question is whether some relaxation of the adversarial error model can lead to possibility results.

\begin{definition*}\label{noisyOracle}
For a function $f:2^{N}\to \mathbb{R}$ we say that $\widetilde{f}:2^{N}\to \mathbb{R}$ is a \emph{\textbf{noisy}} oracle if there exists some distribution $\mathcal{D}$ s.t. $\widetilde{f}(S) = \xi_{S}f(S)$ where $\xi_{S}$ is independently drawn from $\mathcal{D}$ for every $S \subseteq N$.
\end{definition*}

Note that the noisy oracle defined above is \emph{consistent}: for any $S \subseteq N$ the noisy oracle returns the same answer regardless of how many times it is queried.  When the noisy oracle is inconsistent, mild conditions on the noise distribution allow the noise to essentially vanish after logarithmically-many queries, reducing the problem to standard submodular maximization (see e.g.~\cite{KKT03,STK16}).  Consistency implies that the noise is arbitrarily correlated for a given set in different time steps, but i.i.d between different sets.  In fact, we will later generalize the model to the case in which $\xi_{S}$ and $\xi_{T}$ are i.i.d only when $S$ and $T$ are sufficiently far, and arbitrarily correlated otherwise (see Section~\ref{sec:applications}).  At this point, we are interested in identifying a natural non worst-case model of corrupted or approximately submodular functions that is amendable to optimization.

We will be interested in a class of distributions that avoids trivialities like  $\mathcal{D} \subseteq \{0\}$ and is yet general enough to contain natural distributions.  In this paper we define a class which we call \emph{generalized exponential tail} distributions that contains Gaussian, Exponential, and distributions with bounded support which are independent of $n$ (o.w. optimization is impossible, see Appendix~\ref{sec:distributions}).  Note that optimization in this setting always requires that $n$ is sufficiently large. For example, if for every $S$ the noise is s.t. $\xi_{S} = 2^{100}$ with probability $1/ 2^{100}$ and $0$ otherwise, but $n = 50$, it is likely that the noisy oracle will always return $0$, in which case we cannot do better than selecting an element at random. Throughout the paper we assume that $n$ is sufficiently large.
\begin{definition*}
A noise distribution $\mathcal{D}$ has a \textbf{generalized exponential tail} if there exists some $x_0$ such that for $x > x_0$ the probability density function $\pdf(x) = e^{-g(x)}$, where $g(x) = \sum_{i} a_ix^{\alpha_i}$. We do not assume that all the $\alpha_i$'s are integers, but only that $\alpha_0 \ge \alpha_1 \ge \ldots$, and that $\alpha_0 \ge 1$. If $\mathcal{D}$ has bounded support we only require that either it has an atom at its supremum, or that $\pdf$ is continuous and non zero at the supremum.
\end{definition*}
For simplicity, one can always consider the special case where $\mathcal{D} \subseteq [1-\epsilon,1+\epsilon]$, which implies that two sets whose true values are close will remain close in the noisy evaluation.  Even when the noise distribution is uniform in $[1-\epsilon,1+\epsilon]$ it is easy to show that the greedy algorithm fails (see Appendix~\ref{sec:examples}).  The question is whether provable guarantees are achievable in this model.

\subsection{Main result}\label{sec:intro_main}
Our main result is that for the problem of optimizing a monotone submodular function under a cardinality constraint, near-optimal approximations are achievable under noise.

\begin{theorem*}
For any monotone submodular function there is a polynomial-time algorithm which optimizes the function under a cardinality constraint $k> 2$ and obtains an approximation ratio that is w.h.p arbitrarily close to $1-1/e$ using access to a generalized exponential tail noisy oracle of the function.
\end{theorem*}

This proof is a summary of three results, each for a different regime of $k$.  For any $\epsilon>0$ we show:
\begin{itemize}
\item \textbf{$1-1/e-\epsilon$ guarantee for large $k$:} we say that $k$ is \emph{large} when $k\in \Omega( \log \log n/\epsilon^2)$.  For $k$ that is sufficiently larger than $\log \log n/\epsilon^2$ we give a deterministic algorithm which obtains a $(1-1/e-\epsilon)$ approximation guarantee w.h.p over the noise distribution;

\item \textbf{$1-1/e-\epsilon$ guarantee for small $k$:}  we say that $k$ is \emph{small} when $k\in O(\log \log n)\cap \Omega(1/\epsilon)$.  In this regime the problem is surprisingly harder.  We give a different deterministic algorithm which achieves the coveted $(1-1/e-\epsilon$) guarantee, w.h.p. over the noise distribution;

\item \textbf{Guarantees for very small $k$:}  We say that $k$ is \emph{very small} when it is an arbitrarily small constant.  For this case we give a randomized algorithm whose approximation ratio is $1-1/k-\epsilon$ w.h.p. over the randomization of the algorithm and the noise distribution.  Note that this gives $1-1/e-\epsilon$ for any $k>2$, and $1/2-\epsilon$ for $k=2$.  We also give a $k/(k+1)$ approximation which holds in expectation over the randomization of the algorithm. This achieves $1-1/e$ for $k=2$ and $1/2$ for $k=1$.  For $k=1$ no randomized algorithm can obtain an approximation ratio better than $1/2+O(1/\sqrt n)$ and $(2k - 1)/2k + O(1/\sqrt{n})$ for general $k$.
\end{itemize}
At their core, the algorithms are variants of the classic greedy algorithm.   In the presence of noise, greedy fails since it cannot identify the set whose value is maximal in each iteration.  To handle noise, we apply a natural approach we call \emph{smoothing}.  In general, by selecting a family of sets $\mathcal{H}$ we can define a surrogate function $F(S)=\sum_{H'\in \mathcal{H}}f(S \cup H')$ and its noisy analogue $\widetilde{F}(S)=\sum_{H'\in \mathcal{H}}\widetilde{f}(S \cup H')$ which we can evaluate.  Intuitively, when $\mathcal{H}$ is sufficiently large and chosen appropriately, submodularity and monotonicity can be used to argue that $\widetilde{F}(S) \approx F(S)$.  Thus, smoothing essentially makes the noise disappear and instead leaves us to deal with the implications of optimizing with the surrogate $F$ rather than $f$.  In that sense, a large part of the challenge is in using optimization over the surrogate $F$ to approximate the optimum over $f$, i.e.:
\begin{itemize}
\item \textbf{Large $k$.}
In this regime, we first define $\textsc{Smooth-Greedy}$ which takes an arbitrary set $H$ of size $\log\log n$ and runs the greedy algorithm with the surrogate $\widetilde{F}= \sum_{H'\subseteq H}\widetilde{f}(T\cup H')$ on $N\setminus H$.  In the analysis we show that its output together with $H$ is arbitrarily close to $1-1/e$ of the optimal solution evaluated on $f_{H}$ (not $f$).  The \textsc{Slick-Greedy} algorithm runs multiple instantiations of a slightly modified version of \textsc{Smooth-Greedy} with different smoothing sets, and obtains a guarantee arbitrarily close to $1-1/e$ of the true optimum;

\item \textbf{Small $k$.}  In this regime, we use a modified version of greedy which adds a bundle of $O(1/\epsilon)$ elements in each iteration.  For each such bundle $B$ we define a surrogate $\widetilde{F}$ with a smoothing neighborhood of elements which are at distance $2$ on the $\{0,1\}^n$ hypercube from $B$.  In each iteration $\textsc{SM-Greedy}$ identifies the bundle $A$ which maximizes $\widetilde{F}$, but doesn't take it.  Taking a random bundle $\hat{A}$ from the smoothing neighborhood of $A$ gives the $1-1/e$ guarantee but \emph{in expectation}.  To obtain the result w.h.p. $\textsc{SM-Greedy}$ takes the bundle $\hat{A}$ which maximizes $\widetilde{f}(B)$, over all bundles $B$ in the smoothing neighborhood of $A$.  The analysis is then quite technical and strongly leverages the properties of the noise distribution and that $k \in O(\log\log n)$.  It is for this reason it is crucial that $\textsc{Slick-Greedy}$ applies to $k \in \Omega(\log\log n)$;

\item \textbf{Very small $k$.} In this case we consider bundles of size $k$ and smoothing with singletons.
\end{itemize}

\subsection{Extensions}
One of the appealing aspects of the noise model and the algorithms, is that they can easily be extended to a rich variety of related models.  In Section~\ref{sec:extensions} we discuss application to additive noise, marginal noise, correlated noise, information degradation, and approximate submodularity, .

\subsection{Applications}\label{sec:applications}

\begin{itemize}

\item \textbf{Optimization under noise.}  When considering optimization under noise, queries can be independent or correlated in \emph{time} and in \emph{space}.  For $f:2^N \to \mathbb{R}$ the noisy oracle is defined as $\widetilde{f}(S) = \xi_{S}(t) f(S)$ where $\xi_{S}(t) \sim \mathcal{D}$, for every step the oracle is queried $t \in \mathbb{N}$ and $S \subseteq N$.
\begin{definition*}
Noise is \textbf{i.i.d in time} if $\xi_{S}(t)$ and $\xi_{S}(t')$ are independent for any $t\neq t' \in \mathbb{N}$ and $S\subseteq N$.  Similarly, we can say that noise is \textbf{i.i.d in in space} if $\xi_{S}(t)$ and $\xi_{T}(t')$ for any $S\neq T$ and $t,t' \in \mathbb{N}$.  The noise distribution is \textbf{correlated in time (space)} if it is not independent in time (space).
\end{definition*}
The case in which the oracle is inconsistent is one where the noise is i.i.d in time and in space.  From an algorithmic perspective this problem is largely solved, as discussed above.  From Theorem $\ref{thm:adversarial}$ we know that there is no poly-time approximation algorithm for the case in which the errors are arbitrarily correlated in time and in space, even when the support of the noise distribution is arbitrarily small.  The model we describe assumes the noise is \emph{arbitrarily} correlated in time, but i.i.d in space.  In Section~\ref{sec:extensions} we show how one can relax this assumption.  In particular, we show how to generalize the algorithms to obtain approximation ratios arbitrarily close to $1-1/e$ in a noise model where $\xi_{S}(t)$ and $\xi_{T}(t')$ are arbitrarily correlated in time and in space for any $t,t' \in \mathbb{N}$ and $S,T$ for which $|S\triangle T| \in O(\sqrt{k})$ when $k\in \Omega(\log\log n)$ and $|S\triangle T| \in O(1)$ when $k \in O(\log\log n)$.  To the best of our knowledge, this is the first step towards studying submodular optimization under any correlation. 

\item \textbf{Maximizing approximately submodular functions.}  There are cases where one may wish to optimize an \emph{approximately} submodular function.  Theorem~\ref{thm:adversarial} implies that being arbitrarily close to a submodular function is not sufficient.  In statistics and learning theory, to model the fact that data is generated by a function that is approximately in a class of well behaved functions, the function generating the data $\widetilde{f}$ is typically assumed to be a noisy version of a function $f$ from a well-behaved class of functions~\cite{ESL09,W2010,SB14}:
$$\widetilde{f}(\mathbf{x}) = f(\mathbf{x}) + \xi_{\mathbf{x}}, $$
where $\xi_{\mathbf{x}}$ is an i.i.d sample drawn from some distribution $\mathcal{D}$.  In regression problems for instance, one assumes that the data is generated by $\widetilde{f}(\mathbf{x}) = \mathbf{w}^\intercal \mathbf{x} +\xi_{\mathbf{x}}$.  This model captures the idea that some phenomena may not exactly behave in a linear manner, but can be approximated by such a model.  Making a good prediction then involves optimizing the noisy model.  This therefore seems like a natural model to study approximate submodularity, especially in light of Theorem~\ref{thm:adversarial}.  Notice that in this case we would be interested in the optimization problem: $\max_{S:|S|\leq k}\widetilde{f}(S)$.  In Section~\ref{sec:extensions} we describe a black-box reduction which allows one to use the algorithms described here to get optimal guarantees.

\item \textbf{Active learning.}  In \emph{active learning} one assumes a membership oracle that can be queried to obtain labeled data~\cite{Angluin1988}.  In noise-robust learning, the task is to get good approximations to the noise-free target $f$ when the examples are corrupted by some noise.  In this model the assumption is that noise is \emph{consistent and i.i.d}, exactly as in our model.  That is, we observe $\tilde{f}(\mathbf{x}) + \xi_{\mathbf{x}}$ where $\mathbf{x}$ is drawn i.i.d from $\mathcal{D}$ and multiple queries return the same answer (see e.g.~\cite{GKS90a,J94,ss-lesqrpl-95,JSS99,BF02,DF09}).  Our results apply to additive noise, and thus apply to active learning with noisy membership queries of submodular functions.  One example application of active learning where the function is submodular is experimental design~\cite{KSG08,KMGG07,HIM14}.

\item \textbf{Learning and sketching.}  In learning and sketching the goal is to generate a surrogate function which approximates the submodular function well (see e.g. ~\cite{GHIM09,DBH11,BCIW12-valuations,BDFKNR12-sketches,FKV13-submodular_trees,FV13-submodular_juntas,DLBS14_a,DLBS14_b,FK14-coverage,FV15-submodular_polynomials,Balcan15-AAMAS}).  Theorem~\ref{thm:adversarial} implies that a surrogate which approximates a submodular function arbitrarily well may be inapproximable.  Our main result shows that if when sets are sufficiently far the surrogate approximates the function via independent noise, then one can use the surrogate for optimization.  This can therefore be used as a stricter benchmark for learning and sketching which allows optimizing a function learned or sketched from data.
\end{itemize}

\subsection{Paper organization}
The main technical contribution of the paper is the algorithms for the three different regimes of $k$.  The exposition of the algorithms is contained in sections~\ref{sec:largek},~\ref{sec:smallk}, and~\ref{sec:very_smallk}, which can be read independently from each other.  For each algorithm, we suppress proofs and additional lemmas to the corresponding section in the appendix.  All the algorithms employ smoothing arguments which can be found in Appendix~\ref{sec:appendix_smoothing}.  The smoothing arguments are used as a black-box in the proofs of each algorithm, and are not required for reading the main exposition.
In Section~\ref{sec:extensions} we discuss extensions of the algorithms to related models.  In Section~\ref{sec:adversarial} we prove the result for adversarial noise.  Discussion about additional related work is in Section~\ref{sec:related}.

\newpage \section{Optimization for Large $k$}\label{sec:largek}
In this section we describe the \textsc{Slick-Greedy} algorithm whose approximation guarantee is arbitrarily close to $1-1/e$ for sufficiently large $k$.  The algorithm is deterministic and for any desired degree of accuracy $\epsilon>0$ can be applied when the cardinality constraint $k$ is in $ \Omega(\log\log n/\epsilon^2)$, or more specifically when $k\geq 3168\log\log n/\epsilon^2$.  We first describe and analyze the \textsc{Smooth-Greedy} algorithm.  This algorithm is then used as a subroutine by the \textsc{Slick-Greedy} algorithm.

\subsection{The Smooth Greedy Algorithm}
We begin by describing the smoothing technique used by \textsc{Smooth-Greedy}.  We select an \emph{arbitrary} set $H$ and for a given element $a$, the smoothing neighborhood is simply ${\mathcal{H} = \{ H' \subseteq H \ : \ H' \cup a \}}$.  Throughout the rest of this section we assume that $H$ is an arbitrary set of size $\ell$, where $\ell$ depends on $k$.  
 In the case where $k\geq 2400\log n$ we will use $\ell = 25\log n$, and when $k< 2400\log n$ we will use $\ell = 33\log\log n$~
\footnote{W.l.o.g. we assume that $k<n-25\log n$ as for sufficiently large $n$ this then implies that $k\geq (1-\epsilon)n$ and by submodularity optimizing with $k'=n-25\log n$ suffices to get the $1-1/e - \epsilon$ guarantee for any fixed $\epsilon>0$.}.
The precise choice for $\ell$ will become clear later in this section.  Intuitively, $\ell$ is on the one hand small enough so that we can afford to sacrifice $\ell$ elements for smoothing the noise, and on the other hand $\ell$ is large enough so that taking all its subsets gives us a large smoothing neighborhood which enables applying concentration bounds.
\begin{definition*}
For a set $S\subseteq N$ and some fixed set $H\subseteq N$ of size $\ell$, we use $H^{(1)},\ldots,H^{(t)}$ to denote all the subsets of $H$ and $k' = k -\ell$.
The \textbf{smooth value}, \textbf{noisy smooth value} and \textbf{smooth marginal contribution} are, respectively:
\begin{align*}
&(1)\ &   F(S\cup a)  			& := \ &\mathbb{E} \left [  f(S \cup (H^{(i)} \cup a)  \right ]  			& = \ & \frac{1}{t} \sum_{i=1}^{t}  f  \left ( S \cup (H^{(i)} \cup a) \right ) &;&&\\
&(2) \ &  \widetilde{F}(S\cup a)  	& := \ &\mathbb{E} \left [  \widetilde{f}(S \cup (H^{(i)} \cup a)  \right ]  	& = \ &\frac{1}{t} \sum_{i=1}^{t}  \widetilde{f}  \left (S\cup (H^{(i)}  \cup a) \right )&;&&\\
&(3) \ & F_{S}(a)   		& := \ & \mathbb{E} \left [ f_S((H^{(i)} \cup a   ) ) \right ] 				& = \ & \frac{1}{t} \sum_{i=1}^{t} f_S \left (H^{(i)} \cup a \right )&.\hspace{0.2in}&&
\end{align*}
\end{definition*}
\subsubsection{The algorithm}
The smooth greedy algorithm is a variant of the standard greedy algorithm which replaces the procedure of adding $\argmax_{a \in N}f(S \cup a)$ with its smooth analogue.  The algorithm receives a set of elements $H$ of size $\ell$, initializes $S=\emptyset$ and at every stage adds to $S$ the element $a\notin H$ for which the smooth noisy value $\widetilde{F}(S\cup a)$ is largest.  A formal description is added below.
\begin{algorithm}
\caption{\textsc{Smooth-Greedy}}
\label{alg:Smooth-Greedy}
\begin{algorithmic}[1]
\INPUT budget $k$, set ${H}$
	\STATE $S \leftarrow \emptyset$
\WHILE {$|S| <k - |{H}|$}
	\STATE $ S \leftarrow S \cup \arg\max_{a\notin H} \widetilde{F}(S \cup a)$
\ENDWHILE
\RETURN $S$
\end{algorithmic}
\end{algorithm}
\paragraph{Overview of the analysis.}
At a high level, the idea behind the analysis is to compare the performance of the solution returned by the algorithm against an optimal solution which ignores the value of $H$ and any of its partial substitutes. 
More specifically, let \texttt{OPT} denote the value of the optimal solution with $k$ elements evaluated on $f$ and $\texttt{OPT}_{H}$ denote the value of the optimal solution with $k' = k-\ell$ elements evaluated on $f_{H}$, where $f_{H}(T) = f(T\cup H) - f(H)$.  Essentially, we will show that at every step \textsc{Smooth-Greedy} selects an element whose marginal contribution is larger than that of an element from the optimal solution evaluated on $f_{H}$ (we illustrate this idea in Figure~\ref{fig:combinatorial}).  Together with an inductive argument this suffices for a constant factor approximation.

\paragraph{Relevant iterations.}  One of the artifacts of noise is that our comparisons are not precise.  Specifically, when we select an element that maximizes $\widetilde{F}(S\cup a)$, our smoothing guarantee will be that this element respects ${F_S(a) \geq (1-\delta)\max_{b\notin H}F_S(b)}$ for $\delta>0$ that depends on $\epsilon$ and $k$.  This can be guaranteed only for an iteration where two conditions are met: (i) there is at least a single element not yet selected (and not in $H$) whose marginal contribution is at least $\epsilon/k$ fraction of $\texttt{OPT}_H$, and (ii) $\texttt{OPT}_{H}$ is sufficiently large in comparison to $\texttt{OPT}$.  We call such iterations \emph{$\epsilon$-relevant}.  

\begin{definition*}
For a given iteration of {\textsc{Smooth-Greedy}} let $S$ be the set of elements selected in previous iterations.  The iteration is ${\textbf{$\epsilon$-relevant}}$ if (i) $\max_{b \notin H}f_{H \cup S}(b) \geq \frac{\epsilon \cdot \texttt{OPT}_{H}}{k}$ and (ii) $\texttt{OPT}_H \geq \frac{\texttt{OPT}}{e}$.
\end{definition*}
We will analyze \textsc{Smooth-Greedy} in the case where the iterations are $\epsilon$-relevant as it allows applying the smoothing arguments.  In the analysis we will then ignore iterations that are not $\epsilon$-relevant at the expense of a negligible loss in the approximation guarantee.  The main steps are:
\begin{enumerate}
\item In Lemma~\ref{lem:boundona} we show that in each $\epsilon$-relevant iteration the (non-noisy) smooth marginal contribution of the element selected in that iteration by the algorithm is w.h.p. an arbitrarily good approximation to $\max_{b \notin H}F_{S}(b)$.  To do so we need claims~\ref{lem:black_and_white},~\ref{clm:boundedvariance} and~\ref{claim:realslimshady};
\item Next, in Claim~\ref{lem:bound} we show that the element $a$ whose smooth marginal contribution $F_{S}(a)$ is maximal has true marginal contribution $f_{S}(a)$ that is roughly a $k'${th} fraction of the marginal contribution of the optimal solution over $f_{H}$;
\item Finally, in Lemma~\ref{lem:opt_h} we apply a standard inductive argument to show that the fact that the algorithm selects an element with large smooth value in each step results in an approximation arbitrarily close to $1-1/e$ to $\texttt{OPT}_{H}$ (not \texttt{OPT}).  In Corollary~\ref{cor:constant} we show that the bound against $\texttt{OPT}_{H}$ can already be used to give a constant factor approximation to $\texttt{OPT}$.  To get arbitrarily close to ${1-1/e}$, \textsc{Slick-Greedy} executes multiple instantiations of a generalization of \textsc{Smooth-Greedy} as later described in Section~\ref{sec:sl}.
\end{enumerate}

\begin{figure*}[t]
\begin{centering}
                \includegraphics[trim = 0mm 0mm 0mm 0mm, height=40mm]{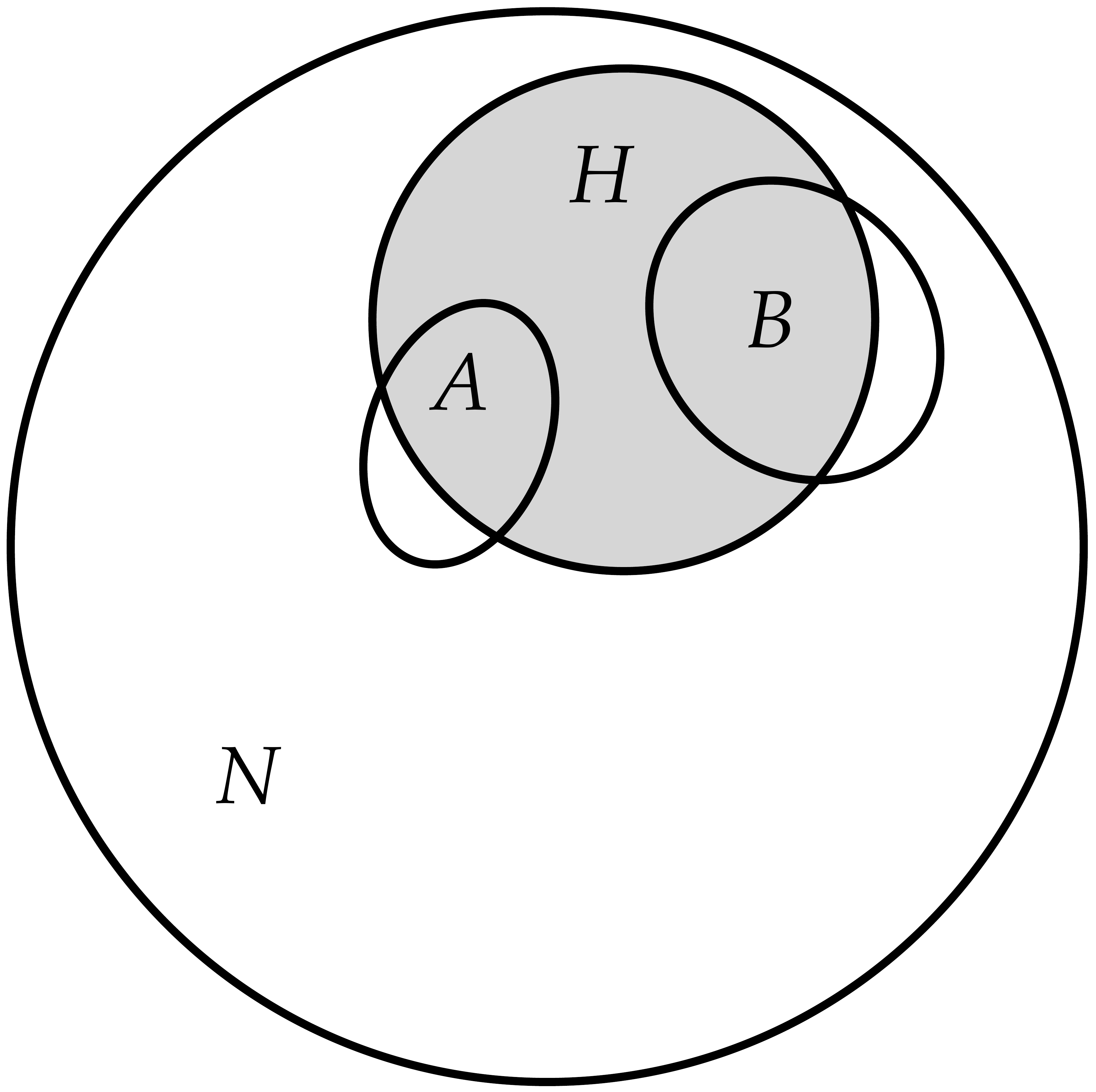}
                 \caption{\footnotesize{An illustration of Claim~\ref{lem:black_and_white} applied on a coverage function.  The set of all elements $N$ and $A,B,H \subset N$ are depicted as circles that illustrate the area of the universe they cover.  Claim~\ref{lem:black_and_white} essentially says that if we select $A$ rather than $B$ this means that the total area $A$ covers (white and grey) must be larger than the white-only (i.e. universe not covered by $H$) of $B$.  Stated in these terms, we use this idea to analyze the performance of \textsc{Smooth-Greedy} evaluated on the white and grey area against the optimal solution evaluated on the white-only area.}}\label{fig:combinatorial}
                 \end{centering}
\end{figure*}

\subsubsection{Smoothing guarantees}
The first step is to prove Lemma \ref{lem:boundona}.  This lemma shows that at every step as \textsc{Smooth-Greedy} adds the element that maximizes the noisy value $\argmax_{a\notin H} \widetilde{F}(S\cup a)$, that element nearly maximizes the (non-noisy) smooth marginal contribution $F_{S}$, with high probability.

\begin{lemma}\label{lem:boundona}
For any fixed $\epsilon >0$, consider an $\epsilon$-relevant iteration of \textsc{Smooth-Greedy} where $S$ is the set of elements selected in previous iterations and ${a \in \arg\max_{b \notin H}\widetilde{F}(S\cup b)}$.  Then for $\delta=\epsilon^2/4k$ and sufficiently large $n$ we have that w.p. $\geq 1-1/n^4$:
$$F_S(a) \geq (1-\delta) \max_{b \notin H}F_{S}(b).$$
\end{lemma}

To prove the above lemma we use claims~\ref{lem:black_and_white},~\ref{clm:boundedvariance}, and~\ref{claim:realslimshady}.  The statements and proofs can be found in Appendix~\ref{sec:appendix_largek} and are best understood after reading the smoothing section in Appendix~\ref{sec:appendix_smoothing}.

\subsubsection{Approximation guarantee}
Lemma~\ref{lem:boundona} lets us forget about noise, at least for the remainder of the analysis of $\textsc{Smooth-Greedy}$.  We can now focus on the consequences of selecting an element $a$ which (up to factor $1-\delta$) maximizes $F_{S}$ rather than the true marginal contribution $f_S$.

\begin{claim}\label{claim:newboundona}
For any $\epsilon>0$, let $\delta\leq \epsilon^2/4k$.  Suppose that the iteration is $\epsilon$-relevant and let ${b^\star \in \argmax_{b\notin H } f_{H \cup S}(b)}$.  If $F_{S}(a) \geq (1-\delta)F_S(b^\star)$, then:
$$f_S(a) \geq (1-\epsilon)f_{H\cup S}(b^\star).$$
\end{claim}
The principle is similar to Claim~\ref{lem:black_and_white}.  In this version we have a weaker condition since $F_{S}(a)$ is not greater than $F_{S}(b^\star)$ but rather $(1-\delta)F_S(b^\star)$, but the claim is less general as it only needs to hold for $b^\star$.  We therefore use a slightly different approach to prove this claim (see Appendix~\ref{sec:appendix_largek}).

\begin{claim}\label{lem:bound}
For any fixed $\epsilon>0$, consider an $\epsilon$-relevant iteration of \textsc{Smooth-Greedy} with $S$ as the elements selected in previous iterations.  Let ${a \in \argmax_{b \notin H}\widetilde{F}(S\cup b)}$.  Then, w.p. $\geq 1-1/n^4$:
$$f_S(a)\geq \Big (1-\epsilon\Big)\left[\frac{1}{k'}\Big (\texttt{OPT}_H - f(S) \Big)\right].$$
\end{claim}

The proof is in Appendix~\ref{lem:bound_proof}.  We can now state the main lemma of this subsection.

\begin{lemma}\label{lem:opt_h} 
Let $S$ be the set returned by \textsc{Smooth-Greedy} and $H$ its smoothing set.  Then, for any fixed $\epsilon>0$ when $k\geq 3\ell/\epsilon$ with probability of at least $1-1/n^3$ we have that:
$$f(S \cup H) \geq \left ( 1-1/e -\epsilon/3 \right) \texttt{OPT}_{H}.$$
\end{lemma}

To prove the lemma we show that if $\texttt{OPT}_{H}<\texttt{OPT}/e$ then $H$ alone provides the approximation guarantee.  Otherwise we can apply Claim~\ref{lem:bound} using a standard inductive argument to show that $S \cup H$ provides the approximation.  The subtle yet crucial aspect of the proof is that the inductive argument is applied to analyze the quality of the solution against the optimal solution for $f_H$ and not against the optimal solution on $f$.  The proof is in Appendix~\ref{lem:opt_h_proof}.

As we will soon see, Lemma~\ref{lem:opt_h} plays a key role in the analysis of the \textsc{Slick-Greedy} algorithm.  It is worth noting that this lemma can also be used to show that \textsc{Smooth-Greedy} alone provides a constant ($\approx 0.387$) albeit suboptimal approximation guarantee (Corollary \ref{cor:constant}).


\subsection{Slick Greedy: Optimal Approximation for Sufficiently Large $k$}\label{sec:sl}
  
The reason \textsc{Smooth-Greedy} cannot obtain an approximation arbitrarily close to $1-1/e$ is due to the fact that a substantial portion of the optimal solution's value may be attributed to $H$.  This would be resolved if we had a way to guarantee that the contribution of $H$ is small.  The idea behind \textsc{Slick-Greedy} is to obtain this type of guarantee.  Intuitively, by running a large albeit constant number of instances of \textsc{Smooth-Greedy} with different smoothing sets, selecting the ``best'' solution will ensure the contribution of the smoothing set is relatively minor.  

\subsubsection{The algorithm}
We can now describe the \textsc{Slick-Greedy} algorithm which is the main result of this section.  Given a constant $\epsilon>0$ we set $\delta = \epsilon/6$ and generate arbitrary sets ${H}_{1},\ldots,{H}_{1/\delta}$, each of size $\ell$ s.t. ${H}_i \cap {H}_j =\emptyset$ for every $i,j \in [1/\delta]$.  We then run a modified version of \textsc{Smooth-Greedy} $1/\delta$ times: in each iteration $j$ we initialize \textsc{Smooth-Greedy} with $R_j = \cup_{i\neq j}H_i$ 
\footnote{By initializing the \textsc{Smooth-Greedy} with $R_j$ we mean that the first iteration begins with $S=R_j$ rather than $S=\emptyset$ and following the initialization the algorithm greedily adds $k-|R_j|-|H_j|$ elements.}
and use ${H}_j$ to generate the smoothing neighborhood.  We denote this as $\textsc{Smooth-Greedy}(k,R_j, {H}_j)$.  We then compare the solution $T_j = S_j \cup {H}_j$ to the best $T_i=S_i \cup H_i$ we've seen so far using a procedure we call \textsc{Smooth-Compare} described below.  The \textsc{Smooth-Compare} procedure compares $T_i$ and $T_j$ by using a set ${H}_{ij}$ s.t. ${H}_{ij} \cap (T_j \cup T_i) = \emptyset$ and $|{H}_{ij}| = \ell$.  If $T_i$ wins, the procedure returns $T_i$ and otherwise returns $T_j$.  The \textsc{Slick-Greedy} then returns the set $T_i$ that survived the \textsc{Smooth-Compare} tournament.

\begin{algorithm}
\caption{\textsc{Slick-Greedy}}
\label{alg:Slick-Greedy}
\begin{algorithmic}[1]
\INPUT budget $k$
	\STATE Select $\ell/\delta$ elements in $N$ and partition them into disjoint
	sets of equal size ${H}_{1}\ldots,{H}_{1/\delta}$
	\STATE $T_i\leftarrow \emptyset$
\FOR{$j \in [1/\delta]$}
	\STATE $R_j \leftarrow \cup_{i\neq j}H_i$
	\STATE $ T_j  \leftarrow \textsc{Smooth-Greedy}(k,R_j, {H}_j) \cup {H}_j$
	\STATE ${H_{ij}}  \leftarrow$ arbitrary set of $\ell$ elements disjoint from $T_i\cup T_j$
	\STATE $T_i \leftarrow \textsc{Smooth-Compare}(\{T_i,T_{j}\},H_{ij})$
\ENDFOR
\RETURN $T_i$
\end{algorithmic}
\end{algorithm}

\paragraph{Overview of the analysis.}  Consider the smoothing sets $H_1,\ldots, H_{1/\delta}$.
Let $H_l$ be the smoothing set whose marginal contribution to the others is minimal, i.e. $H_l \in \argmin_{i\in [1/\delta]}f_{R_i}(H_i)$.
Notice that from submodularity we are guaranteed that $f_{R_l}( H_l) \leq \delta f(R_l \cup H_l )$.  In this case, the fact that the marginal contribution of $H_l$ to the rest of the smoothing sets $R_l$ is small, together with the fact that the solution is initialized with $R_l$, enables the tight analysis.  The two main steps are:
\begin{enumerate}
\item In Lemma \ref{claim:fr} we show that w.h.p. $T_l$ provides an approximation arbitrarily close to $(1 - 1/e)$.  Intuitively, this happens since the marginal contribution of $H_l$ to the rest of the smoothing sets $R_l = \cup_i H_i \setminus H_l$ is small, and since the solution to \textsc{Smooth-Greedy} is initialized with $R_l$, losing the value of $H_l$ is negligible.  The proof relies on Claim~\ref{claim:itslate} and Lemma~\ref{lem:opt_hr} that generalize the guarantees of \textsc{Smooth-Greedy} to the case it is initialized (see Appendix); 
\item We then describe and analyze the \textsc{Smooth-Compare} procedure.  In the absence of noise, one can simply select the set whose value is largest.  To overcome noise, we run a tournament to extract the solution whose value is approximately largest, or at least arbitrarily close to $(1-1/e)\texttt{OPT}$.  Specifically, we prove that w.h.p. the set $T_i$ that wins the \textsc{Smooth-Compare} tournament (i.e. the set $T_i$ returned by \textsc{Slick-Greedy}) satisfies ${f(T_i) \ge (1 - \epsilon/3) \min \{f(T_l)},(1-1/e - 2\epsilon/3)\texttt{OPT}\}$.  Since $f(T_l)$ is arbitrarily close to ${(1 - 1/e)\texttt{OPT}}$, this concludes the proof.
\end{enumerate}

\subsubsection{Generalizing guarantees of smooth greedy}

\begin{lemma}\label{claim:fr}
Let $S_l$ be the set returned by \textsc{Smooth-Greedy} that is initialized with $R_l$ and $H_l$ its smoothing set.  Then, for any fixed $\epsilon>0$ when $k\geq 36\ell/\epsilon^2$ w.p. at least $1-1/n^3$ we have that:
$$ f(S_l \cup H_l)  \geq (1-1/e -2\epsilon/3) \texttt{OPT}.$$
\end{lemma}

\subsubsection{The smooth comparison procedure}
We can now describe the \textsc{Smooth-Compare} procedure we use in the algorithm.  For a given set $H_{ij} \subseteq N$ of size $\ell$ and two sets $T_i,T_j \subseteq N \setminus {H_{ij}}$, we compare $\widetilde{f}(T_i \cup H_{ij}')$ with $\widetilde{f}(T_j \cup H'_{ij})$ for all $H'_{ij} \subset H_{ij}$.  We select $T_i$ if in the majority of the comparisons with $H'_{ij} \subset H_{ij}$ (breaking ties lexicographically) we have that $\widetilde{f}(T_i\cup H'_{ij}) \geq \widetilde{f}(T_j\cup H'_{ij})$, and otherwise we select $T_{j}$.  

\begin{algorithm}
\caption{\textsc{Smooth-Compare}}
\label{alg:Smooth-Compare}
\begin{algorithmic}[1]
\INPUT $T_i, T_j, H_{ij} \subseteq N \setminus (T_i \cup T_j)$,
	\STATE Compare $\widetilde{f}(T_i \cup H'_{ij})$ with $\widetilde{f}(T_j \cup H'_{ij})$ for all $H'_{ij} \subset H_{ij}$	
	\STATE if $T_i$ won the majority of comparisons return $T_i$ otherwise return $T_j$
\end{algorithmic}
\end{algorithm}

\begin{lemma}\label{lem:boosting}
Assume $k \ge 96 \ell /\epsilon^2$.  Let $T_i$ be the set that won the $\textsc{Smooth-Compare}$ tournament.  Then, with probability at least $1-1/n^2$:
$$
f(T_i)\geq \left (1-\frac{\epsilon}{3} \right )\min \left \{ \left(1-\frac{1}{e}-\frac{2\epsilon}{3}\right)\texttt{OPT},\max_{j \in [1/\delta]}f(T_j)\right \}
$$
\end{lemma}
The proof of this lemma has two parts.
\begin{enumerate}
\item
First we show in Claim~\ref{clm:slick} that if a set $T_i$ has moderately larger value than another set $T_j$ (more specifically, if the gap is $1-\epsilon\delta/3$) then as long as $f(T_j)$ is not arbitrarily close to $(1-1/e)\texttt{OPT}$ then $f(T_i\cup H'_{ij})$ is larger than $f(T_j \cup H'_{ij})$, for any $H'_{ij} \subseteq H_{ij}$.  At a high level, this is because elements in $H'_{ij}$ are candidates for \textsc{Smooth-Greedy} and the fact that they are not selected indicates that their marginal contribution to $T_j = S_j\cup H_j$ is low.  Thus, elements in $H'_{ij}$ cannot add much value, and since $|H_{ij}|\ll k$ adding subsets of $H_{ij}$ does not distort the comparison by much.  If $f(T_j)$ is arbitrarily close to $(1-1/e)\texttt{OPT}$, we may have that $T_j$ beats $T_i$, but this would still ultimately result in an approximation arbitrarily close to $1-1/e$;  

\item The next step (Claim~\ref{lem:tournament}) then shows that if for every $H'_{ij}$ we have $f(T_i \cup H'_{ij}) \geq f(T_j \cup H'_{ij})$ then with high probability $T_i$ wins the comparison against $T_j$ in \textsc{Smooth-Compare}.
\end{enumerate}
Using these two parts we then conclude since we are running the \textsc{Smooth-Compare} tournament between $1/\delta$ sets, the winner is an $(1 - \epsilon \delta /3)^{1/\delta} \ge (1 - \epsilon/3)$ approximation to the competing set with the highest value or a set whose approximation is arbitrarily close to $1-1/e$.  The claims and proofs can be found in Appendix~\ref{sec:smooth_compare}.

\subsubsection{Approximation guarantee of \textsc{slick greedy}}

Finally, putting everything together, we can prove the main result of this section (see Appendix~\ref{thm:largek_main_proof}).

\begin{thm}\label{thm:largek_main}
Let $f:2^{N} \to \mathbb{R}$ be a monotone submodular function.  For any fixed $\epsilon>0$, when $k\geq 3168\log \log n/\epsilon^2$, then given access to a noisy oracle whose noise distribution has a generalized exponential tail, the $\textsc{Slick-Greedy}$ algorithm returns a set which is a $(1-1/e-\epsilon)$ approximation to $\max_{S:|S|\leq k}f(S)$, with probability at least $1-1/n$.
\end{thm}

\newpage \section{Optimization for Small $k$}\label{sec:smallk}
When $k$ is small we cannot use the smoothing technique from the previous section, since it requires including the smoothing set of size $\Theta(\log \log n)$ in the solution.  In this section we describe the \emph{sampled mean method} which can be applied to ${k \in \Omega(1/\epsilon) \cap O(\log \log n)}$ and results in a $1-1/e -\epsilon$ approximation.  This result is obtained by applying a greedy algorithm on a surrogate function $F:2^{N} \to \mathbb{R}_{+}$ which is what we call the \emph{sampled mean} of $f$.  The use of the surrogate function makes it relatively easy to obtain the $1-1/e-\epsilon$ approximation, albeit \emph{in expectation}.  The main technical challenge is the transition from a guarantee that holds in expectation to one that holds with high probability.  This difficulty is what limits this method to be applicable only when $k$ ranges between $\Omega(1/\epsilon)$ and $O(\log \log n)$, and heavily exploits the generalized exponential tail property.  

\subsection{Combinatorial averaging}
The sampled-mean method is based on averaging sets to find elements whose marginal contribution is high, which can then be greedily added to the solution.  The intuition for this method comes from continuous optimization.  Consider optimizing a function $f:\mathbb{R}^n \to \mathbb{R}$ given access to a noisy value oracle $\widetilde{f}:\mathbb{R}^n \to \mathbb{R}$ which for each point $\mathbf{x} \in \mathbb{R}^n$ returns $\widetilde{f}(\mathbf{x}) = \xi_{\mathbf{x}}f(\mathbf{x})$ where $\xi_\mathbf{x} \sim \mathcal{D}$.  A natural approach would be to sample $t$ points $\mathbf{x}_1,\ldots,\mathbf{x}_t$ from an $\epsilon$-ball $\mathcal{B}_{\epsilon}$ around $\mathbf{x}$, for some small $\epsilon>0$, and estimate the value of $\mathbf{x}$ using the sampled mean: 
$$\widetilde{F}(\mathbf{x}) := \mathbb{E} \left [ \widetilde{f}(\mathbf{x})\right ] = \frac{1}{t}\sum_{\mathbf{x}_i \sim \mathcal{B}_{\epsilon}}\widetilde{f}(\mathbf{x}_i)$$ 
Under some smoothness assumptions on $f$, for sufficiently large $t$ and small $\epsilon$, concentration bounds kick in, and one can apply an optimization algorithm on $\widetilde{F}$ to optimize $f$.  The method in this section translates this idea to a combinatorial domain.  To do so effectively, rather than considering singletons $a \in N$ we obtain multidimensionality by considering \emph{bundles} of size $c \in O(1/\epsilon)$.   

\begin{definition*}
Let $f:2^N \to \mathbb{R}$.  For a set $S\subseteq N$ and \textbf{bundle} $A\subseteq N$ of fixed size $c$, we define 
$A_{ij} := \left (A \setminus \{a_i\}\right ) \cup \{a_j\}$ for $a_i \in A$ and $a_j \notin S \cup A$, and $t = c(n-c-|S|)$.  
The \textbf{mean value}, \textbf{noisy mean value}, and \textbf{mean marginal contribution} of $A$ given $S$ are, respectively:
\begin{align*}
&(1)\ &   \score(S\cup A) 			& := \ &\mathbb{E} \left [  f(S \cup A_{ij})  \right ]  			& = \ &\frac{1}{t}&\sum_{i \in A}\sum_{j \notin S\cup A}  f(S \cup A_{ij});&&\\
&(2) \ &  \widetilde{\score}(S\cup A) 	& := \ &\mathbb{E} \left [  \widetilde{f}(S \cup A_{ij})  \right ]  	& = \ &\frac{1}{t}&\sum_{i \in A}\sum_{j \notin S\cup A}  \widetilde{f}(S \cup A_{ij});&&\\
&(3) \ & \score_{S}(A)			& := \ & \mathbb{E} \left [ f_S(A_{ij}) \right ] 				& = \ &\frac{1}{t}&\sum_{i \in A}\sum_{j \notin S\cup A}  f_S(A_{ij}).&&
\end{align*}
\end{definition*}

The above definition mimics the continuous case by considering a \emph{bundle} of elements $A$ of fixed size $c$ (we will use $c \approx 1/\epsilon$) as a point, and the points in the $\epsilon$-ball are modeled by all the sets $A_{ij}$ obtained by replacing an element from $A$ with an element from $N \setminus (S \cup A)$.  We illustrate this idea in Figure~\ref{fig:sm}.  Although the combinatorial analogue is not as well-behaved as the continuous case, the sampled mean approach defined here extracts some of its desirable properties.  

\subsection{The Sampled Mean Greedy Algorithm}\label{sec:sm-alg}
The \textsc{SM-Greedy} begins with the empty set $S$ and at every iteration considers all bundles of size $c \in O(1/\epsilon)$ to add to $S$.  At every iteration, the algorithm first identifies the bundle $A$ which maximizes the noisy mean value.  After identifying $A$, it then considers all possible bundles $A_{ij}$ and takes the one whose noisy mean value is largest.  We describe the algorithm formally below.

\begin{algorithm}
\caption{\textsc{SM-Greedy}}\label{alg:SM-Greedy}
\label{c}
\begin{algorithmic}[1]
\INPUT budget $k$, precision $\epsilon>0$, $c \in O(\frac{1}{\epsilon})$
	\STATE $S \leftarrow \emptyset$
\WHILE {$|S| < c \cdot \left \lfloor \frac{k}{c} \right \rfloor $}
	\STATE $A \leftarrow \argmax_{B : |B|=c}\widetilde{\score}(S\cup B)$~\label{alg:sm3}
	\STATE $ S \leftarrow S \cup \arg\max_{i\in A,j\notin S\cup A} \widetilde{f}(S \cup A_{ij})~\label{alg:sm4}$
\ENDWHILE
\RETURN $S$
\end{algorithmic}
\end{algorithm}

At a high level, the major steps in the analysis can be described as follows.
\begin{enumerate}
\item We begin with smoothing guarantees.  In Lemma~\ref{algoworks} we apply Lemma~\ref{lem:avg} as well as other arguments to show that w.h.p. in each iteration $A \in \argmax_{B:|B|=c}\widetilde{F}(S \cup B)$ well approximates the bundle with maximal (non-noisy) mean marginal contribution $\argmax_{B:|B|=c}F_{S}(B)$;
\item Lemma~\ref{lem:meanapx} argues that if the marginal contribution $f_{S}(\MS)$ of the set $\MS$ we select at every iteration is close to the mean marginal contribution $\score_S(A)$ we obtain an approximation arbitrarily close to $1-1/e$.  This suffices for an approximation guarantee that holds in expectation;
\item The last step is Lemma~\ref{single} which is the technical crux of this section.  We show that taking $\MS \in \argmax_{i,j} \widetilde{f}(S\cup A_{ij})$ in line~\ref{alg:sm4} of the algorithm gives us, with sufficiently high probability that the marginal contribution $f_{S}(\MS)$ is arbitrarily close to the mean marginal contribution $F_{S}(A)$.  We can therefore invoke Lemma~\ref{lem:meanapx} and recover the optimal approximation guarantee.
\end{enumerate}
\subsection{Smoothing Guarantees}

We first show that the largest marginal contribution is well approximated by its mean contribution.  

\begin{lemma}\label{lem:avg}
For any $\epsilon>0$ and any set $S \subset N$, let $A^\star \in \arg\max_{A:|A|=1/\epsilon}f_{S}(A)$.  Then:
$$ \left( 1- \epsilon\right)f_{S}(A^\star)\leq \score_{S}(A^\star)\leq f_{S}(A^{\star}).$$
\end{lemma}

The proof is in Appendix~\ref{lem:avg_proof} and exploits a natural property of submodular functions: the removal of a random element from a large set does not significantly affect its value, in expectation.

\paragraph{Significant iterations.}  Similar to the previous section, we define an assumption on the iterations of the algorithm which allows us to employ the smoothing technique in this section.
\begin{definition*}
Let $B\in \argmax_{B:|B|= c}f_S(B)$.  An iteration of \textsc{SM-Greedy} is \textbf{$\epsilon$-significant} if for the given set $S$ selected before the iteration we have that $f_{S}(B) \geq \frac{\epsilon\cdot c \cdot \texttt{OPT}}{k}$.
\end{definition*}

 \begin{figure*}[t]
\begin{centering}
                \includegraphics[trim = 0mm 0mm 0mm 0mm, height=47mm]{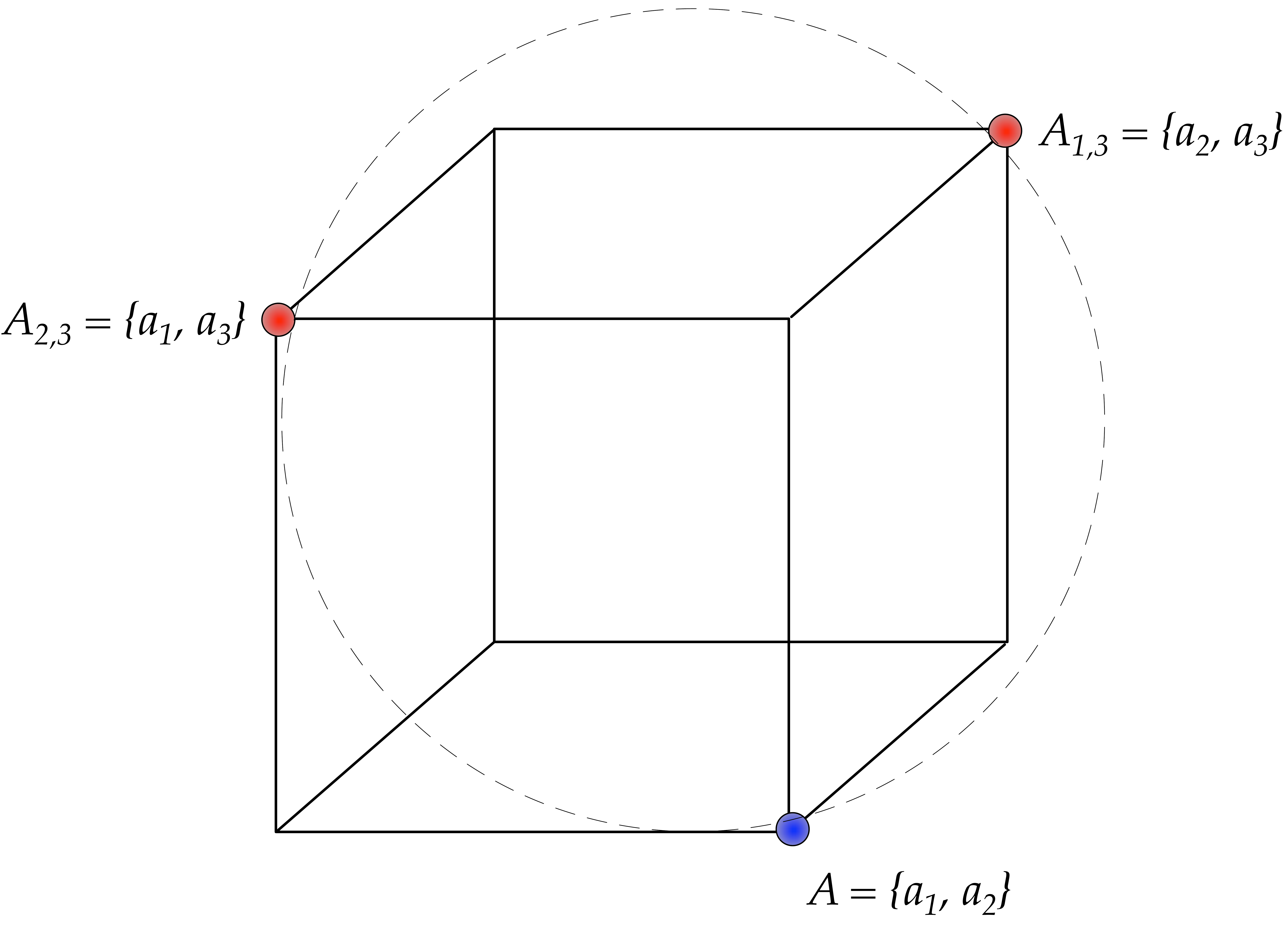}
                 \caption{\footnotesize{An illustration of the smoothing neighborhood.  In this example $N = \{a_1,a_2,a_3\}$, and the bundle we wish to evaluate is $A = \{a_1,a_2\}$.  We think of $A$ as a point in $\mathbb{R}^{3}$ and the smoothing neighborhood of $A = (1,1,0)$ is the points $A_{1,3} = \{a_2,a_3\} = (0,1,1)$ and $A_{2,3} = \{a_1,a_3\} = (1,0,1)$.  The circle illustrates the ball surrounding $A$.}}\label{fig:sm}
                 \end{centering}
\end{figure*}

The following lemma implies that at every step we add a bundle whose smooth marginal contribution is comparable with the largest smooth marginal contribution obtainable.

\begin{lemma}\label{algoworks}
Let $A \in \argmax_{B:|B|=c} \widetilde{\score}(S\cup B)$ where $c \geq \frac{16}{\epsilon}$, and assume that the iteration is $\frac{\epsilon}{4}$-significant.  Then, with probability at least $1-  e^{-\Omega(n^{{1}/{10}})} $ we have that:
$$\score_{S}(A) \geq (1-\epsilon) \max_{B:|B|=c}F_S(B).$$
\end{lemma}
The proof relies on arguments from the smoothing framework (Appendix~\ref{sec:appendix_smoothing}).  In this case, the application of smoothing is a bit subtle as we do not apply smoothing on the noisy version of $F$ directly.  The proof uses Lemma~\ref{lem:avg} above as well as Claim~\ref{cl:bound_on_var} which bounds the variation in values of sets $A^\star_{ij}$, when $A^\star \in \argmax_{B:|B|=c}f_{S}(B)$.  Details and proofs are in Appendix~\ref{sec:appendix_smallk}.  

\subsection{Approximation Guarantee in Expectation}

\begin{lemma}\label{lem:meanapx}
Let $\delta>0$ and assume $k> 16 /\delta^2$, $c = 16/\delta$.  Suppose that in every $\delta/4$-significant iteration of \textsc{SM-Greedy} when $S$ are the elements selected in previous iterations, $A\in \argmax_{B:|B|=c} \widetilde{\score}(S \cup B)$, the bundle added $\hat{A}$ respects $f_{S}(\hat{A}) \geq (1-\delta)F_{S}(A)$.  Let $\bar{S}$ be the solution after $\lfloor k/c \rfloor$ iterations.  Then, w.p. $\geq 1-1/n^2$:
$$f(\bar{S}) = (1-1/e - 5\delta)\texttt{OPT}.$$
\end{lemma}

This lemma implicitly proves an approximation guarantee that holds \emph{in expectation}.  This is simply because we know that if we choose $\hat{A} = A\setminus \{a_i\}\cup \{a_j\}$ uniformly at random over all choices of $i\in[c],a_j\notin S\cup A$ we get $\mathbb{E} [f_{S}(\hat{A})] = F_{S}(A) > (1-\delta)F_{S}(A)$ in every iteration, and thus by Lemma~\ref{lem:meanapx} we would be arbitrarily close to $1-1/e$, in expectation over all our choices. 

\subsection{From Expectation to High Probability}
From Lemma~\ref{algoworks} we know that $A \in \argmax_{B:|B|=c}\widetilde{F}(S \cup B)$ has mean marginal contribution arbitrarily close to $\max_{B:|B|=c}F_{S}(B)$, but for Lemma \ref{lem:meanapx} to hold we need the true marginal contribution $f_{S}(\MS)$ to be arbitrarily close to $\max_{B:|B|=c}F_{S}(B)$.  Simply adding $A$ can easily lead to an arbitrarily bad approximation (see Appendix~\ref{sec:examples} ).  In order to prove that $\textsc{SM-Greedy}$ provides the desired approximation guarantee, we need to show that when $\hat{A} \in \argmax_{i\in[c],j\notin S\cup A}\widetilde{f}(S \cup A_{ij})$ then with sufficiently high probability $f_{S}(\hat{A})$ is arbitrarily close to $\score_{S}(A)$ as required by Lemma~\ref{lem:meanapx}.

\paragraph{High-level overview to show high probability guarantee.}  Let $A^\star \in \argmax_{B:|B|=c}f_{S}(B)$ and $A\in \argmax_{B:|B|=c} \widetilde{\score}(S\cup B)$.  We will define two kinds of sets in $\{A_{ij}\}_{i\in [c],j \notin S\cup A}$, called \textbf{good} and \textbf{bad}.  A good set is a set $G$  for which $f_{S}(G) \geq (1-2\epsilon) f_{S}(A^\star)$ and a bad set is a set $B$ for which $f_{S}(B)\leq (1-3\epsilon)f_{S}(A^\star)$.  Our goal is to prove $\argmax \{ \widetilde{f}(S\cup A_{ij}) :  a_i \in A, a_j \notin S\cup A  \}$ is w.h.p. not bad.  Doing so implies that in every iteration w.h.p. we add a bundle whose true marginal value is at least $(1-3\epsilon)$ of $f_{S}(A^\star)$ which is an upper bound on $\max_{B:|B|=c} F_{S}(B)$ (and thus also on $F_{S}(A)$).

\begin{lemma}\label{single}
For any $\epsilon>0$, suppose we run \textsc{SM-Greedy} where in each iteration we add a bundle of size $c= 16/\epsilon$.  For any $\epsilon/8$-significant iteration where the set previously selected is ${S: |S| \in O(\log\log n)}$, let $A \in \argmax \widetilde{\score}(S\cup A)$ and $\MS = \argmax_{(i,j)\in A \times N\setminus S\cup A} \widetilde{f}(S \cup A_{ij})$.  
Then, w.p. $\geq 1-3/\log n$ we have:
$$f_S(\MS) \ge (1 - 3\epsilon)\score_S(A).$$
\end{lemma}

At a high level, the proof follows the following steps: 
\begin{enumerate}
\item In Claim~\ref{clm:goodandbad} we show that for $A\in \argmax_{B:|B|=c}\widetilde{\score}(S \cup B)$, at least half of the sets in $\{A_{ij}\}_{i\in A,j \notin S \cup A}$ are good, and at most half are bad;

\item Next, we define two thresholds: $\mg$ and $\mb$.  Intuitively, $\mg$ is a lower bound on the maximum of noise multipliers from the good sets, and $\mb$ is an upper bound on the maximum of noise multipliers from bad sets.  We then show in Lemma~\ref{relationMbMg} that $\mg \geq (1-\eps)\mb$, for any $\eps = \Omega(1/ \log \log n)$.  This lemma is quite technical, and it is where we fully leverage the property of the generalized exponential tail distribution and the fact that $k \in O(\log\log n)$;  

\item From $\mg \geq (1-\eps)\mb$ and Claim~\ref{clm:goodandbad} we can prove that  w.h.p. there is at least one good set whose noisy value is sufficiently larger than the noisy value of a bad set.  The fact that a bad set loses to a good set implies that the value of the set we end up selecting must at least be as high as that of a bad set, i.e. $f_{S}(\MS)\geq (1-3\epsilon)f_{S}(A^\star)$.  Notice that by definition $f_S(A^\star)$ is an upper bound on $F_S(B)$ for any bundle $B$ of size $c$ which therefore completes the proof.  
\end{enumerate}

Lemma~\ref{single} above essentially tells us that at every iteration we select the bundle whose marginal contribution is almost maximal.  Together with previous arguments from this section, this proves our main theorem for the case in which ${k  \in \Omega({1}/{\epsilon^2}) \cap O(\log \log n)}$.  For $k \in \Omega(\frac{1}{\epsilon}) \cap O(\frac{1}{\epsilon^2})$ we run a single iteration of \textsc{SM-Greedy} with $c= k$ (o.w. the approximation is $\approx 1/2$, when $k = 2c-1$).

\begin{theorem}\label{final-thm}
For any monotone submodular function $f:2^N \to \mathbb{R}$ and $\epsilon>0$, when ${k  \in \Omega({1}/{\epsilon}) \cap O(\log \log n)}$, there is a $(1 - 1/e - \epsilon)$ approximation for $\max_{S:|S|\leq k} f(S)$, with probability $1-4/\log n$ given access to a noisy oracle whose distribution has a generalized exponential tail.
\end{theorem}

\newpage \section{Optimization for Very Small $k$}\label{sec:very_smallk}
The smoothing guarantee from the previous section actually necessitates selecting bundles of size $c \in\Theta(1/\epsilon)$ and does not apply to very small values of $k\in O(1/\epsilon)$\footnote{The dependency on $\epsilon$ originates in Claim~\ref{cl:bound_on_var} where we bound on the variation of $c-1$ sets $\Ai$, and thus smoothing depends on $c \geq 4/\epsilon$.}.  For small constants we propose a different algorithm that uses a different smoothing technique.  The algorithm is simple and applies the same principles as the ones from the previous section.   
We show that this simple algorithm obtains an approximation ratio arbitrarily close to $1-1/e$ w.h.p. when $k>2$ and in expectation when $k=2$.  For $k=1$ we get arbitrarily close to $1/2$, which is tight.  We show lower bounds for small values of $k$ and in particular when $k=1$ show that no algorithm can obtain an expected approximation ratio better than $1/2 +o(1)$.  All proofs and details are in Appendix~\ref{sec:very_smallk_appendix}.

\subsection{Smoothing Guarantees}
The smoothing here is straightforward.  For every set $A$ consider the smoothing neighborhood $\mathcal{H}(A) = \{A \cup x \ : \ x \notin A\}$, $F(A) = \mathbb{E}_{X \in \mathcal{H}(A)}[f(X)]$ and $\widetilde{F}(A) = \mathbb{E}_{X \in \mathcal{H}(A)}[\widetilde{f}(X)]$.    

\begin{lemma}\label{smoothing_verysmall}
Let $A \in \argmax_{B:|B|=k}\widetilde{F}(B)$.  Then, for any fixed $\epsilon>0$ w.p. $1-e^{-\Omega(\epsilon^2(n-k))}$:
$$F(A) \geq (1-\epsilon)\max_{B:|B|=k} F(B).$$
\end{lemma}

\subsection{An Approximation Algorithm for Very Small $k$}

\paragraph{Approximation guarantee in expectation.}  The algorithm will simply select the set $\MS$ to be a random set of $k$ elements from a random set of $\mathcal{H}(A)$ where $A \in \argmax_{B:|B|=k} \widetilde{F}(B)$.  For any constant $k$ and any fixed $\epsilon>0$ this is a $\left(k/(k+1)-\epsilon\right)$ approximation \emph{in expectation} (see Theorem~\ref{thm:tinyk}).

\paragraph{High probability.} To obtain a result that holds w.h.p. we will consider a modest variant of the algorithm above.  The algorithm enumerates all possible subsets of size $k-1$, and identifies the set $A\in \argmax_{B:|B|=k-1}\widetilde{F}(B)$.  The algorithm then returns $\MS \in \argmax_{X \in \mathcal{H}(A)} \widetilde{f}(X)$. 

\begin{theorem}\label{thm:verysmallkwhp}
For any submodular function $f:2^N \to \mathbb{R}$ and any fixed $\epsilon>0$ and constant $k$, there is a $\left(1-1/k-\epsilon\right)$-approximation algorithm for $\max_{S:|S|\leq k}f(S)$ which only uses a generalized exponential tail noisy oracle, and succeeds with probability at least $1-6/\log n$.
\end{theorem}

\subsection{Information Theoretic Lower Bounds for Constant $k$}
Surprisingly, even for $k=1$ no algorithm can obtain an approximation better than $1/2$, which proves a separation between large and small $k$.  In Claim~\ref{clm:lower_constant} we show no randomized algorithm with a noisy oracle can obtain an approximation better than $1/2 +O(1/\sqrt{n})$ for $\max_{a \in N} f(a)$, and in Claim~\ref{clm:evenlower_constant} approximation better than $(2k-1)/2k +O(1/\sqrt{n})$ for the optimal set of size $k$.

\newpage \section{Extensions}\label{sec:extensions}
In this section we consider extensions of the optimization under noise model.  In particular, we show that the algorithms can be applied to several related problems: additive noise, marginal noise, correlated noise, degradation of information, and approximate submodularity.

\subsection{Additive Noise}
Throughout this paper we assumed the noise is multiplicative, i.e. we defined the noisy oracle to return $\widetilde{f}(S) = \xi_{S}\cdot f(S)$.  An alternative model is one where the noise is \emph{additive}, i.e. $\widetilde{f}(S) = f(S)+\xi_{S}$, where $\xi_{S} \sim \mathcal{D}$. The impossibility results for adversarial noise apply to the additive case as well.

From a modeling perspective, the fact that the noise may be independent of the value of the set queried may be an advantage or a disadvantage, depending on the setting.  From a technical perspective, the problem remains non-trivial.  Fortunately, all the algorithms described above apply to the additive noise model, modulo the smoothing arguments which become straightforward.  That is, we still need to apply smoothing on the surrogate functions, but it is easy to show arguments like $A \in \argmax_{B} \widetilde{F}(S\cup B)$ implies w.h.p. $F_{S}(A) \geq (1-\delta)\max_{b}F_{S}(B)$.  In the additive noise model:
$$\widetilde{F}(S\cup A) = \sum_{X \in \mathcal{H}(A)} \widetilde{f}(S\cup X) = \sum_{X \in \mathcal{H}(A)} \left ( f(S\cup X) + \xi_{S\cup X} \right )= \sum_{X \in \mathcal{H}(A)} f(S\cup X) + \sum_{X \in \mathcal{H}(X) }\xi_{S\cup X}$$
Thus, by applying a concentration bound we can show that a set $A$ whose smooth value is maximal implies that its non-noisy smooth marginal contribution $F_{S}(A)$ is approximately maximal as well.

\subsection{Marginal Noise}
An alternative noise model is one where the noise acts on the marginals of the distribution. In this model, a query to the oracle is a pair of sets $S,T \subseteq N$ and the oracle returns $\xi_{S,T} \cdot f_S(T)$ in the \emph{multiplicative marginal noise} model and $f_S(T) + \xi_{S,T}$ in the \emph{additive marginal noise} model.

\paragraph{Adversarial additive marginal noise is generally impossible.}  If the error is adversarial, and the noise is additive, the lower bound of \ref{thm:adversarial} follows for any magnitude of the noise. Letting $\epsilon$ denote the maximal magnitude of the noise, we consider a function in which no element ever gives a contribution higher than $\epsilon$, and then getting marginal information does not help.

\paragraph{Adversarial multiplicative marginal noise is approximable.}  If the marginal error is adversarial but multiplicative within factor $\alpha$, it is well known one can obtain a $1 - 1/e^{\alpha}$ approximation.

\paragraph{Marginal i.i.d noise is approximable.}  If one is allowed to query the oracle on any two sets $S,T$ and get $\xi_{S,T}\cdot f_S(T)$ (or $f_S(T) + \xi_{S,T}$) where $\xi_{S,T}$ is drawn i.i.d for any pair $S,T$, then one can simply apply all the algorithms and analysis as is, by always considering $f_{\emptyset}(S \cup T)$.  If one is only allowed to query $S,T$ where $|T|=1$, the algorithms still work, but we need to be careful with the analysis, since we need to show that we are  calling the oracle on different sets.  It is easy to show that if the noise is weak and multiplicative (e.g. $\xi \in [1 - \epsilon, 1 + \epsilon]$) we can obtain a $( 1 - 1/e -\epsilon)$ approximation.

\subsection{Correlated Noise}
As discussed in the Introduction, Theorem~\ref{thm:adversarial} implies that no algorithm can optimize a monotone submodular function under a cardinality constraint given access to a noisy oracle whose noise multipliers are arbitrarily correlated across sets, even when the support of the distribution is arbitrarily small.  In light of this, one may wish to consider special cases of correlated distributions.  We first show that even very simple correlations can result in inapproxiability.  We then show an interesting class of distributions we call \emph{$d$-correlated}, for which optimal guarantees are obtainable.

\paragraph{Impossibility result for correlated distributions.} Having taken the first step showing algorithms for the i.i.d. in space model, a natural question is whether this assumption is necessary.

\begin{theorem}
Even for unit demand functions there are simple space-correlated distributions for which no algorithm can achieve an approximation strictly better than $1/n$.
\end{theorem}

\begin{proof}
Consider a unit demand function $f(S)=\max_{a\in S}f(a)$ which operates on a ground set with $n$ elements. There are $n-1$ \emph{regular} elements and one \emph{special} element $a^\star$. The value of $f$ on any regular element is $1$, but $f(a^\star) = M$ for some arbitrarily large $M$. The noise distribution is such that it returns $1$ on sets which do not contain $a^\star$, and $1/M$ on sets that contain $a^\star$. The best one can do in this case is to choose a random element without querying the oracle at all.
\end{proof}

\paragraph{Guarantees for $d$-correlated distributions.}
Our algorithms can be extended to a model in which querying similar sets may return results that are arbitrarily correlated, as long as querying sets which are sufficiently far from each other gives independent answers. 

\begin{definition*}
We say that the noise distribution is \textbf{$d$-correlated} if for any two sets $S$ and $T$, such that $|S \setminus T| + |T \setminus S| > d$ we have that the noise is applied independently to $S$ and to $T$.
\end{definition*}

Notice that if a distribution is $d$-correlated, any two points on the hypercube at distance at most $d$ can be arbitrarily correlated.  For this model we show that when $k \in \Omega(\log\log n)$ then we can obtain an approximation arbitrarily close to $1-1/e$ for $O(\sqrt{k})$-correlated distributions.  Alternatively, in this regime we can get this approximation guarantee for any distribution that is arbitrarily correlated when querying two sets $S,T$ whose symmetric difference is larger than $\sqrt{\max\{|T|,|S|\}}$.  When $k \in \Omega(\log\log n)$ we can get arbitrarily close to $1-1/e$ for $O(1)$-correlated noise.

\paragraph{Modification of algorithms for large $k$ for $\sqrt{k}$-correlated noise.}  For large $k$, if we have that $k \gg d^2$, then the approximation guarantee we get is still arbitrarily close to $1 - 1/e$ even when $\mathcal{D}$ is $d$-correlated. To do this, we modify the smoothing neighborhood and the definition of smooth values as follows.  Recall that in \textsc{Smooth-Greedy}, we select an arbitrary set of elements $H$ of size $\ell$ for smoothing, and compute the noisy smooth value of $S\cup a$ by averaging all subsets of $H$:
\[ \widetilde{F}(S\cup a) = \frac{1}{2^\ell}\sum_{H' \subset H} \widetilde f \left(S \cup \left ( a \cup H' \right ) \right).\]

In the $d$-correlated case, for each $1 \le i \le d$ and $1 \le j \le \ell$ we choose a \emph{bundle} $h(i)_j$ of $d$ elements, such that every two bundles are disjoint. Denote $H(i) = \{h(i)_1, \ldots h(i)_\ell$, and $H = \Cup_{i,j}h(i)_j$ the set of all elements we used.
The noisy smooth value with smoothing set $H(i)$ is now:

\begin{align*}
\widetilde{F}^{(i)}(S\cup a) =\frac{1}{2^\ell}\sum_{H' \subset H(i)} \widetilde f(S \cup a \cup H')
\end{align*}
where we abuse notation and use $S \cup a \cup H'$ instead of $S \cup \{a\} \cup_{h(i)_j \in H'}h(i)_j$.

We will run \textsc{Smooth-Greedy} with the smoothing sets $H(1),\ldots,H(d)$, where in each iteration $i \mod d$ we use $H(i)$ as the smoothing set.  Exactly as in the original algorithm, we generate $S$ by iteratively adding $k - |H|$ elements from $N\setminus H$ that maximize the smooth value in every iteration, and we then return $S\cup H$.  As before, \textsc{Slick- Greedy} employs \textsc{Smooth-Greedy}. 

To prove correctness of the algorithm we need to show that the evaluations of the surrogate functions are independent.  We will first show by induction on $|S|$ that between iterations, the oracle calls are independent.  

\begin{claim}
Any oracle call at iteration $i$ is independent of any previous oracle call at iteration $r<i$.
\end{claim}

\begin{proof}
Let $S(i)$ be the set of elements we have already committed to in stage $i$.  Consider an evaluation of
$\widetilde{f} (S(i) \cup a \cup H')$ for some non empty $H' \subset H(i \mod d)$ at iteration $i$, and an oracle evaluation $\widetilde{f} (S(r) \cup b \cup H'')$ made at some iteration $r < s$ with some non empty $H'' \subset H(r \mod d)$ and $b \notin S(r) \cup H$.
If $r \le i - d$, then the symmetric difference between $S(i) \cup a$ and $S(r) \cup b$ is at least of size $d$.  Since $a,b \notin H$, and $S(i) \cap H = \emptyset$, this means that the symmetric difference of $S(i) \cup a  \cup H'$ and $S(r) \cup b  \cup H''$ is at least of size $d$, for any $H''\subset H(r \mod d)$, and thus the calls are independent.  If $r > s - d$, then $i \mod d \neq r\mod d$, and hence $S(i) \cup a \cup H'$ and $S(r) \cup b \cup H''$ are independent because of the symmetric difference between $H'$ and $H''$.
\end{proof}

\begin{claim}\label{independentElement}
When evaluating $\widetilde{F}^{(i)}(S\cup a)$, all noise multipliers are independent.
\end{claim}

\begin{proof}
When evaluating $\widetilde{F}^{(i)}(S\cup a)$ we call the noisy oracle on sets of the form $S \cup a \cup H'$. Since each $H'$ corresponds to a different subset of $H(i)$, and $H(i)$ is a collection of $\ell$ bundles of size $d$, the symmetric difference between every two sets $H',H'' \subseteq H(i)$, is at least $d$.
\end{proof}

As in the original \textsc{Smooth-Greedy} procedure, we can show that at every iteration, when $S$ is the set of elements we selected in previous iterations, an element $a$ added to $S$ implies that w.h.p. $F(S\cup a)$ is arbitrarily close to $\max_{b \notin H} F(S\cup b)$ (see Claim \ref{independentElement}). Let
$a_1, a_2, \ldots a_{n - |S| - |H|}$ denote the elements which are being considered. For each element $a_i$, we have that if $F(S\cup a_i)$ is non negligible then w.h.p $\widetilde F(S\cup a_i)$ approximates $F(S\cup a_i)$, and if $F(S\cup a_i)$ is negligible then so is $\widetilde F(S\cup a_i)$. While for $a_i$, $a_j$ these events may well be correlated, since the probability of failure is inverse polynomially small and there are only $n - |S| - |H|$ events, we can take a union bound and say that with high probability for every $i$ if $F(S\cup a_i)$ is negligible so is $\widetilde F(S\cup a_i)$, and if $F(S\cup a_i)$ is non negligible then it is well approximated by $\widetilde F(S\cup a_i)$.

Thus, we know that at every iteration $i$ when $S$ is the set of elements selected in previous iterations, we have selected the element $a$ that is arbitrarily close to  $\max_{b \notin H} F^{(i)}(S\cup b)$.  From the arguments in the paper we know that this implies that for an arbitrarily small $\gamma>0$ we have: 
$$f_{S}(a)\geq (1-\gamma)f_{S\cup H(i)}(b) \geq (1-\gamma)f_{S\cup H}(b)$$
where the right inequality is due to submodularity and the fact that $H(i) \subseteq H$.  The guarantees of \textsc{Smooth-Greedy} therefore apply in this case as well.  What remains to show is that \textsc{Slick-Greedy} is unaffected by this modification.  This is easy to verify as \textsc{Slick-Greedy} takes $1/\delta$ disjoint sets $H_{1},\ldots,H_{1/\delta}$, and the arguments discussed apply for every such set.  Since we apply \textsc{Smooth-Compare} $1/\delta$ times with sets of size $\ell$ it is easy to implement as well. 

\paragraph{Modification of algorithms for small $k$ for $O(1)$-correlated noise.}  A similar idea works also for the small $k$ case, assuming $d$ is constant. In this case, we add $c \gg d/\epsilon$ elements at each phase of the algorithm. We modify the definition of $\widetilde F$ in the following way. First we take a an arbitrary partition $P_1, \ldots P_{(n - |S|)/d}$ on the elements not in $S$, in which each $P_i$ is of size $d$, and a partition $Q_1 \ldots Q_{(|S|+|A|)/d}$ of the elements in $S \cup A$. We estimate the value of a set $A$ given $S$ using:
\[\widetilde F(S \cup A) = \frac{d^2}{(|S|+|A|) (|N| - |S| - |A|)} \sum {Q_i \in A} \sum_{P_j} \widetilde f(((S \cup A)\setminus Q_i) \cup P_j)\]
and modify the rest of the algorithm accordingly. 

Correctness relies on three steps:
\begin{enumerate}
\item First, when we are in iteration $i$ of the algorithm (after we already added $(i - 1)c$ elements to $S$), all the sets we apply the oracle on are of size $c\cdot i$, and hence they are independent of any set of size $c(i-1)$ or less which were used in previous phases;

\item Second, when we evaluate $\widetilde F(S \cup A)$ for a specific set $A$, we only use sets which are independent in the comparison. Here we rely on changing $d$ elements in $A$ each time, and replacing them by another set of $d$ elements;

\item Finally, we treat each set $A$ separately, and show that if its marginal contribution is negligible then w.h.p its mean smooth value is not too large, and if its marginal contribution is not negligible, then w.h.p. $\widetilde F(S \cup A)$ approximates $F(S \cup A)$ well. Taking a union bound over all the bad events we get that the set $A$ chosen has large (non-noisy) smooth mean value.
\end{enumerate}

\subsection{Information Degradation}
We have written the paper as if the algorithm gains no additional information for querying a point twice. The generalization to a case where the algorithm gets more information each time but there is a degradation of information is simple: whenever the algorithms we presented here want to query a point just query it multiple times, and feed the expected value of the point given all the information one has to the algorithm. Hence it makes sense to focus on the extreme case where only the first query is helpful, as common in the literature of noisy optimization (e.g. \cite{braverman2008noisy})

\subsection{Approximate Submodularity}
In this paper our goal is to obtain near optimal guarantees as defined on the original function that was distorted through noise.  That is, we assume that there is an underlying submodular function which we aim to optimize, and we only get to observe noisy samples of it.  An alternative direction would be to consider the problem of optimizing functions that are approximately submodular:
$$\max_{S:|S|\leq k}\widetilde{f}(S) $$
The notion of approximate submodularity has been studied in machine learning~\cite{KC10,DK11,DDK12,EKDN16}.  More generally, given the desirable guarantees of submodular functions, it is interesting to understand the limits of efficient optimization with respect to the function classes we aim to optimize.

\paragraph{Impossibility for $\epsilon$-adversarial approximation.}  If we assume that the function is an adversarial $(1\pm \epsilon)$ approximation of a submodular function, our lower bound from Section~\ref{sec:adversarial} for erroneous oracles implies that no polynomial time algorithm can obtain a non-trivial approximation.

\paragraph{Trivial reduction for noise in $[1-\epsilon,1+\epsilon]$.}  When $\mathcal{D} \subseteq [1-\epsilon,1+\epsilon]$, and the noise is i.i.d across sets, the algorithms in the paper obtain a solution arbitrarily close to $\left(\frac{1-\epsilon}{1+\epsilon}\right)\left(1-\frac{1}{e}\right)$ of $\max_{S:|S|\leq k} \widetilde{f}(S)$. 

\paragraph{Impossibility for unbounded noise.}  If we assume that a noisy process of a distribution with unbounded support altered a submodular function, then there are trivial impossibility results. Suppose that the initial submodular function is the constant function that gives $1$ to every set. If we apply (e.g.) Gaussian noise to it, then the optimal algorithm is just to try random sets and hope for the best, and no polynomial time algorithm can achieve a constant factor approximation.

\paragraph{Optimal approximation via black-box reduction.}
First, note that there is an algorithm which runs in time $n^k$ and finds the optimal subset of size $k$: query $\widetilde f$ on all subsets of size at most $k$, and choose the maximal one.  Notice that this is in contrast to the setting we study throughout the paper in which there is a lower bound of $(2k - 1)/2k + O(1/\sqrt{n})$.  The interesting regime is $k = \omega(1)$, where there is a black-box reduction from the problem of maximizing a submodular function given an approximately submodular function, to the problem of maximizing an approximately submodular function.  Since we can solve the original problem within a factor arbitrarily close to $1-1/e$ we get an optimal approximation guarantee in this case as well.  Let $\max\DD(t) = \mathbb{E}[ \max_{\xi_1, \ldots \xi_{t} \sim \DD}\{\xi_1, \ldots, \xi_t\}]$ be the expected maximum value of $t$ i.i.d samples of $\DD$.

\begin{lemma}
An algorithm which uses $t \le {n \choose k}$ queries to $\widetilde f$ cannot achieve approximation ratio better than:
  \[\frac{\max\DD(t)}{\max \DD({{n \choose k}} )}.\]
\end{lemma}

\begin{proof}
Suppose that $f(S) = 1$ for every set $S$. The best that the algorithm can do is query $t$ sets with at most $k$ elements, and output the maximal one. The approximation ratio of this is exactly
  \[\frac{\max \DD(t)}{\max \DD({{n \choose k}} )}\]
If the algorithm queries sets with more than $k$ elements, the approximation would deteriorate.
\end{proof}

\begin{lemma}
Suppose there exists an algorithm which given $k \in \omega(1)$ returns a solution $S$ s.t. $f(S) \ge \gamma \max_{T: |T| \le k}f(T)$ using $q$ queries to a noisy oracle.  Then, for any $t \in \poly(n)$ there is an algorithm that uses $q+t$ to a noisy oracle and returns a solution $S'$ s.t.:
$$\widetilde{f}(S') \geq \Big (\gamma - o(1) \Big ) 
\left ( \frac{\max\mathcal{D}(t)}{\max \mathcal{D}({n \choose k})}\right ) \max_{T:|T|\leq k }\widetilde{f}(T).$$
\end{lemma}
\begin{proof}
  Let $r$ be such that ${n-k \choose r} \ge t$. Since $t$ is polynomial in $n$, we have that $r$ is constant. Run the algorithm to obtain a set $G$ of size $k-r$. From submodularity and the fact that $r$ is constant:
  \[f(G) \ge \gamma \max_{S: |S|\le k-r} f(S) \ge (1 - r/k) \gamma \max_{S:|S|\le k} f(S) \ge (1 - o(1)) \gamma \max_{S:|S|\le k} f(S)\]
For every set of $r$ elements $\{x_1, \ldots, x_r\}$ where $x_i \not \in G$, the algorithm queries $\widetilde f$ on $G \cup \{x_1, \ldots x_r\}$, and chooses the set with maximum value. It is easy to see that the expected value of this set would be at least
  $\max \mathcal{D}(t) (1 - r/k) \gamma \max_{S:|S|\le k} f(S)$, which gives the ratio.
\end{proof}

\newpage \section{Impossibility for Adversarial Noise}\label{sec:adversarial}
In this section we show that there are very simple submodular functions for which no randomized algorithm with access to an $\epsilon$-erroneous oracle can obtain a reasonable approximation guarantee with a subexponential number of queries to the oracle.  Intuitively, the main idea behind this result is to show that a noisy oracle can make it difficult to distinguish between two functions whose values can be very far from one another.  The functions we use are similar to those used to prove information theoretic lower bounds for submodular optimization and learning~\cite{MSV08,PSS08,FMV11,DBH11,V13}.

\begin{theorem}\label{thm:adversarial}
No randomized algorithm can obtain an approximation strictly better than $O(n^{-1/2+\delta})$ to maximizing monotone submodular functions under a cardinality constraint using $e^{n^{\delta}}/n$ queries to an $\epsilon$-erroneous oracle, for any 
fixed $\epsilon,\delta<1/2$.
\end{theorem}

\begin{proof}
We will consider the problem of $\max_{S:|S|\leq k}f(S)$ where $k=n^{1/2+\delta}$.  Let $X \subseteq N$ be a random set constructed by including every element from $N$ with probability $n^{-1/2+\delta}$.  We will use this set to construct two functions that are close in expectation but whose maxima have a large gap, and show that access to a noisy oracle implies distinguishing between these two functions.  The functions are:

\begin{itemize}
\item $f_{1}(S) = \min \left \{ |S \cap X| \cdot n^{1/2}+ \frac{n^{1/2+\delta}}{\epsilon}, |S|\cdot n^{1+\delta} \right \}$ 
\item $f_{2}(S) = \min \left \{|S| \cdot n^{\delta} + \frac{n^{1/2+\delta}}{\epsilon}, |S|\cdot n^{1+\delta} \right \}$
\end{itemize}
Notice that both functions are normalized monotone submodular: when $S=\emptyset$ both functions evaluate to $0$, and otherwise are affine. 
By the Chernoff bound we know that $|X| \geq n^{1/2+\delta}/2$ with probability $1-e^{-\Omega(n^{1/2+\delta})}$.  Conditioned on this event we have that $\max_{S:|S|\leq k}f_{1}(S) = f_{1}(X) \in O(n^{1+\delta})$ whereas $f_{2}$ is symmetric and $\max_{S:|S|\leq k}f_{2}(S) \in O(n^{1/2+2\delta})$.  Thus, an inability to distinguish between these two functions implies there is no approximation algorithm with approximation better than $O(n^{-1/2+\delta})$.
We define the erroneous oracle as follows.  If the function is $f_{2}$, its oracle returns the exact same value as $f_{2}$ for any given set.  Otherwise, the function is $f_{1}$ and its erroneous oracle is defined as:
\[
    \widetilde{f}(S)=
\begin{cases}
    f_{2}(S), & \text{if } (1-\epsilon)f_{1}(S) \leq f_{2}(S) \leq (1+\epsilon)f_{1}(S)\\
    f_1(S)   & \text{otherwise}
\end{cases}
\]
Notice that this oracle is $\epsilon$-erroneous, by definition.

Suppose now that the set $X$ is unknown to the algorithm, and the objective is $\max_{S:|S|\leq k}f_1(S)$.  We will first show that no deterministic algorithm that uses a single query to the erroneous oracle $\widetilde{f}$ can distinguish between $f_{1}$ and $f_{2}$, with exponentially high probability (equivalently, we will show that a single query to the algorithm cannot find a set $S$ for which $f_1(S) < (1 - \epsilon) f_2(S)$ or $f_1(S) > (1 + \epsilon) f_2(S)$ with exponentially high probability).  For a single query algorithm, we can imagine that the set $X$ is chosen after the algorithm chooses which query to invoke, and compute the success probability over the choice of $X$. In this case, all the elements are symmetric, and the function value is only determined by the size of the set that the single-query algorithm queries.  

In case the query is a set $S$ of cardinality smaller or equal to $n^{1/2}$, by the Chernoff bound we have that $|S \cap X| \leq (1+\beta)n^{\delta}$ for any $\beta<1$ with probability at least ${1-e^{-\Omega(\beta^{2}n^{\delta})}}$.  Thus:
\begin{align*}
\hspace{2in}
\frac{n^{1/2+\delta}}{\epsilon}& \leq & f_{1}(S)  & \leq & \Big (1+\beta+\frac{1}{\epsilon}\Big)n^{1/2+\delta}&\hspace{2in}\\
\hspace{2in}
\frac{n^{1/2+\delta}}{\epsilon}  & \leq & f_{2}(S)  & \leq & \Big (1+\frac{1}{\epsilon}\Big )n^{1/2+\delta} &\hspace{2in}\\
\end{align*}
It is easy to verify that for $\beta<\epsilon/(1-\epsilon)$: $(1-\epsilon)f_{1}(S )\leq f_{2}(S) \leq (1+\epsilon)f_{1}(S)$.  Thus, for any query of size less or equal to $n^{1/2}$ the likelihood of the oracle returning $f_{1}$ is $1-e^{-\Omega(n^{\delta})}$.

In case the oracle queries a set of size greater than $n^{1/2}$ then again by the Chernoff bound, for any $\beta<1$ we have that with probability at least ${1-e^{-\Omega(\beta^{2}n^{1/2})}}$:
$$\Big (1-\beta \Big ) \frac{|S|}{n^{1/2-\delta}}\leq |S \cap X| \leq \Big (1+\beta \Big)\frac{|S|}{n^{1/2-\delta}}$$
For $\beta \leq \epsilon/(1-\epsilon)$, this implies that:
$$(1-\epsilon)f_{1}(S) \leq f_{2}(S) \leq (1+\epsilon)f_{1}(S)$$
Therefore, for any fixed $\epsilon \in (0,1)$, the algorithm cannot distinguish between $f_1$ and $f_2$ with probability $1-e^{-\Omega{(n^{\delta})}}$ by querying the erroneous oracle with a set larger than $n^{1/2}$.  To conclude, by a union bound we get that with probability $1-e^{-\Omega{(n^{\delta})}}$ no algorithm can distinguish between $f_1$ and $f_2$ using a single query to the erroneous oracle, and the ratio between their maxima is $O(n^{1/2-\delta})$.

To complete the proof, suppose we had an algorithm running in time $e^{n^{\delta}}/n$ which can approximate the value of a submodular function, given access to an $\epsilon$-erroneous oracle with approximation ratio strictly better than $O(n^{-1/2+\delta})$ which succeeds with probability 2/3. This would let us solve the following decision problem:
{\em Given access to an $\epsilon$-erroneous oracle for either $f_1$ or $f_2$, determine which function is being queried.}  To solve the decision problem, given access to an erroneous oracle of unknown function, we would use the hypothetical approximation algorithm to estimate the  value of the maximal set of size $n^{1/2 + \delta}$. If this value is strictly more than $n^{1/2 + 2\delta}$, the function is $f_1$ (since $f_1(X) = O(n^{1 + \delta}))$, and otherwise it is $f_2$.

The reduction allows us to show that distinguishing between the functions in time $e^{n^{\delta}}/n$ and success probability $2/3$ is impossible.  For purpose of contradiction, suppose that there is a (randomized) algorithm for the decision problem, and let $p$ denote the probability that it outputs $f_2$ if it sees an oracle which is fully consistent with $f_2$. To succeed with probability $2/3$, it must be the case that whenever the algorithm gets $f_1$ as an input, it finds a set $S$ for which the noisy oracle returns $f_1(S)$ with probability at least $2/3 - p/2 \ge 1/6$. Whenever it finds such a set, the algorithm is done, since it can compute $f_2(S)$ without calling the oracle, and hence it knows that $f_1$ was chosen in the decision problem.

In this case, we know that the algorithm makes up to $e^{n^{\delta}}/n$ queries, until it sees a set for which it gets $f_1(S)$. But this means that there is an algorithm with success probability at least $O(n/6e^{n^{\delta}})$ that makes a single query. This algorithm guesses some index $i < e^{n^{\delta}}/n$, and simulates the original algorithm for $i - 1$ steps (by feeding it with $f_2$ without using the oracle), and then using the oracle in step $i$. If the algorithm guesses $i$ to be the first index in which the exponential time algorithm sees $f_1(S)$, then the single query algorithm would succeed.  Hence, since we showed that no single query (randomized) algorithm can find a set $S$ such that $f_1(S) < (1 - \epsilon) f_2(S)$ or $f_1(S) > (1 + \epsilon) f_2(S)$ with just one query this concludes the proof. 
\end{proof}

The following remarks are worth mentioning:
\begin{itemize}
\item The functions we used in the lower bound are very simple examples of coverage functions;
\item If one does not require the function to be normalized, then the lower bound holds for affine functions, i.e. $f(S) = \sum_{a\in S}f(a) + C$, where $C$ independent of $S$;
\item The lower bound is tight: for any $\epsilon$-erroneous oracle there is a $\frac{1-\epsilon}{1+\epsilon}\cdot \max\{n^{-1/2},1/k\}$ approximation by simply partitioning the ground sets to arbitrary sets of size $\min\{\sqrt{n},k\}$, and select the set whose value according to the erroneous oracle is maximal;
\item The lower bound applies to additive noise by simply applying an additive version of the Chernoff bound.
\end{itemize}

Somewhat surprisingly, the above theorem suggests that a good approximation to a submodular function does not suffice to obtain reasonable approximation guarantees.  In particular, guarantees from learning or sketching where the goal is to approximate a submodular function up to constant factors may not necessarily be meaningful for optimization.  It is important to note that for some classes of submodular functions such as additive functions ($f(S)=\sum_{a \in S}f(a)$), we can obtain algorithms that are robust to adversarial noise.  A very interesting open question is to characterize the class of submodular functions that are robust to adversarial noise.

\newpage \section{More related work}\label{sec:related}
\paragraph{Submodular optimization.}  Maximizing monotone submodular functions under cardinality and matroid constraints is heavily studied.  The seminal works of \cite{nemhauser1978analysis,fisher1978analysis} show that the greedy algorithm gives a factor of $1-1/e$ for maximizing a submodular function under a cardinality constraint and a factor $1/2$ approximation for matroid constraints.  For max-cover which is a special case of maximizing a submodular function under a cardinality constraint, Feige shows that no poly-time algorithm can obtain an approximation better than 1-1/e unless P=NP~\cite{feige1998threshold}.  Vondrak presented the continuous greedy algorithm which gives a $1 - 1/e$ ratio for maximizing a monotone submodular function under matroid constraints~\cite{vondrak08}.  This is optimal, also in the value oracle model~\cite{MSV08,khot2005inapproximability,nemhauser1978best}.  It is interesting to note that with a demand oracle the approximation ratio is strictly better than $1 - 1/e$~\cite{feige2006approximation}.  When the function is not monotone, constant factor approximation algorithms are known to be obtainable as well~\cite{feige2011maximizing,LMNS09,BFNS12,BFNS14}.  In general, in the past decade there has been a development in the theory of submodular optimization, through concave relaxations~\cite{AS04,CE11}, the multilinear relaxation~\cite{calinescu2007maximizing,vondrak08,CJV15}, and general rounding technique frameworks~\cite{contention}.  In this paper, the techniques we develop arise from first principles: we only rely on basic properties of submodular functions, concentration bounds, and the algorithms are variants of the standard greedy algorithm.

\paragraph{Submodular optimization in game theory.}  Submodular functions have been studied in game theory almost fifty years ago~\cite{sh71}.  In mechanism design submodular functions are used to model agents' valuations~\cite{LLN01} and have been extensively studied in the context of combinatorial auctions (e.g.~\cite{DNS05,dobzinski2006improved,DLN08,MSV08,SODA10,DFK11,PP11,DRY11,DV12}).  Maximizing submodular functions under cardinality constraints have been studied in the context of combinatorial public projects~\cite{PSS08,SS08,BSS10,LSST13} where the focus is on showing the computational hardness associated with not knowing agents valuations and having to resort to incentive compatible algorithms.  Our adversarial lower bound implies that if agents err in their valuations, optimization may be hard, regardless of incentive constraints.

\paragraph{Submodular optimization in machine learning.}  In the past decade submodular optimization has become a central tool in machine learning and data mining (see surveys~\cite{krgue11acm,icmltut2,nipstut}).  Problems include identifying influencers in social networks~\cite{KKT03,gor11} sensor placement~\cite{leskovec07,golovin10}, learning in data streams~\cite{sgk09,gomes10,KMVV13,BMKK14}, information summarization~\cite{libi11,linb11}, adaptive learning~\cite{golkr11}, vision~\cite{jegelkaB11,jb11inference,kohliOJ13}, and general inference methods~\cite{krgu07,jb11inference,djolonga14}.  In many cases the submodular function is learned from data, and our work aims to address the case in which there is potential for noise in the model.

\paragraph{Learning submodular functions.}  One of the main motivations we had for studying optimization under noise is to understand whether submodular functions that are learned from data can be optimized well.  The standard framework in the literature for learning set functions is  \emph{Probably Mostly Approximately Correct} (\texttt{PMAC}) learnability due to Balcan and Harvey~\cite{BH11-PMAC}.  This framework nicely generalizes Valiant's notion of \emph{Probably Approximately Correct} (\texttt{PAC}) learnability~\cite{Valiant84-PAC}.  Informally, PMAC-learnability guarantees that after observing polynomially-many samples of sets and their function values, one can construct a surrogate function that is with constant probability over the distributions generating the samples, likely to be an approximation of the submodular function generating the data.  Since the seminal paper of  Balcan and Harvey there has been a great deal of work on learnability of submodular functions~\cite{FK14-coverage, BCIW12-valuations,  BDFKNR12-sketches, FV13-submodular_juntas, feldman2015tight, Balcan15-AAMAS}.  As discussed in the paper, our lower bounds imply that one cannot optimize the surrogate function \texttt{PMAC} learned from data.  If the approximation is via i.i.d noise on sets sufficiently far, this may be possible. 

\paragraph{Approximate submodularity.}  The concept of approximate submodularity has been studied in machine learning for dictionary selection and feature selection in linear regression~\cite{KC10,DK11,DDK12,EKDN16}.  Generally speaking, this line of work considers approximate submodularity by defining a notion of the \emph{submodularity ratio} of a function, defined in terms of how close it is to have a diminishing returns property.  This ratio depends on the instance, which in the worst-case may result in a function that poorly approximates a submodular function.  In practice however, these works show that in a broad range of applications the functions of interest are sufficiently close to submodular.  Recently, the notion of approximate \emph{modularity} (i.e. additivity) has been studied in~\cite{DCD0K15} which give an optimal algorithm for approximating an approximately modular function via a modular function.  These notions of approximate modularity and approximate submodularity are the model in which we have noise on the marginals.  As discussed in Section~\ref{sec:extensions}, if the error on the marginals is adversarial, there are regimes in which non-trivial guarantees are impossible.  If one assumes the marginal approximations are i.i.d our positive results apply.    

\paragraph{Combinatorial optimization under noise.}  Combinatorial optimization with noisy inputs can be largely studied through consistent (independent noisy answers when querying the oracle twice) and inconsistent oracles. For inconsistent oracles, it usually suffices to repeat every query $O(\log n)$ times, and eliminate the noise. To the best of our knowledge, submodular optimization has been studied under noise only in instances where the oracle is inconsistent or equivalently small enough so that it does not affect the optimization~\cite{KKT03,KG05}.  One line of work studies methods for reducing the number of samples required for optimization (see e.g. \cite{feige1994computing,ben2008bayesian}), primarily for sorting and finding elements. On the other hand, if two identical queries to the oracle always yield the same result, the noise can not be averaged out so easily, and one needs to settle for approximate solutions, which has been studied in the context of tournaments and rankings~\cite{kenyon2007rank,braverman2008noisy,ajtai2009sorting}. 

\paragraph{Convex optimization under noise.}  Maximizing functions under noise is also an important topic in convex optimization.  The analogue of our model here is one where there is a zeroth-order noisy oracle to a convex function.  As discussed in the paper, the question of polynomial-time algorithms for noisy convex optimization is straightforward and the work in this area largely aims at improving the convergence rate~\cite{elster1995grid,glad1977optimization,khuri1996response,kushner1978stochastic,polyak1987introduction}.

\newpage \section{Acknowledgements}
A.H. was supported by ISF 1241/12; Y.S. was supported by NSF grant CCF-1301976, CAREER CCF-1452961, a Google Faculty Research Award, and a Facebook Faculty Gift.  We thank Vitaly Feldman who pointed out the application to active learning.  We are deeply indebted to Lior Seeman, who has carefully read previous versions of the manuscript and made multiple invaluable suggestions.

\clearpage
\newpage
\bibliographystyle{abbrv}
\bibliography{noise-bib}
\clearpage
\newpage

\appendix
\addcontentsline{toc}{section}{Appendices}
\section*{Appendix}
\newpage \section{Combinatorial Smoothing}\label{sec:appendix_smoothing}
In this section we illustrate a general framework we call \emph{combinatorial smoothing} that we will use in the subsequent sections.  Intuitively, combinatorial smoothing mitigates the effects of noise and enables finding elements whose marginal contribution is large.

\paragraph{Some intuition.}  Recall from our earlier discussion that implementing the greedy algorithm requires identifying $\arg\max f(S \cup a)$ for a given set $S$ of elements selected by the algorithm in previous iterations.  Thus, if for some $a,b \in N$ we can compare $S \cup a$ and $S \cup b$ and decide whether $f(S\cup a) > f(S\cup b)$ or vice versa, we can implement the greedy algorithm.  Put differently, viewing a set as a point on the hypercube, given two points in $\{0,1\}^n$ we need to be able to tell which one has the larger true value, using a noisy oracle.  In a world of continuous optimization, a reasonable approach to estimate the true value of a point in $[0,1]^{n}$ with access to a noisy oracle is to take a small neighborhood around the point, sample values of points in its neighborhood, and average their values.  Taking polynomially-many samples allows concentration bounds to kick in, and using a small enough diameter can often guarantee that the averaged value is a reasonable estimate of the point's true value.  Surprisingly, the spirit of this idea can used in submodular optimization.

\paragraph{Smoothing neighborhood.}  For a given subset $A\subseteq N$ a \emph{smoothing function} is a method which assigns a family of sets $\mathcal{H}(A)$ called the \emph{smoothing neighborhood}.  The smoothing function will be used to create a smoothing neighborhood for a small set $A$.  This set $A$ whose marginal contribution we aim to evaluate, is essentially a candidate for a greedy algorithm.  In the application in Section~\ref{sec:largek} the set $A$ is simply be a single element, whereas in Section~\ref{sec:smallk} the set $A$ is of size $O(1/\epsilon)$.

\begin{definition}
For a given function ${f:2^{N} \to \mathbb{R}}$, $A,S \subseteq N$, and smoothing neighborhood $\mathcal{H}(A)$:
\begin{itemize}
\item  $F_{S}(A) := \mathbb{E}_{X\in \mathcal{H}(A)} \left [ f_{S}(X) \right ]$ (called the \emph{smooth marginal contribution} of $A$),
\item $F(S\cup A) := \mathbb{E}_{X\in \mathcal{H}(A)}   \left [  f(S \cup X) \right ]$ (called the smooth value of $S \cup A$)
\item $\widetilde{F}(S\cup A) := \mathbb{E}_{X\in \mathcal{H}(A)} \left [ \widetilde{f}(S \cup X) \right ]$ (called the \textbf{noisy} smooth value of $S\cup A$).
\end{itemize}
\end{definition}

The idea behind combinatorial smoothing is to select a smoothing neighborhood which includes sets whose value is in some sense close to the value of the set $A$ whose marginal contribution we wish to evaluate.  Intuitively, when the sets are indeed close, by averaging the values of the sets in $\mathcal{H}(A)$ we can mitigate the effects of noise and produce meaningful statistics (see Figure~\ref{fig:smoothing}).

\begin{figure*}[t]
\begin{centering}
                \includegraphics[trim = 0mm 0mm 0mm 0mm, height=110mm]{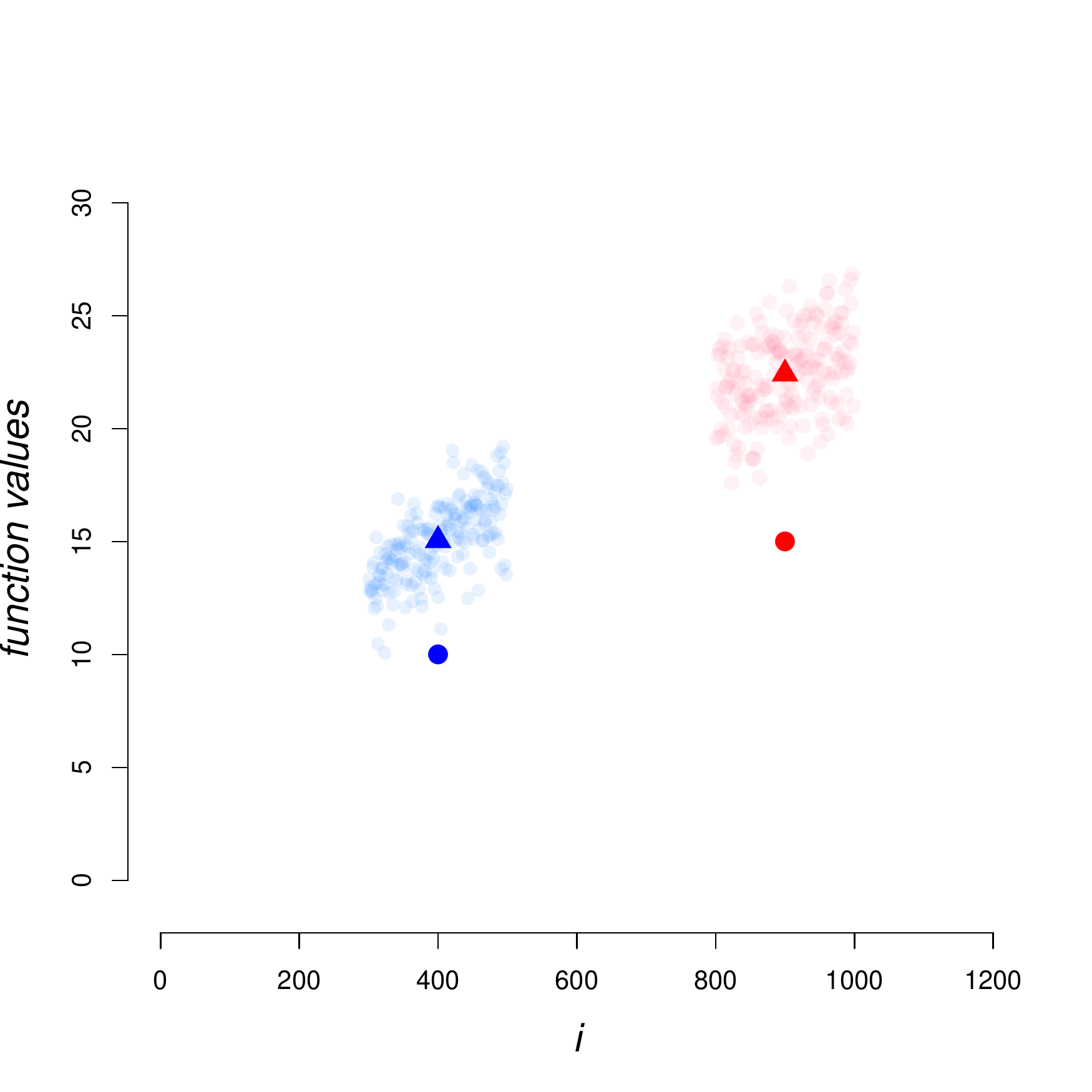}
                 \caption{\footnotesize{An illustration of smoothing.  For every element in the ground set we associate an index $i \in [n]$ and define the submodular function as $f(S) = \sqrt{\sum_{i\in S}i}/2 -c$ for a constant $c>0$.  The blue dot depicts the \emph{true} value of the element $a$ associated with the index $i=400$ and the red dot depicts the \emph{true} value of the element $b$ associated with the index $j=900$.  The light blue and light red dots depict the noisy function values of elements associated with indices $i$ in the range $|i-400|\leq 100$ and $|i-900|\leq 100$.  For $S=\emptyset$, and smoothing neighborhoods $\mathcal{H}(a) =\{i: |i-a| \leq 100\}$ and $\mathcal{H}(b) =\{i: |i-b| \leq 100\}$ we depict $\widetilde{F}(S\cup a)$ and $\widetilde{F}(S\cup b)$ as the blue and red triangles, respectively.  Intuitively, an algorithm which needs to decide whether $a$ (blue point) is larger than $b$ (red point) will decide by comparing $\widetilde{F}(S\cup a)$ (blue triangle) and $\widetilde{F}(S\cup b)$ (red triangle).}}\label{fig:smoothing}
                 \end{centering}
\end{figure*}

\subsection*{Smoothing arguments}
In our model, the algorithm may only access $\widetilde{F}(S\cup A)$.  
Ideally, given a set $S$ and a smoothing neighborhood $\mathcal{H}(A)$ we would have liked to apply concentration bounds and show that the noisy smooth value is arbitrarily close to the non-noisy smooth value, i.e. $F(S\cup A) \approx \widetilde{F}(S\cup A)$ or: 
$$ \sum_{i \in \mathcal{H}(A)} f(S \cup X_i) \approx  \sum_{i \in \mathcal{H}(A)}\xi_i f(S \cup X_i)$$
If the values in $\{f(S \cup X_i)\}_{i=1}^{|\mathcal{H}(A)|}$ were arbitrarily close, we could simply apply a concentration bound by taking the value of any one of the sets, say $S \cup X_j$, and for $v_j = f(S \cup X_j)$, since all the values are close, we would be guaranteed that: 
$$ \sum_{i \in \mathcal{H}(A)}\xi_i f(S \cup X_i) \approx \sum_{i \in \mathcal{H}(A)}\xi_i f(S \cup X_j) =  v_j \cdot  \sum_{i \in \mathcal{H}(A)}\xi_i$$ 
In continuous optimization this is usually the case when averaging over an arbitrarily small ball around the point of interest, and concentration bounds apply.  In our case, due to the combinatorial nature of the problem, the values of the sets in the smoothing neighborhood may take on very different values.  For this reason we cannot simply apply concentration bounds.  The purpose of this section is to provide machinery that overcomes this difficulty.  The main ideas can be summarized as follows:
\begin{enumerate}
\item In general, there may be cases in which we cannot perform smoothing well and cannot get the noisy smooth values to be similar to the true smooth values.  We therefore define a more modest, yet sufficient goal.  Since our algorithms essentially try to replace the step of adding the element $a \in\argmax_{b}f(S\cup b)$ in the greedy algorithm with $a' \in \argmax_{b} F(S\cup b)$, it suffices to guarantee that for the set $A$ which maximizes the noisy smooth values, that set also well approximates the (non-noisy) smooth values.  More precisely our goal is to show that if for an arbitrarily small $\delta>0$ we have that $A\in \argmax_{B} \widetilde{F}(S\cup B)$ then ${F(S\cup A) \geq (1-\delta) \max_{B} F(S\cup B)}$;
\item To show that $A\in \argmax \widetilde{F}(S\cup A)$ implies $F(S\cup A) \geq (1-\delta) \max_{B} F(S\cup B)$ for an arbitrarily small $\delta>0$, we prove two bounds.  
Lemma~\ref{lem:noisysmoothing_lowerbound} lower bounds the noisy smooth contribution of a set in terms of its (true) smooth contribution.  Lemma~\ref{lem:boundonb} upper bounds the smooth noisy contribution of any element against its smooth contribution.  The key difference between these lemmas is that Lemma~\ref{lem:noisysmoothing_lowerbound} lower bounds the value in terms the \emph{variation} of the smoothing neighborhood.  The variation of the neighborhood is the ratio between the set with largest value and that with lowest value in the neighborhood.  Intuitively, for elements with large values the variation of the neighborhood is bounded, and thus we can show that the noisy smooth value of these elements is nearly as high as their true smooth values.  

\item Together, these lemmas are used in subsequent sections to show that an element with the largest noisy smooth marginal contribution is an arbitrarily good approximation to the element with the largest (non-noisy) smooth marginal contribution.  This is achieved by showing that the lower bound on the smooth value of an element with large (non-noisy) smooth marginal contribution beats the upper bound on the smooth (non-noisy) value of an element with slightly smaller smooth contribution.
\end{enumerate} 

The first lemma gives us tail bounds on the upper and lower bounds of the value of the noise multiplier in any of the calls made by a polynomial-time algorithm.  We later use these tail bounds in concentration bounds we use in the smoothing procedures. 

\begin{lemma}\label{lem:upperBoundNoise}
Let $\omega_{\max} = \max\{\xi_{1},\ldots, \xi_{m}\}$ and $\omega_{\min} = \min\{\xi_{1},\ldots, \xi_{m}\}$, where $\xi_i \sim \mathcal{D}$ and $\mathcal{D}$ is a noise distribution with a generalized exponential tail.  For any $\delta > 0$ and sufficiently large $m$, we have that: 
\begin{itemize}
\item $\Pr[\omega_{\max} < m^{\delta}] >1 -  e^{-\Omega( m^\delta/\ln m)}$
\item $\Pr[\omega_{\min} > m^{-\delta}] >1 -  e^{-\Omega( m^\delta/\ln m)}$
\end{itemize}
\end{lemma}

\begin{proof}
As $m$ tends to infinity, this lemma trivial for any noise distribution which is bounded, or has finite support. If the noise distribution is unbounded, we know that its tail is subexponential. Thus, at any given sample the probability of seeing the value $m^{\delta}$ is at most $e^{-O(m^\delta)}$ where the constant in the big $O$ notation depends on the magnitude of the tail. Iterating this a polynomial number of times gives the bound. The proof of the lower bound is equivalent.
\end{proof}

The definition below of the variation of the neighborhood quantifies the ratio between the largest possible value and the smallest possible value achieved by a set in the neighborhood.

\begin{definition}\label{def:variation}
For given sets $A,S \subseteq N$, the \textbf{variation} of the neighborhood denoted $v_S(\mathcal{H}(A))$ is:
$$v_S(\mathcal{H}(A)) = \frac{\max_{T \in \mathcal{H}(A)}f_S(T)}{\min_{T \in \mathcal{H}(A)}f_S(T)}.$$
\end{definition}

The following lemma gives a lower bound on the noisy smooth value in terms of the (non-noisy) smooth value and the variation.  Intuitively, when an element has large value its variation is bounded, and the lemma implies that its noisy smooth value is close to its smooth value.  Essentially, when the variation is bounded $\widetilde{F}(S)\approx (1-\lambda)(1-\epsilon)F(S)$ for $\lambda$ and $\epsilon$ that vanish as $n$ grows large. 

\begin{lemma}\label{lem:noisysmoothing_lowerbound}
Let $f:2^{N} \to \mathbb{R}$, $A,S\subset N$, $\omega = \max_{A_i \in \mathcal{H}(A)}\xi_{A_i}$, and $\mu$ be the mean of the noise distribution.  For 
$\epsilon = \min \left \{1, 2v_S(\mathcal{H})\cdot |\mathcal{H}(A)|^{-1/4} \right \}$ for any $\lambda<1$ w.p $1- e^{-\Omega(\frac{\lambda^2 t^{{1}/{4}}}{\omega})}$ we have:
$$\widetilde{F}(S\cup A)
> (1-\lambda)\mu \cdot \left ( f(S) +  (1-\epsilon) \cdot F_{S}(A) \right ).
$$
\end{lemma}

\begin{proof}
Let $A_{1},\ldots,A_{t}$ be the sets in $\mathcal{H}(A)$ and let $\alpha_{1},\ldots, \alpha_{t}$ denote the corresponding marginal contributions and $\xi_{1}\ldots,\xi_{t}$ denote their noise multipliers.  In these terms the noisy smooth value is:
\begin{align}
{\widetilde{F}}(S \cup A)
= \frac{1}{t}\sum_{i=1}^{t}\xi_{i} ( f(S) + \alpha_i) = \frac{1}{t}\sum_{i=1}^{t}\xi_{i}  f(S)  +  \frac{1}{t}\sum_{i=1}^{t}\xi_{i} \alpha_i.\label{eq:sk1}
\end{align}
Let $\omega$ be the upper bound on the value of the noise multiplier.  Applying the Chernoff bound, we get that for any $\lambda <1$ with probability at least $1-e^{-\Omega(\lambda^2 t/ \omega)}$:
$$\frac{1}{t}\sum_{i=1}^{t}\xi_{i} f(S) \geq (1-\lambda)\mu f(S).$$
To complete the proof we need to argue about concentration of the second term in~(\ref{eq:sk1}).  To do so, in our analysis we will consider a fine discretization of $\{\alpha_i\}_{i \in [t]}$ and apply concentration bounds on each discretized value.  Define $\alpha_{\max} = \max_{i \in [t]}\alpha_i$ and $\alpha_{\min} = \min_{i \in [t]}\alpha_i$.  We can divide the set of values $\{\alpha_{i}\}_{i \in [t]}$ to $t^{1/4}$ bins $\textsc{bin}_1,\ldots,\textsc{bin}_{t^{1/4}}$, where a value $\alpha_{i}$ is placed in the bin $\textsc{bin}_{q}$ if
$$(q-1)\cdot \alpha_{\max} t^{-1/4} \leq \alpha_{i} \leq q\cdot \alpha_{\max} t^{-1/4}$$

Say a bin is \emph{dense} if it contains at least $t^{1/4}$ values and \emph{sparse} otherwise.
Consider some dense bin $\textsc{bin}_q$ and let $\alpha_{\min(q)} = \min_{i \in \textsc{bin}_q}\alpha_i$ and $\alpha_{\max(q)} = \max_{i \in \textsc{bin}_q}\alpha_i$.  Since every bin is of width $\alpha_{\max}\cdot t^{-1/4}$ we know that:
$$\alpha_{\min(q)} \geq \alpha_{\max(q)} - \alpha_{\max}\cdot t^{-1/4}$$
Applying concentration bounds as above, we get that $\sum_{i \in \textsc{bin}_q}\xi_{i} \geq (1-\lambda) \mu \cdot |\textsc{bin}_q|$ with probability at least $1-e^{-\Omega(\lambda^2 t^{1/4}/ \omega)}$ for any $\lambda <1$.  Thus, with this probability:
\begin{align*}
\sum_{i \in \textsc{bin}_q}\xi_{i}\alpha_{i}
& \geq \sum_{i \in \textsc{bin}_q}\xi_{i} \alpha_{\min(q)} \\
& \geq (1-\lambda)\mu \cdot |\textsc{bin}_q| \cdot  \alpha_{\min(q)}\\
& \geq (1-\lambda)\mu \cdot |\textsc{bin}_q| \cdot \left ( \max\left\{0,\alpha_{\max(q)} - \alpha_{\max}\cdot {t^{-1/4}}\right\} \right )  \\
& > (1-\lambda) \mu  \cdot |\textsc{bin}_q| \cdot \left  ( \max\left\{0,1 -  \frac{\alpha_{\max}}{\alpha_{\max(q)}}  \cdot t^{-1/4}\right\}\right )\alpha_{\max(q)}\\
& \geq (1-\lambda) \mu  \cdot |\textsc{bin}_q| \cdot \left  ( \max\left\{0,1 -  \frac{\alpha_{\max}}{\alpha_{\min}} \cdot t^{-1/4}\right\}\right )\alpha_{\max(q)}\\
& = (1-\lambda) \mu \cdot |\textsc{bin}_q| \cdot \left  ( \max\left\{0,1 -  v_S\left(\mathcal{H}(A)\right) \cdot t^{-1/4}\right\}\right )\alpha_{\max(q)}
\end{align*}
Taking a union bound over all (at most $t^{1/4}$) dense bins, we get that with probability $1- e^{-\Omega(\lambda^2 t^{1/4}/ \omega)}$:
\begin{align}
\sum_{i \in \textrm{\tiny{dense}}}\xi_{i} \alpha_{i}
& \geq (1-\lambda)\mu \cdot \left (1 -   \max\left\{0,  v_S\left(\mathcal{H}(A)\right) \cdot t^{-1/4} \right\} \right ) \sum_{\textsc{bin}_q \in \textrm{\tiny{dense}}}|\textsc{bin}_q|\cdot \alpha_{\max(q)}\nonumber \\
& \geq (1-\lambda)\mu \cdot \left (\max\left\{0, 1 -   v_S\left(\mathcal{H}(A)\right)\cdot t^{-1/4}\right\} \right )\sum_{i \in \textrm{\tiny{dense}}}\alpha_i.\label{eq:conc1}
\end{align}
Let $\alpha=\frac{1}{t}\sum_{i=1}^{t}\alpha_{i}$.  Since we have less than $t^{1/4}$ elements in a sparse bin, and in total $t^{1/4}$ bins, the number of elements in sparse bins is at most $t^{1/2}$.  We can use this to effectively lower bound the values in sparse bins in terms of $\alpha$:
\begin{align}
\sum_{i \in \textrm{\tiny{dense}}}\alpha_{i}\nonumber
& =  \sum_{i=1}^{t}\alpha_i - \sum_{i \in \textrm{\tiny{sparse}} }\alpha_i\nonumber\\
& \geq \max\left\{0,\sum_{i=1}^{t}\alpha_i - t^{1/2}\alpha_{\max}\right \}\nonumber\\
& \geq \max\left\{0,t \alpha  - t^{1/2} \alpha_{\max}\right \}\nonumber\\
& > \max\left\{0,t \cdot  \left ( 1 -\frac{\alpha_{\max}}{\alpha_{\min}} \cdot  t^{-1/2}\right ) \alpha\right \}\nonumber\\
& = \max\left\{0,t \cdot  \left ( 1 -v_S(\mathcal{H}) \cdot t^{-1/2}\right )\alpha \right\} \label{eq:conc2}
\end{align}
Putting (\ref{eq:conc1}) and (\ref{eq:conc2}) we get that for any $\lambda<1$, with probability $1- e^{-\Omega(\lambda^2 t^{1/4}/\omega)}$:
\begin{align*}
{\widetilde{F}}_{S}(A)
& = \frac{1}{t}\sum_{i =1}^{t}\xi_{i}\cdot \alpha_i \\
& \geq \frac{1}{t}\sum_{i \in \textrm{\tiny{dense}}}\xi_{i}\cdot \alpha_i \\
& \geq (1-\lambda)\mu \cdot (\max \left \{0,1-v_S\left(\mathcal{H}(A)\right)\cdot t^{-1/4} \right \})\cdot \frac{1}{t}\sum_{i \in \textrm{\tiny{dense}}}\alpha_{i}\\
& \geq (1-\lambda)\mu\cdot (\max \left \{0,1-v_S\left(\mathcal{H}(A)\right)\cdot t^{-1/4}\right \})(\max \left \{0,1-v_S\left(\mathcal{H}(A)\right)\cdot t^{-1/2}\right\})\alpha\\
& > (1-\lambda)\mu\cdot (\max \left \{0,1-2v_S\left(\mathcal{H}(A)\right)\cdot t^{-1/4}\right \})\alpha\\
& = (1-\lambda)\mu\cdot (\max \left \{0,1-2v_S\left(\mathcal{H}(A)\right)\cdot t^{-1/4}\right \})F_{S}(A)
\end{align*}
Taking a union bound we get that for any positive $\lambda<1$ with probability $1- e^{-\Omega(\lambda^2 t^{{1}/{4}}/\omega)}$:
\begin{align*}
\widetilde{F}(S \cup A)
& = \frac{1}{t}\sum_{i=1}^{t}\xi_{i}  f(S)  +  \frac{1}{t}\sum_{i=1}^{t}\xi_{i} \alpha_i \\
 & > (1-\lambda)\mu \cdot \Big ( f(S) +  (\max \left \{0,1-2v_S(\mathcal{H}(A))\cdot t^{-1/4}) \cdot F_{S}(A)\right \} \Big )\\
& = (1-\lambda)\mu \cdot \Big ( f(S) +  (1- \min \left \{1,2v_S(\mathcal{H}(A))\cdot t^{-1/4}) \cdot F_{S}(A)\right \} \Big ).\qedhere
\end{align*}
\end{proof}

The next lemma gives us an upper bound on the noisy smooth value.  The bound shows that for sufficiently large $t$ (the size of the smoothing neighborhood, which always depends on $n$), for small $\lambda>0$ we have that $\widetilde{F}(S) \approx (1+\lambda)F(S)+3t^{-1/4}\cdot \alpha_{\max}$.  In our applications of smoothing $\alpha_{\max}\leq\texttt{OPT}$, and $t$ is large.  Since we use this upper bound to compare against elements whose value is at least some bounded factor of $\texttt{OPT}$, the dependency of the additive term on $\alpha_{\max}$ will be insignificant. 

\begin{lemma}\label{lem:boundonb}
Let $f:2^{N} \to \mathbb{R}$, $A,S \subseteq N$, $\omega = \max_{A_i \in \mathcal{H}(A)}\xi_{A_i}$, $\alpha_{\max} = \max_{A_i \in \mathcal{H}(A)}f_S(A_i)$ and $\mu$ be the mean of the noise distribution.  For $\epsilon = 3t^{-1/4} \alpha_{\max}$ we have that for any $\lambda<1$
with probability $1-e^{-\Omega(\lambda^2 t^{1/4}/\omega)}$:
$$ \widetilde{F}(S\cup A) < (1+\lambda) \mu \cdot \left (f(S)+ F_{S}(A) + \epsilon \right ).
$$
\end{lemma}

\begin{proof}
As in the proof of Lemma~\ref{lem:noisysmoothing_lowerbound} let $A_1,\ldots,A_{t}$ denote the sets in $\mathcal{H}(A)$, and for each set $A_i$ we will again use $\alpha_i$ to denote the marginal value $f_{S}(A_i)$ and $\xi_{i}$ to denote the noise multiplier $\xi_{S\cup \{A_{i}\}}$.
\begin{align}
\widetilde{F}(S\cup A) = \frac{1}{t}\sum_{i=1}^{t}\xi_i f(S) + \frac{1}{t}\sum_{i=1}^{t}\xi_i \alpha_i.\label{eq:fb}
\end{align}
As before, we will focus on showing concentration on the second term.  Define ${\alpha_{\max} = \max_{i}\alpha_i}$ and $\alpha_{\min}=\min_{i}\alpha_i$.  To apply concentration bounds on the second term, we again partition the values of $\{\alpha_i\}_{i \in [t]}$ to bins of width $\alpha_{\max}\cdot t^{-1/4}$ and call a bin dense if it has at least $t^{1/4}$ values and sparse otherwise.  Using this terminology:
$$ \sum_{i=1}^{t}\xi_i\alpha_i = \sum_{i \in \textrm{\tiny{dense}}}\xi_i\alpha_i + \sum_{i \in \textrm{\tiny{sparse}}}\xi_i\alpha_i.$$
Let $\textsc{bin}_{\ell}$ be the dense bin whose elements have the largest values.  Consider the $t^{1/4}/2$ largest values in $\textsc{bin}_{\ell}$ and call the set of indices associated with these values $L$.  We have:
\begin{align*}
\sum_{i=1}^{t}\xi_i\alpha_i
= \sum_{i \in \textrm{\tiny{dense}} \setminus L}\xi_i\alpha_i + \sum_{i \in L \cup {\textrm{\tiny{sparse}}}}\xi_i\alpha_i
\end{align*}
The set $L \cup {\textrm{\tiny{sparse}}}$ is of size at least $t^{1/4}/2$ and at most $t^{1/4}/2 + t^{1/2}$.  This is because $L$ is of size exactly $t^{1/4}/2$ and there are at most $t^{1/2}$ values in bins that are sparse since there are $t^{1/4}$ bins and a bin that has at least $t^{1/4}$ is already considered dense.  Thus, when $\omega$ is an upper bound on the value of the noise multiplier, from Chernoff, for any $\lambda < 1$ with probability $1-e^{-\Omega(\lambda^2 t^{1/4}/ \omega)}$:
\begin{align*}
\sum_{i \in L \cup {\textrm{\tiny{sparse}}}}\xi_i\alpha_i
& \leq \sum_{i \in L \cup {\textrm{\tiny{sparse}}}}\xi_{i}\alpha_{\max} \\
& < (1+\lambda)\mu \cdot |L \cup{\textrm{\tiny{sparse}}}|\cdot \alpha_{\max} \\
& \leq (1+\lambda) \mu \cdot \left (\frac{t^{1/4}}{2} + t^{1/2}\right)\alpha_{\max}  \\
& < (1+\lambda)\mu\cdot 2t^{1/2} \alpha_{\max}
\end{align*}
We will now use the same logic as in the proof of Lemma~\ref{lem:noisysmoothing_lowerbound} to apply concentration bounds on the values in the dense bins.  For a dense bin $\textsc{bin}_q$, let $\alpha_{\max(q)}$ and $\alpha_{\min(q)}$ be the maximal and minimal values in the bin, respectively.  As in Lemma~\ref{lem:noisysmoothing_lowerbound}, for any $\lambda <1$ with probability $1-e^{-\Omega(\lambda^2 t^{1/4}/ \omega)}$:
\begin{align*}
\sum_{i \in \textsc{bin}_q}\xi_{i}\alpha_{i}
& \leq \sum_{i \in \textsc{bin}_q}\xi_{i} \cdot \alpha_{\max(q)} \\
& \leq (1+\lambda)\mu \cdot  \alpha_{\max(q)}\cdot |\textsc{bin}_q|\\
& \leq (1+\lambda)\mu \cdot \left ( \alpha_{\min(q)} + \alpha_{\max}\cdot {t^{-1/4}} \right )  \cdot |\textsc{bin}_q|\\
& < (1+\lambda) \mu \cdot \left ( |\textsc{bin}_q|\cdot  {\alpha_{\min(q)}} +  |\textsc{bin}_q|\alpha_{\max}  \cdot t^{-1/4} \right )
\end{align*}
Applying a union bound we get with probability $1-e^{-\Omega(\lambda^2 t^{1/4}/ \omega)}$:
\begin{align*}
\sum_{i \in \textrm{\tiny{dense}}\setminus L} \xi_{i}\alpha_i
& < \sum_{q} (1+\lambda) \mu \cdot \left ( |\textsc{bin}_q|\cdot  {\alpha_{\min(q)}} +  |\textsc{bin}_q|\alpha_{\max}  \cdot t^{-1/4} \right )\\
& < (1+\lambda) \mu \cdot t \left (\alpha + t^{-1/4} \alpha_{\max} \right)
\end{align*}
Together we have:
\begin{align*}
\frac{1}{t}\sum_{i=1}^{t}\xi_i \alpha_i
& =  \frac{1}{t} \left ( \sum_{i \in \textrm{\tiny{dense}}\setminus L} \xi_{i}\alpha_i + \sum_{i \in L \cup {\textrm{\tiny{sparse}}}}\xi_i\alpha_i \right )\\
& < (1+\lambda) \mu \cdot \left (\alpha + t^{-1/4} \alpha_{\max} + 2t^{-1/2}\alpha_{\max} \right)\\
& < (1+\lambda)  \mu \cdot \left (\alpha + 3t^{-1/4} \alpha_{\max} \right)\\
& < (1+\lambda)  \mu \cdot \left (F_{S}(A) + 3t^{-1/4} \alpha_{\max} \right)\\
\end{align*}
By a union bound we get that with probability $1-e^{-\Omega(\lambda^2 t^{1/4}/\omega)}$:
\begin{equation*}
\widetilde{F}(S\cup A)
= \frac{1}{t}\sum_{i=1}^{t}\xi_i f(S) + \frac{1}{t}\sum_{i=1}^{t}\xi_i \alpha_i \\
\leq (1+\lambda) \mu \cdot \left (f(S)+ F_{S}(A) + 3t^{-1/4} \alpha_{\max} \right ).\qedhere
\end{equation*}
\end{proof}

\newpage \section{Optimization for Large $k$}
\subsection*{The Smooth Greedy Algorithm}\label{sec:appendix_largek}

\subsubsection*{Smoothing guarantees}

\begin{lemma*}[\ref{lem:boundona}]
For any fixed $\epsilon >0$, consider an $\epsilon$-relevant iteration of \textsc{Smooth-Greedy} where $S$ is the set of elements selected in previous iterations and ${a \in \arg\max_{b \notin H}\widetilde{F}(S\cup b)}$.  Then for $\delta=\epsilon^2/4k$ and sufficiently large $n$ we have that w.p. $\geq 1-1/n^4$:
$$F_S(a) \geq (1-\delta) \max_{b \notin H}F_{S}(b).$$
\end{lemma*}

To prove the above lemma we will need claims~\ref{lem:black_and_white},~\ref{clm:boundedvariance} and~\ref{claim:realslimshady}.  After proving~\ref{claim:realslimshady} the proof will follow by verifying that the number of sets in the smoothing set is sufficient to obtain the desired approximation $(1-\delta)$.

\begin{claim}\label{lem:black_and_white}
If $F_{S}(a) \geq F_{S}(b)$ then $f_{S}(a) \geq f_{S \cup H}(b)$.
\end{claim}
\begin{proof}
Assume for purpose of contradiction that $f_{S}(a) < f_{S \cup H}(b)$.  Since $f$ is a submodular function, $f_{S}(T)= f(S \cup T) - f(S)$ is also submodular (hence also subadditive).  Therefore $\forall H'\subseteq H$:
\begin{align*}
f_{S}(H' \cup a) & \leq f_{S}(H') + f_{S}(a) \rmk{subadditivity of $f_S$}\\
& < f_{S}(H') + f_{S \cup H}(b) 	\rmk{by assumption}\\
& \leq f_{S}(H') + f_{S \cup H'}(b) \rmk{submodularity of $f_S$}\\
& = f_{S}(H' \cup b).
\end{align*}
Notice however, that this contradicts our assumption:
\begin{equation*}
F_{S}(a) = \frac{1}{t}\sum_{H' \subseteq H}f_{S}(H' \cup a) < \frac{1}{t}\sum_{H' \subseteq H}f_{S}(H' \cup b) = F_{S}(b).\qedhere
\end{equation*}
\end{proof}

The following claim bounds the variation (see Definition~\ref{def:variation}) of the smoothing neighborhood of the element we selected.  This is a necessary property for later applying the smoothing arguments.

\begin{claim}\label{clm:boundedvariance}
Let $\epsilon>0$.  For an $\epsilon$-relevant iteration of \textsc{Smooth-Greedy}, let $S$ be the set of elements selected in previous iterations.   If ${a^\star \in \arg\max_{a \notin H}F_{S}(a)}$ then ${v_S\left(\mathcal{H}(a^\star)\right) < 3k/\epsilon}$.
\end{claim}

\begin{proof}
Let $b ^\star \in \argmax_{b \notin H} f_{H \cup S}(b)$.  By the maximality of $a^\star$ we have that $F_{S}(a^\star) \geq F_{S}(b^\star)$, and thus by Claim~\ref{lem:black_and_white} we get ${{f_{S}(a^\star) \geq f_{H\cup S}(b^\star)}}$.  Since the iteration is $\epsilon$-relevant we have that $f_{H\cup S}(b^\star) \geq \epsilon \cdot \texttt{OPT}_{H}/k$, and from monotonicity of $f$ we get:
$$\min_{H' \subseteq H}f_S(H' \cup a^\star) \geq f_{S}(a^\star)  \geq f_{H\cup S}(b^\star) \geq \frac{\epsilon \cdot  \texttt{OPT}_{H}}{k}$$
and since every set in $\mathcal{H}(a^\star)$ is of size at most $k$ we know that $\max_{H' \subseteq H}f_S(H' \cup a^\star) \leq \texttt{OPT}$.  Together with the fact that $\texttt{OPT} \leq e\cdot\texttt{OPT}_{H}$ we get:
\begin{equation*}
v_S\left(\mathcal{H}(a^{\star})\right) = \frac{\max_{H' \subseteq H}f_S(H' \cup a^\star)}{\min_{H' \subseteq H}f_S(H' \cup a^\star)} \leq \frac{\texttt{OPT}}{\texttt{OPT}_{H}}\cdot \frac{k}{\epsilon} < \frac{3k}{\epsilon}.\qedhere
\end{equation*}
\end{proof}

We can now show that in $\epsilon$-relevant iterations the value of the element which maximizes the noisy smooth value is comparable to that of the (non-noisy) smooth value, with high probability.  Recall that we use $t$ to denote the size of the smoothing neighborhood.

\begin{claim}\label{claim:realslimshady}
Given $\epsilon>0$ assume $t\geq \left (\frac{110k\cdot \log n}{\epsilon\delta}\right)^8$.  For an $\epsilon$-relevant iteration of \textsc{Smooth-Greedy}, let $S$ be the elements selected in previous iterations and ${a \in \arg\max_{b \notin H}\widetilde{F}(S\cup b)}$.  Then, w.p. $\geq 1-1/n^4$:
$$F_S(a) \geq (1-\delta) \max_{b \notin H}F_{S}(b).$$
\end{claim}

\begin{proof}
Let $a^\star$ be the element which maximizes smooth marginal contribution: 
$$a^\star \in \argmax_{b \notin H}F_{S}(a)$$  
We will show that for any element $b$ whose smooth marginal contribution is a factor of $(1-\delta)$ smaller than the smooth marginal contribution of $a^\star$, then w.h.p. its \emph{noisy} value of is smaller than that of $a^\star$.  That is, for any $b \notin H$ for which ${F_{S}(b) < (1-\delta) F_{S}(a^{\star})}$ we get that $\widetilde{F}(S\cup b) < \widetilde{F}(S\cup a^{\star})$ with probability at least $\Omega(1-1/n^{5})$.  The result will then follow by taking a union bound over all comparisons.  We will show that $a^\star$ likely beats $b$ by lower bounding $\widetilde{F}(S \cup a^\star)$ and upper bounding $\widetilde{F}(S \cup b)$ using the smoothing arguments from the previous section.  We use $\omega$ to denote the value of the largest noise multiplier realized throughout the iterations of the algorithm.  We later argue that we can upper bound $\omega\leq 6\log n$ as the noise distribution has an exponentially decaying tail.
\begin{itemize}
\item \textbf{Lower bound on $\widetilde{F}(S \cup a^\star)$:}
First, from Claim~\ref{clm:boundedvariance} we know that $v_S(\mathcal{H}(a^\star)) \leq 3k/\epsilon$.  Together with Lemma~\ref{lem:noisysmoothing_lowerbound} we get that $\forall \lambda<1$ with probability $1-e^{-\Omega(\lambda^2t^{1/4}/\omega)}$:
\begin{align}
\widetilde{F}(S \cup a^\star) > (1-\lambda)\mu \cdot \Big ( f(S) +  \left(1-\frac{6k}{\epsilon}\cdot t^{-1/4}\right) \cdot F_{S}(a^{\star}) \Big )\label{eq:upperboundonastar}
\end{align}
\item \textbf{Upper bound on $\widetilde{F}(S\cup b)$:} Letting $\beta_{\max} = \max_{X \in \mathcal{H}(b)}f(X)$, from Lemma~\ref{lem:boundonb}, we get that $\forall \lambda<1$ with probability $1- e^{-\Omega(\lambda^2 t^{1/4}/\omega)}$:
\begin{align}
\widetilde{F}(S \cup b) < (1+\lambda)\mu \cdot \left ( f(S) + F_{S}(b) + 3t^{-1/4}\beta_{\max}\right )\label{eq:fas}
\end{align}
We'll express this inequality in terms of $f(S)$ and $F_{S}(a^\star)$ as well.  First, since all sets in $\mathcal{H}(b)$ are of size at most $k$ we also know that $\beta_{\max} \leq \texttt{OPT}$.  Thus:
\begin{align}\label{eq:boundonbeta}
3t^{-1/4}\beta_{\max} \leq 3t^{-1/4}\cdot \texttt{OPT}
\end{align}
We will now bound $\texttt{OPT}$ in terms of $F_{S}(a^\star)$.  Since every set in $\mathcal{H}(a^{\star})$ includes $a^\star$, from monotonicity we get that $F_{S}(a^\star) \geq f_{S}(a^\star)$.  Let ${b^\star \in \argmax_{b\notin H} f_{H\cup S}(b)}$.  Due to the maximality of $a^\star$ we have that $F_S(a^\star) \geq F_{S}(b^\star)$ and by Claim~\ref{lem:black_and_white} we know that $f_{S}(a^\star) \geq f_{S \cup H}(b^\star)$. Since the iteration is $\epsilon$-relevant we get:
\begin{equation}\label{eq:boundonF}
F_{S}(a^\star) \geq f_{S}(a^\star) \geq f_{S \cup H}(b^\star) \geq \frac{f_{S \cup H}(O_H)}{k} \geq \frac{\epsilon\cdot \texttt{OPT}_{H}}{k} > \frac{\epsilon\cdot \texttt{OPT}}{3k}
\end{equation}
Putting~(\ref{eq:boundonF}) together with~(\ref{eq:boundonbeta}) we get:
$$3t^{-1/4}\beta_{\max} \leq \frac{k }{\epsilon} \cdot 9t^{-1/4}\cdot F_{S}(a^\star) $$
Plugging into~(\ref{eq:fas}) and using the assumption that $F_S(b)<(1-\delta)F_S(a^\star)$ we get:
\begin{align}
\widetilde{F}(S\cup b)
& < (1+\lambda)\mu \cdot \left ( f(S)+ F_S(b) + \left (  9t^{-1/4}\cdot \frac{k}{\epsilon}\right )F_{S}(a^\star)\right )\\
& < (1+\lambda)\mu \cdot \left ( f(S)+\left (  9t^{-1/4}\cdot \frac{k}{\epsilon} + (1-\delta)\right )F_{S}(a^\star)\right )\label{eq:fab}
\end{align}
\end{itemize}
Putting~(\ref{eq:upperboundonastar}) together with~(\ref{eq:fab}) we get that $\forall \lambda<1$ with probability at least $1-2e^{-\Omega(\lambda^2t^{1/4}/\omega)}$:
\begin{align*}
\widetilde{F}(S \cup a^\star) -  \widetilde{F}(S \cup b)
& > \mu \cdot \left (F_{S}(a^\star) \left [(1-\lambda) \left ( 1 - \frac{6k}{\epsilon} t^{-1/4} \right)  - (1+\lambda)\left(\frac{9k}{\epsilon} t^{-1/4} + (1-\delta) \right )\right ] - 2\lambda f(S) \right )\\
& \geq \mu \cdot \left (F_{S}(a^\star) \left [(1-\lambda) \left ( 1 - \frac{6k}{\epsilon} t^{-1/4} \right)  - (1+\lambda)\left(\frac{9k}{\epsilon} t^{-1/4} + (1-\delta)\right )\right ] - 2\lambda \texttt{OPT} \right )\\
& > \mu \cdot \left (F_{S}(a^\star) \left [(1-\lambda) \left ( 1 - \frac{6k}{\epsilon} t^{-1/4} \right)  - (1+\lambda)\left(\frac{9k}{\epsilon} t^{-1/4} + (1-\delta)\right )\right ] - 2\lambda \frac{3k}{\epsilon}F_{S}(a^\star) \right )\\
& = \mu \cdot \left (F_{S}(a^\star) \left [(1-\lambda) \left ( 1 - \frac{6k}{\epsilon} t^{-1/4} \right)  - (1+\lambda)\left(\frac{9k}{\epsilon} t^{-1/4} + (1-\delta)\right )  - 2\lambda \frac{3k}{\epsilon}\right ]  \right )\\
& = \mu \cdot \left (F_{S}(a^\star) \left [\delta - \frac{15k}{\epsilon}\cdot t^{-1/4} - \lambda \left ( (2-\delta) + \frac{3k}{\epsilon}\cdot t^{-1/4} + \frac{6k}{\epsilon}    \right ) \right ]  \right )\\
%
%
%
& > \mu \cdot \left ( F_{S}(a^\star) \left [\delta - \frac{k}{\epsilon}\left(15t^{-1/4} +10\lambda\right) \right ] \right ) \\
\end{align*}
The second inequality above is an application of~(\ref{eq:boundonF}) and the fact that $f(S)\leq \texttt{OPT}$ since $|S| \leq k$.  The third is from~(\ref{eq:boundonF}).  

For the result to hold we need the above difference to be strictly positive, and hold with probability $\Omega(1-1/n^5)$.  Thus, sufficient conditions would be:
\begin{enumerate}
\item $\frac{k}{\epsilon}\cdot 15 t^{-1/4} \leq \frac{\delta}{2}$, and
\item $10\lambda \leq \frac{\delta}{2}$, and
\item $1-2\exp(\frac{-\lambda^2t^{1/4}}{\omega}) \in \Omega(1-1/n^{5}).$
\end{enumerate} 

The first condition holds when $t\geq (30k/\epsilon\delta)^4$; the second condition holds when $\lambda =  \epsilon\delta/20k$.  For $\omega = 6\log n$ and $\lambda =  \epsilon\delta/20k$, the third condition is satisfied when:
$$ \frac{(\epsilon\delta)^2 t^{1/4}}{20^2k^2 \omega} = \frac{(\epsilon\delta)^2 t^{1/4}}{20^2k^2 6\log n} \geq 5\log n$$
rearranging:
$$ t\geq 12000^4\left ( \frac{k \log n}{ \epsilon\delta }  \right ) ^8 $$
Thus, since $t$ in the lemma statement respects:
$$t\geq \left ( \frac{110 k\log n}{ \epsilon\delta}\right )^8 > 12000^4\left ( \frac{k \log n}{ \epsilon\delta }  \right ) ^8$$
we have that the first, second, and third conditions are met conditioned on $\omega\leq 6\log n$.  That is, we have that the difference is positive with probability $1-2\exp(\frac{-\lambda^2t^{1/4}}{\omega}) \geq 1-2/n^{5}$, conditioned on $\omega \leq 6\log n$.  From lemma~\ref{lem:upperBoundNoise} we know that the probability of $\omega > 6\log n$ is smaller than $1/n^5$ for sufficiently large $n$.  Therefore, by taking a union bound on the probability of the event in which the difference is negative and the probability that  $\omega > 6\log n$, both occurring with probability smaller than $2/n^{5}$ we have that the probability of the difference being positive is at least $1-4/n^5 \in \Omega(1-1/n^5)$, as required.
\end{proof}

\begin{proof}[\textbf{Proof of Lemma~\ref{lem:boundona}}]
By Claim~\ref{claim:realslimshady}, when $\delta = \epsilon^2/4k$ for any fixed $\epsilon>0$ we need to verify that for sufficiently large $n$:
$$t > \left ( \frac{110 k\log n}{ \epsilon\delta}\right )^8 =   \frac{(440k^{2}\log n)^8}{\epsilon^3}$$  
In the case where $k\geq \log n$ we use $\ell = 25\log n$ and thus $t = 2^{\ell} = n^{25}$ and the above inequality holds.  When $k < \log n$ we use $\ell = 33\log\log n$ and thus $t=\log^{33}n$ and the above inequality holds in this case as well.  We therefore have the result with probability at least $1-1/n^4$.\footnote{Note that we could have used smaller values of $\ell$ to achieve the desired bound.  The reason we exaggerate the values of $\ell$ is to be consistent with the analysis of \textsc{Slick-Greedy} which necessitates these slightly larger values of $\ell$.}
\end{proof}

\subsection*{Approximation guarantee}

\begin{claim*}[\ref{claim:newboundona}]
For any $\epsilon>0$, let $\delta\leq \epsilon^2/4k$.  Suppose that the iteration is $\epsilon$-relevant and let ${b^\star \in \argmax_{b\notin H } f_{H \cup S}(b)}$.  If $F_{S}(a) \geq (1-\delta)F_S(b^\star)$, then:
$$f_S(a) \geq (1-\epsilon)f_{H\cup S}(b^\star).$$
\end{claim*}
\begin{proof}
%
First, we upper bound $F_S(a)$:
\begin{align*}
F_{S}(a)
& = \frac{1}{t} \sum_{H'\subseteq H} f_{S}(H'\cup a)    \rmk{by definition of $F_S$}\\
& = \frac{1}{t} \sum_{H'\subseteq H}\left ( f_{S}(H') + f_{S\cup H'}(a)    \right)\\
& \leq \frac{1}{t} \sum_{H'\subseteq H}\left ( f_{S}(H') + f_{S}(a) \right)  \rmk{by submodularity of $f$}  \\
& = f_{S}(a) + \frac{1}{t} \sum_{H'\subseteq H} f_{S}(H') \rmk{$t = 2^{|H|}$}\\
\end{align*}
Next, we lower bound $(1-\delta)F_{S}(b^\star)$:
\begin{align*}
(1-\delta)F_{S}(b^\star)
& = (1-\delta)\frac{1}{t} \sum_{H'\subseteq H} f_{S}(H'\cup b^\star)    \rmk{by definition of $F_S$}\\
& = (1-\delta)\frac{1}{t} \sum_{H'\subseteq H}\left ( f_{S}(H') + f_{S\cup H'}(b^\star)    \right)\\
& \geq (1-\delta) \frac{1}{t} \sum_{H'\subseteq H}\left ( f_{S}(H') + f_{S \cup H}(b^\star) \right)  \rmk{by submodularity of $f$}  \\
& = (1-\delta)f_{H \cup S}(b^\star) - \delta \frac{1}{t} \sum_{H'\subseteq H} f_{S}(H') + \frac{1}{t} \sum_{H'\subseteq H} f_{S}(H') \rmk{$t = 2^{|H|}$}
\end{align*}
Since $F_{S}(a) \geq (1-\delta)F_{S}(b^\star)$ this implies that:
\begin{align*}
f_S(a) 
& \geq (1-\delta)f_{H\cup S}(b^\star) - \delta \frac{1}{t}\sum_{H'\subseteq H}f_S(H')\\
& \geq (1-\delta)f_{H\cup S}(b^\star) - \delta \frac{1}{t}\sum_{H'\subseteq H}f_S(H)\rmk{monotonicity of $f$}\\
& \geq (1-\delta)f_{H\cup S}(b^\star) - \delta f_S(H) \rmk{$t=|H'|$}\\
& \geq (1-\delta)f_{H\cup S}(b^\star) - \delta \texttt{OPT} \rmk{$|H|\leq k$}\\
& \geq (1-\delta)f_{H\cup S}(b^\star) - e\delta \texttt{OPT}_H \rmk{$\texttt{OPT}_H \geq \texttt{OPT}/e$}\\
& \geq (1-\delta)f_{H\cup S}(b^\star) - e\delta \cdot \frac{k}{\epsilon}\cdot f_{H\cup S}(b^\star) \rmk{$\epsilon$-relevant iteration}\\
& = \left (1-\delta \left (1+ \frac{e\cdot k}{\epsilon}\right) \right )f_{H\cup S}(b^\star) \\
& \geq \left (1-\delta \left (\frac{4k}{\epsilon}\right) \right )f_{H\cup S}(b^\star) \\
& = (1-\epsilon)f_{H\cup S}(b^\star). \rmk{$\delta \leq \epsilon^2/4k$}\qedhere
\end{align*}
\end{proof}

\begin{claim*}[\ref{lem:bound}]
For any fixed $\epsilon>0$, consider an $\epsilon$-relevant iteration of \textsc{Smooth-Greedy} with $S$ as the elements selected in previous iterations.  Let ${a \in \arg\max_{b \notin S \cup H}\widetilde{F}(S\cup b)}$.  Then, w.p. $\geq 1-1/n^4$:
$$f_S(a)\geq \Big (1-\epsilon\Big)\left[\frac{1}{k'}\Big (\texttt{OPT}_H - f(S) \Big)\right].$$
\end{claim*}

\begin{proof}\label{lem:bound_proof}
Let $O \in \textrm{argmax}_{T:|T| \leq k'}f_{H}(T)$, $o^{\star} \in \textrm{argmax}_{o \in O} f_{H\cup S}(o)$ and $b^\star \in \argmax_{b \notin H }f_{H\cup S}(b)$.
From Lemma~\ref{lem:boundona} we know that with probability $1-1/n^4$ we have $F_{S}(a)\geq (1-\delta)F_{S}(b^\star)$ for $\delta=\epsilon^2/4k$, and together with Claim~\ref{claim:newboundona} we get:
$$f_{S}(a) \geq (1-\epsilon)f_{H\cup S}(b^\star)\geq (1-\epsilon) f_{H \cup S}(o^{\star})$$
From subadditivity $f_{H \cup S}(o^{\star}) \geq f_{H \cup S}(O)/k'$ and thus:
\begin{equation*}
f_{S}(a)   \geq (1-\epsilon) f_{H \cup S}(o^{\star}) \geq \left (\frac{1-\epsilon}{ k' } \right) f_{H \cup S}(O)
\geq  \left (\frac{1-\epsilon }{ k' }\right)  \Big ( f_{H}(O) - f (S) \Big ).\qedhere
\end{equation*}
\end{proof}

\begin{lemma*}[\ref{lem:opt_h}]
Let $S$ be the set returned by \textsc{Smooth-Greedy} and $H$ its smoothing set.  Then, for any fixed $\epsilon>0$ when $k\geq 3\ell/\epsilon$ with probability of at least $1-1/n^3$ we have that:
$$f(S \cup H) \geq \left ( 1-1/e -\epsilon/3 \right) \texttt{OPT}_{H}.$$
\end{lemma*}

\begin{proof}\label{lem:opt_h_proof}
In case  $\texttt{OPT}_{H}<\texttt{OPT}/e$ then $H$ alone provides a $1-1/e - \epsilon/3$ approximation.  To see this, let $O \in \argmax_{T:|T|\leq k}f(T)$ and $O' \in \argmax_{T:|T|\leq k'}f(T)$, and $O_H \in \argmax_{T:|T|\leq k'}f_H(T)$.  We get:
\begin{align*}
(1-\epsilon/3)f(O)
& \leq f(O') \rmk{$k'=k-\ell$ and $k \geq 3\ell/\epsilon$} \\
& \leq f(H \cup O') \rmk{monotonicity }      \\
& = f(H) + f_{H}(O') \\
& \leq f(H)+f_{H}(O_{H}) \rmk{optimality of $O_{H}$}  \\
& < f(H) + f(O)/e \rmk{$e\texttt{OPT}_{H} < \texttt{OPT}$}
\end{align*}
Thus:
$$f(H) \geq \left (1-\frac{1}{e} - \frac{\epsilon}{3} \right)\texttt{OPT} \geq \left (1-\frac{1}{e} - \frac{\epsilon}{3} \right )\texttt{OPT}_H$$

In case $\texttt{OPT}_{H} \geq \texttt{OPT}/e$ we set $\gamma=\min\{1/e,\epsilon/6\}$.  We will use the following notation.  At every iteration $i \in [k']$ of the \emph{while} loop in the algorithm, we will use $a_i$ to denote the element that was added in that step, and $S_i : =\{a_1,\ldots,a_i\}$. 

First, notice that if there exists an iteration $i$ that is not $\gamma$-relevant, our bound trivially holds:
%
%
\begin{align*}
f_{H \cup S_i}(O_H)  \leq k'\cdot \max_{o \in O_H} f_{H\cup S_i}(o) \leq k' \cdot \max_{b \notin S_i \cup H} f_{H\cup S_i}(b) \leq k'\cdot \frac{\gamma \texttt{OPT}_H}{k} <  \gamma \texttt{OPT}_H
\end{align*}
Since $f_{H\cup S_i}(O_H) = f(H\cup S_i \cup O_H) - f(H \cup S_i)$, the above inequality implies that $f(H \cup S_i) > f(H\cup S_i \cup O_H) - \gamma\texttt{OPT}_{H}$.  But this implies:
\begin{align*}
f(S \cup H) 
& \geq f(S_i \cup H) \\
& > f(O_H \cup S_i \cup H)  - \gamma \texttt{OPT}_H\\
& \geq f(O_H)  - \gamma \texttt{OPT}_H\\
& \geq f_H(O_H)  - \gamma \texttt{OPT}_H\\
& = (1- \gamma) \texttt{OPT}_H\\
& \geq (1- 1/e) \texttt{OPT}_H
\end{align*}

It remains to prove the approximation guarantee in the case that every iteration is $\gamma$-relevant.  To do so, we can apply a standard inductive argument on Claim~\ref{lem:bound} to show that $S$ alone provides a $1-1/e - \epsilon/3$ approximation.  Claim~\ref{lem:bound} states that for $\gamma$-relevant iterations, at every stage $i \in [k']$:
\begin{align}
f(S_{i+1}) - f(S_{i})\geq (1-\gamma)\left [\frac{1}{k'} \left ( f_H(O_{H}) - f(S_{i})\right )\right].\label{eq:induct}
\end{align}
We will show that at every stage $i \in [k']$:
$$f(S_{i}) \geq (1-\gamma)\left(1- \left(1- \frac{1}{k'}\right)^i\right)f_H(O_H).$$
The proof is by induction on $i$.  For $i=1$ we have that $S_i = \{a_1\}$ and invoking Claim~\ref{lem:bound} with ${S=\emptyset}$ we get that $f(a_i) \geq(1-\gamma) \frac{1}{k'}f_H(O_H)$.  Therefore:
$$f(S_1) = f(a_1) \geq (1-\gamma)\frac{1}{k'}f_H(O_H) = (1-\gamma)\left(1 - \left (1-\frac{1}{k'}\right)\right)f_H(O_H).$$
We can now assume the claim holds for $i= l < k'$ and show that it holds for $i= l+1$:
\begin{align*}
f(S_{l+1})
&\geq (1-\gamma) \left ( \frac{1}{k'}  \left ( f_H(O_H) - f(S_{l})\right ) \right )+ f(S_{l}) \rmk{By (\ref{eq:induct})}\\
&> (1-\gamma) \left ( \left ( \frac{1}{k'} f_H(O_H) \right )  +  \left(1 - \frac{1}{k'}  \right) f(S_{l}) \right ) \rmk{$\delta>0$}\\
&\geq (1-\gamma) \left ( \frac{1}{k'} f_H(O_H) \right )  + (1-\gamma) \left(1 - \frac{1}{k'}  \right) \left(1-\left(1-\frac{1}{k'}\right)^{l}\right)f_{H}(O_H) \rmk{inductive hypothesis} \\
& = (1-\gamma)\left(1 - \left (1-\frac{1}{k'} \right)^{l+1}\right)  f_H(O_H)
\end{align*}
Note that for any $l > 1$ we have that $(1-{1}/{l})^{l} \leq {1}/{e}$, and thus:
\begin{align*}
f(S)
& = f(S_{k'}) \rmk{}\\
& \geq (1-1/e - \gamma)f_{H}(O_H) \rmk{by the induction}\\
& >  (1- 1/e - \epsilon/3)\texttt{OPT}_{H}.\rmk{$\gamma=\epsilon/6$} \qedhere
\end{align*}
\end{proof}

\begin{cor}\label{cor:constant}
Let $S$ be the set returned by \textsc{Smooth-Greedy} and $H$ be its smoothing set.  For any fixed $\epsilon>0$ and $k>{3\ell}/{\epsilon}$, we have that with probability at least $1-1/n^3$:
$$f(S\cup H) >\left (\frac{e-1}{2e-1-\epsilon}-2\epsilon \right) \texttt{OPT}.$$
\end{cor}

\begin{proof}
Let $O_{H} \in \textrm{argmax}_{T:|T| \leq k'}f_{H}(T)$.  From Lemma~\ref{lem:opt_h}, with probability at least $1-1/n^{3}$:
\begin{align}
f(S \cup H) > \left (1-\frac{1}{e}-\frac{\epsilon}{3}\right)f(O_H)\label{eq:constantfactor}
\end{align}
Let $O' \in \argmax_{T:|T|\leq k-|H|}f(T)$.  From submodularity and the fact that $k\geq 3\ell/\epsilon>|H|/\epsilon$ we get that $(1-\epsilon)\texttt{OPT} \leq f(O')$.  Putting everything together:
\begin{align*}
(1-\epsilon)\texttt{OPT}
& \leq f(O') \rmk{submodularity of $f$}\\
& \leq f(O_{H} \cup H) \rmk{monotonicity of $f$}\\
& \leq f(O_{H}) + f(H) \rmk{subadditivity of $f$}\\
& \leq \left (\frac{e}{e-1-\epsilon} \right) f(S \cup H) + f(H) \rmk{by~(\ref{eq:constantfactor})}\\
& \leq \left (\frac{2e-1-\epsilon}{e-1-\epsilon} \right)f(S \cup H).\rmk{monotonicity of $f$}
\end{align*}
Therefore $f(S\cup H) > \left (\frac{e-1}{2e-1-\epsilon}-2\epsilon \right)  \texttt{OPT}$ as required.
\end{proof}


\subsection*{Slick Greedy: Optimal Approximation for Sufficiently Large $k$}

As described in the main body of the paper, in \textsc{Slick-Greedy} we apply a slightly more general version of $\textsc{Smooth-Greedy}$ where in each iteration $i \in [1/\delta]$ the algorithm $\textsc{Smooth-Greedy}$ is initialized with the set of elements $R_i = \cup_{j \neq i} H_{j}$ and uses the smoothing set $H_i$.  $\textsc{Smooth-Greedy}$ from the previous section is a special case in which $R_i = \emptyset$.  As one might imagine, the guarantees from the previous section carry over, using the appropriate definitions.  

\subsubsection*{Generalizing guarantees of smooth greedy}  
To make the transition to the case in which $\textsc{Smooth-Greedy}$ is being initialized with $R_i$ of size $\ell/\delta - \ell$ and selects $k'' = k - |R_i| - |H_i| = k-\ell/\delta$ elements, we extend our definitions as follows.  For a given set $R_i$ used for initialization, it'll be convenient to consider the function $g_i(T) = f_{R_i}(T)$, and its smooth value $G_i(a) = \frac{1}{t}\sum_{i=1}^t g\left (S \cup \left (H_i\cup  a \right )  \right)$.  When the smoothing set is clear from context we will generally use $R,H,g,G$ instead of $R_i,H_i,g_i,G_i$.  The value of the optimal solution here is $\texttt{OPT[G]} = \max_{T:|T|\leq k''}g(T)$ where $k'' = k-|R|-|H|$.  We can then also define $\texttt{OPT[G]}_{H} =\max_{T:|T|\leq k''} g_{H}(T)$.  For a given set $S$ of elements selected by $\textsc{Smooth-Greedy}$ and $b^\star \in \argmax_{b \notin H }g_{S\cup H}(b)$, an \emph{$\epsilon$-relevant iteration} is one in which $g_{H\cup S}(b^\star) \geq \epsilon \texttt{OPT[G]}_{H}/k$ and $\texttt{OPT[G]}_{H}\geq \texttt{OPT[G]}/e$.

\paragraph{Lower bounding the marginal contribution in each iteration.}  We first show that when \textsc{Smooth-Greedy} is initialized with a set $R$ and run with smoothing set $H$, then in every {$\gamma$-relevant} iteration the element $a$ selected respects $g_{S}(a)\geq (1-\gamma)g_{H\cup H}(b^\star)$.  This claim is necessary for proving Lemma~\ref{lem:opt_hr} which shows the approximation guarantee of \textsc{Smooth-Greedy} in each iteration of \textsc{Slick-Greedy} as well as for proving guarantees of \textsc{Smooth-Compare} in Lemma~\ref{lem:boosting}.  

\begin{claim}\label{claim:itslate}
For a given set $R\subset N$, let $g(T) = f_{R}(T)$.  For any fixed $\gamma>0$ consider a $\gamma$-relevant iteration of \textsc{Smooth-Greedy} initialized with some set $R$ using smoothing set $H$ s.t. $H\cap R=\emptyset$, and let $S$ be the set of elements selected before the iteration.  If $a \in \argmax_{b\notin H} \widetilde{F}(R\cup S \cup b)$ then w.p.$\geq 1-1/n^4$:
$$ g_{S}(a) \geq (1-\gamma)g_{H\cup S}(b^\star)$$
\end{claim}

\begin{proof}
Let $G$ denote the smooth value function of $g$, i.e. $G(S\cup a) = \frac{1}{t}\sum_{H'\subset H}g(S\cup H'\cup a)$.  The proof in a chaining of four simple arguments.  Let $\lambda = \gamma^2/4k $ and $\alpha = \gamma\lambda/3k$. We show:

\begin{align*}
&1.& \widetilde{F}(R\cup S \cup a) & \geq & &&     &&\widetilde{F}(R\cup S \cup b^\star) &\hspace{0.4in}\implies\hspace{0.4in} &F_{R\cup S}(a) &\geq& &(1-\alpha)& F_{R\cup S}(b^\star) \\
&2.& F_{R\cup S}(a) &\geq & &(1-\alpha)& & & F_{R\cup S}(b^\star)&\hspace{0.4in} \implies \hspace{0.4in}& G(S\cup a) &\geq& &(1-\alpha)& G(S\cup b^\star) \\
&3.& G(S\cup a) &\geq & &(1-\alpha)& & &G(S\cup b^\star) &\hspace{0.4in}\implies\hspace{0.4in}& G_{S}(a) &\geq& &(1-\lambda)&G_{S}(b^\star) \\
&4.& G_{S}(a) & \geq & &(1-\lambda)& & & G_{S}(b^\star) &\hspace{0.4in}\implies\hspace{0.4in} & g_{S}(a) &\geq & &(1-\gamma) & g_{H\cup S}(b^\star)
\end{align*}
The above arguments can be justified as follows:
\begin{enumerate}
\item To see $\widetilde{F}(R\cup T \cup a) \geq \widetilde{F}(R\cup T \cup b^\star)$ implies $F_{R\cup T}(a) \geq (1-\alpha) F_{R\cup T}(b^\star)$, we invoke Claim~\ref{claim:realslimshady} on $S = R\cup T$.  To do so, since $\alpha \leq \gamma^3/24k^2$ for sufficiently large $n$ we need to verify: 
$$t>  \left (\frac{110k\log n}{\gamma\alpha} \right)^8 = \left (\frac{2640k^3\log n}{\gamma^3} \right)^8$$
In the case where $k\geq 2400\log n$ we use $\ell = 25\log n$ and thus $t = 2^{\ell} = n^{25}$ and the above inequality holds.  When $k < 2400\log n$ we use $\ell = 33\log\log n$ and thus $t=\log^{33}n$ and the above inequality holds in this case as well.  We therefore have the result w.p. $\geq 1-1/n^4$. 

\item Assuming that $F_{R\cup S}(a) \geq (1-\alpha) F_{R\cup S}(b^\star)$ we will show that $G(S\cup a) \geq (1-\alpha) G(S\cup b^\star)$:
\begin{align*}
	 	 & F_{R\cup S}(a) & \geq&  (1-\alpha) F_{R\cup S}(b^\star)\\
\implies	 & \frac{1}{t}\sum_{H'\subset H}f_{R\cup S} (H' \cup a) & \geq &  (1-\alpha)\frac{1}{t}\sum_{H'\subset H}f_{R\cup S} (H' \cup b^\star) \\
\implies & \frac{1}{t}\sum_{H'\subset H} \left ( f(R \cup S \cup H' \cup a) - f(R\cup S) \right ) & \geq &  (1-\alpha)\frac{1}{t}\sum_{H'\subset H} \left ( f (R\cup S \cup H' \cup b^\star) - f(R\cup S) \right )\\
\implies & \frac{1}{t}\sum_{H'\subset H} \left ( f(R \cup S \cup H' \cup a) - f(R) \right ) & \geq &  (1-\alpha)\frac{1}{t}\sum_{H'\subset H} \left ( f (R\cup S \cup H' \cup b^\star) - f(R) \right )\\
\implies &\frac{1}{t} \sum_{H'\subset H}  f_R(S\cup H' \cup a) 							   & \geq &  (1-\alpha)\frac{1}{t}\sum_{H'\subset H}  f_R (S\cup H' \cup b^\star)\\
\implies &\frac{1}{t} \sum_{H'\subset H}  g(S\cup H' \cup a) 							   	   & \geq &  (1-\alpha)\frac{1}{t}\sum_{H'\subset H}  g (S\cup H' \cup b^\star)\\
\implies &G(S\cup a) 							   	   & \geq &  (1-\alpha)G (S\cup b^\star)
\end{align*}

\item $G(S\cup a) \geq (1-\alpha) G(S\cup b^\star) \implies G_{S}(a) \geq (1-\lambda)G_{S}(b^\star)$:
We first argue $G_{S}(b^\star) > \frac{\gamma \texttt{OPT[G]}}{e\cdot k''}$:
\begin{align*}
G_S(b^\star) 
& = \frac{1}{t}\sum_{H' \subset H} \left ( g(S \cup b^\star \cup H') - g(S) \right ) \\
& \geq \frac{1}{t}\sum_{H' \subset H} \left ( g(S \cup b^\star \cup H') - g(S \cup H') \right ) \rmk{monotonicity of $g$}\\
& \geq \frac{1}{t}\sum_{H' \subset H} \left ( g(S \cup b^\star \cup H) - g(S \cup H) \right ) \rmk{submodularity of $g$}\\
& =  g(S \cup b^\star \cup H) - g(S \cup H) \\
& =  g_{S \cup H}( b^\star ) \\
& \geq \frac{\gamma}{k''}\texttt{OPT[G]}_{H} \rmk{$\gamma$-relevant iteration}\\
& > \frac{\gamma}{e\cdot k''}\texttt{OPT[G]} \rmk{$\texttt{OPT[G]}_{H} > \texttt{OPT[G]}/e$}
\end{align*}
Now, in a similar fashion to Claim~\ref{claim:newboundona}:
\begin{align*}
G_{S}(a) 
& = G(S\cup a) - G(S) \\
& \geq (1-\alpha)\left ( G(S\cup b^\star) - G(S)  \right ) - \alpha G(S)\\
& \geq (1-\alpha)\left ( G(S\cup b^\star) - G(S)  \right ) - \alpha \texttt{OPT[G]} \\
& \geq (1-\alpha)\left ( G(S\cup b^\star) - G(S)  \right ) - \alpha \frac{e\cdot k''}{\gamma}\cdot G_{S}(b^\star)\rmk{$G_{S}(b^\star) >\frac{\gamma \texttt{OPT[G]}}{e\cdot k''}$}\\
& = (1-\alpha)\left ( G_{S}(b^\star)  \right ) - \alpha \frac{e\cdot k''}{\gamma}\cdot G_{S}(b^\star)\\
& = \left (1-\alpha \left (1+\frac{e\cdot k''}{\gamma} \right) \right) G_{S}(b^\star)   \\
& = \left (1-\lambda \right) G_{S}(b^\star)  \rmk{$\alpha=\epsilon\lambda/3k$ and $k\geq k''+1$}
\end{align*}
\item $G_{S}(a) \geq (1-\lambda)G_{S}(b^\star) \implies g_{S}(a) \geq (1-\gamma)g_{H\cup S}(b^\star)$:
by direct application of Claim~\ref{claim:newboundona} \qedhere
\end{enumerate}
\end{proof}

\begin{definition} Given two disjoint sets $H$ and $R$, let $\texttt{OPT}_{H,R} = f(H \cup R \cup O_{H,R}) - f_R(H)$ where: 
$$O_{H,R} \in \argmax_{T:|T|\leq k - |H\cup R|} f(H \cup R \cup T).$$
\end{definition}
Notice that when $R=\emptyset$ we have that $O_{H, R} = O_{H} \in \argmax_{T:|T|\leq k - |H|} f_H(T)$ as defined in the previous subsection.  In that sense, the value of $O_{H, R}$ is that of the optimal solution evaluated on $f_{H}$ when initialized with $R$.  In the same way Lemma~\ref{lem:opt_h} shows \textsc{Smooth-Greedy} obtains a $1-1/e-\epsilon/3$ approximation to $\texttt{OPT}_H$, the following lemma shows that when \textsc{Smooth-Greedy} is initialized with $R$ it obtains the same guarantee against $\texttt{OPT}_{H, R}$.  Details are in Appendix~\ref{sec:slick_apx_appendix}.

\begin{lemma}\label{lem:opt_hr}
Let $S$ be the set returned by \textsc{Smooth-Greedy} that is initialized with a set ${R\subseteq N}$ and has $H$ as its smoothing set of size $\ell$, which is disjoint from $R$ and $S$.  Then, for any fixed $\epsilon>0$ when ${k\geq 3|H\cup R|/\epsilon}$ with probability of at least $1-1/n^3$ we have that:
$$f(R\cup S \cup H) \geq \left ( 1-1/e -\epsilon/3 \right) \texttt{OPT}_{H, R}.$$
\end{lemma}
\begin{proof}
Notice that the proof of Lemma~\ref{lem:opt_h} applies for the application of \textsc{Smooth-Greedy} on any submodular function $v$ where in every $\gamma$-relevant iteration $v_{S}(a) \geq (1-\gamma)v_{S\cup H}(b^\star)$ with probability $1-1/n^4$, for $\gamma \in \min\{1/e,\epsilon/6\}$, and $S$ being the elements added in the previous iteration.  From Claim~\ref{claim:itslate} we have that for any $\gamma$-relevant iteration $g_{S}(a) \geq (1-\gamma)g_{S\cup H}(b^\star)$ w.p. $\geq 1-1/n^4$.  We can therefore apply the exact same proof on $g$ and get: 
\begin{align}\label{eq:gh}
g(S\cup H) \geq (1-1/e -\epsilon/3)\texttt{OPT[G]}_{H}
\end{align}
Let $O_{H} \in \argmax_{T:|T|\leq k - |R\cup H|}g(T)$ and let $O_{H,R} \in \argmax_{T:|T|\leq k - |H\cup R|}f(H\cup R\cup T)$.  Observe that by definition of $g(X) = f_{R}(X)$ we have that: 
$$f(H \cup R \cup O_{H,R}) =  f(H \cup R \cup O_{H})$$ 
and thus from~(\ref{eq:gh}) we get:
\begin{align*}
f(R \cup S \cup H) - f(R)
& = f_{R}(S\cup H)\\
& = g(S \cup H) \\
& \geq (1-1/e - \epsilon/3)g_H(O_{H})\\
& \geq (1-1/e - \epsilon/3)\left ( g(O_{H} \cup H) - g(H) \right )\\
& = (1-1/e - \epsilon/3) \left ( f_R(O_{H} \cup H) - f_R(H) \right )\\
& \geq (1-1/e - \epsilon/3) \left ( f(R \cup O_{H} \cup H) - f(R) - f_R(H) \right )\\
& \geq (1-1/e - \epsilon/3) \left ( f(R \cup O_{H,R} \cup H) - f_R(H) \right ) - (1-1/e - \epsilon/3)f(R)
\end{align*}
and we therefore have that $f(R \cup S \cup H)\geq (1-1/e - \epsilon/3) \left ( f(R \cup O_{H,R} \cup H) - f_R(H) \right )$.
\end{proof}

We will instantiate the Lemma with $R=R_l$ and $H=H_l$ as discussed above: for any $i\in [1/\delta]$ we will define $R_i = \cup_{j\neq i}H_j$ and use the index $l$ to denote the smoothing set in $\{H_{i}\}_{i=1}^{1/\delta}$ which has the least marginal contribution to the rest, i.e. $H_l = \argmin_{i\in [1/\delta]} f_{R_i}(H_i)$.  We first show that the iteration of \textsc{Slick-Greedy} on $l$ finds a solution arbitrarily close to $1-1/e$ for sufficiently large $k$.

\begin{lemma*}[\ref{claim:fr}]
Let $S_l$ be the set returned by \textsc{Smooth-Greedy} that is initialized with $R_l$ and $H_l$ its smoothing set.  Then, for any fixed $\epsilon>0$ when $k\geq 36\ell/\epsilon^2$ with probability of at least $1-1/n^3$ we have:
$$ f(S_l \cup H_l)  \geq (1-1/e -2\epsilon/3) \texttt{OPT} $$
\end{lemma*}
\begin{proof}
To ease notation, let $R=R_l$, $H=H_l$, and $O=O_l$ where $O_l$ is the solution which maximizes $f(H\cup R \cup T)$ over all subsets $T$ of size at most $k-|H\cup R|$.  Let $\beta = |H\cup R|/k$.  Notice that by submodularity we have that:
\begin{align}
f(H \cup R \cup O)  \geq\left  (1-\frac{|H\cup R|}{k} \right )\texttt{OPT} = (1-\beta)\texttt{OPT}\label{eq:anotherboundonopt}
\end{align}
Notice also that by the minimality of $H=H_l$ and submodularity we have that $f_R(H)\leq \delta f(H\cup R)$.  Recall also that $\delta = \epsilon/6$ and notice that whenever $k\geq \ell/\delta^2 = 36\ell/\epsilon^2$ we have that $\beta<\delta$ and hence $\beta+\delta < \epsilon/3$.  Therefore, by application of Lemma~\ref{lem:opt_hr} we get that with probability $1-1/n^3$:
\begin{align*}
f(S\cup R \cup H)
& \geq \left (1-\frac{1}{e} -\frac{\epsilon}{3} \right ) \texttt{OPT}_{H, R} \rmk{by Lemma \ref{lem:opt_hr} }\\
& = \left (1-\frac{1}{e} -\frac{\epsilon}{3} \right ) \left (   f(H\cup R \cup O) - f_R(H)     \right ) \rmk{by definition}\\
& \geq \left (1-\frac{1}{e} -\frac{\epsilon}{3} \right ) \left (   f(H\cup R \cup O) -\delta \cdot f(H\cup R)     \right ) \rmk{$f_R(H)\leq \delta f(H\cup R)$}\\
& \geq \left (1-\frac{1}{e} -\frac{\epsilon}{3} \right ) \left ( (1-\delta)  f(H\cup R \cup O)  \right ) \rmk{monotonicity of $f$}\\
& \geq \left (1-\frac{1}{e} -\frac{\epsilon}{3} -\delta \right ) \left ( f(H\cup R \cup O)  \right ) \\
& \geq \left (1-\frac{1}{e} -\frac{\epsilon}{3} -\delta \right ) \left ( 1- \beta \right ) \texttt{OPT} \rmk{by~(\ref{eq:anotherboundonopt})}\\
& \geq \left (1-\frac{1}{e} -\frac{2\epsilon}{3} \right )\texttt{OPT}.\rmk{$\beta+\delta<\epsilon/3$}\qedhere
\end{align*}
\end{proof}

\subsubsection*{The smooth comparison procedure}\label{sec:smooth_compare}
\begin{lemma*}[\ref{lem:boosting}]
Assume $k \ge 96 \ell /\epsilon^2$.  Let $T_i$ be the set that won the $\textsc{Smooth-Compare}$ tournament.  Then, with probability at least $1-1/n^2$:
$$f(T_i)\geq \left (1-\frac{\epsilon}{3} \right )\min \left \{ \left(1-\frac{1}{e}-\frac{2\epsilon}{3}\right)\texttt{OPT},\max_{j \in [1/\delta]}f(T_j)\right \}$$
\end{lemma*}

The proof of the lemma uses the following two claims.  

\begin{claim}\label{clm:slick}
Let $T_i = S_i \cup {H}_i$ and $T_j = S_j \cup {H}_j$ be two sets that are compared by \textsc{Smooth-Compare}, and suppose that (i)$f(T_i) \geq (1+2\beta)f(T_{j})$ where $\beta = |H_{ij}|/k''$ and $k''=k-\ell/\delta$, and (ii) $f(T_j)<(1-1/e -2\epsilon/3)\texttt{OPT}$ for any $\epsilon \geq 3(1- k''/k)/2$.  Then, for any set $H'_{ij}\subseteq H_{ij}$ w.p. $\geq 1-1/n^3$:
$$f(T_i \cup H'_{ij}) \geq f(T_j \cup H'_{ij}).$$
\end{claim}

\begin{proof}
Recall that ${H}_{ij} \cap \Big (T_i \cup T_j \Big ) =\emptyset$.  We will argue that assuming $f(T_j)<(1-1/e)\texttt{OPT}$, the fact that every element in $H'_{ij}$ was a candidate for selection by \textsc{Smooth-Greedy} and wasn't selected, implies that w.h.p. either (i) $f(T_j)$ is arbitrarily close to $1-1/e$ (in which case we wouldn't mind that if it wins the comparison) or (ii) the marginal contribution of $H'_{ij}$ to $T_j$ is bounded from above by $2\beta f(T_j)$ which suffices since then we get:
\begin{equation*}
f(T_j \cup H'_{ij}) 
 = f(T_j)+f_{T_j}(H'_{ij}) 
 \leq (1+2\beta)f(T_j)  
 < f(T_i)
 \leq f(T_{i} \cup H'_{ij})
 \end{equation*}
To prove this, consider the instantiation of $\textsc{Smooth-Greedy}$ initialized with $R_j$ with smoothing set $H_j$, and let $S$ be the set selected after its $k'' = k-|R_j| - |H_j|$ iterations.  Recall that $S_j=R_j \cup S$ and that $T_j = S_j \cup H_j$.  To ease notation let $R = R_j$ and $H=H_j$.  

We will first prove the statement in the case that the iteration is $\gamma$-relevant for $\gamma = 1/4$.   For every iteration $r \in [k'']$ let $S{(r)}$ be the set of elements selected in the previous iterations and $a({r})$ be the element added to the solution at that stage by \textsc{Smooth-Greedy}.  From Claim~\ref{claim:itslate} we know that since $a(r) \in \argmax_{b} \widetilde{F}(R\cup S(r) \cup b)$ and the size of the smoothing neighborhood $t$ is sufficiently large then w.p. $\geq 1-1/n^4$:
$$ g_{S(r)}(a(r)) \geq (1-\gamma)\max_{b\notin H}g_{H\cup S(r)}(b)$$
We therefore have that:
\begin{align*}
g(S) 
& = \sum_{r=1}^{k''}g_{S{(r)}}(a_{r}) \\
& \geq \sum_{r=1}^{k''}(1-\gamma) \max_{b \notin H} g_{S(r) \cup H  }(b) \\
& \geq \sum_{r=1}^{k''}(1-\gamma) \max_{b \notin H} g_{S \cup H  }(b) \\
& = k'' (1-\gamma) \max_{b \notin H} g_{S \cup H }(b)  \\
& \geq k'' (1-\gamma) \max_{h \in H'_{ij}}g_{S \cup H }(h)  \\
& \geq \frac{k'' (1-\gamma) }{|H'_{ij}|} g_{S \cup H }(H'_{ij})  \\
& \geq \frac{(1-\gamma) k''}{\ell} g_{S \cup H }(H'_{ij})  \\
\end{align*} 
Since $g(T) = f_{R}(T)$ and $\gamma = 1/4$ this implies:
$$ f(R\cup S) -f(R)> \frac{k''}{2\ell}  \left ( f(R \cup H \cup H'_{ij}) - f(R\cup S) \right )$$
Since $T_j = R_j \cup S \cup H_j = R\cup  S \cup H$ we get:
$$f_{T_j}(H'_{ij}) < \frac{2\ell}{k''}f(T_j) = 2\beta f(T_j).  $$
If the iteration is not $\gamma$-relevant, assume first that $e\cdot \texttt{OPT[G]}_{H} \geq \texttt{OPT[G]}$.  
In this case, let $O_H = \argmax_{T:|T|\leq k''}g_{H}(T)$.  Notice that the fact that iteration is not relevant in this case says that there is an iteration $r$ for which $\max_{b \notin H}g_{H \cup S(r)}(b) < \gamma \texttt{OPT[G]}_{H}/k$ and from submodularity of $g$ since $S(r) \subseteq S$ we get $\max_{b \notin H}g_{H \cup S}(b) < \gamma \texttt{OPT[G]}_{H}/k$.  Thus:

\begin{align*}
g_{H\cup S}(O_{H})
& \leq k''\cdot g_{H\cup S}(b^\star) \\
& \leq k''\cdot \frac{\gamma \texttt{OPT[G]}_{H}}{k}\\
& < \gamma \texttt{OPT[G]}_{H}
\end{align*}
which implies:
\begin{align*}
g(H\cup S) 
& > g(O_{H} \cup H \cup S) - \gamma \texttt{OPT[G]}_{H} \\
& \geq g_H(O_{H})- \gamma \texttt{OPT[G]}_{H}\\
& = (1-\gamma)\texttt{OPT[G]}_{H}
\end{align*}
Using this bound we get:
\begin{align*}
g_{H\cup S}(H'_{ij}) 
& \leq |H'_{ij}| \max_{h\in H'_{ij}}g_{H\cup S}(h) \\
& \leq |H'_{ij}| \max_{b\notin H}g_{H\cup S}(b) \\
& \leq |H'_{ij}| \frac{\gamma}{k}\texttt{OPT[G]}_{H} \\
& < \frac{\gamma \ell}{k(1-\gamma)}g(H\cup S)
\end{align*}
Again, as before for $\delta = 1/4$ we get that in this case:
 $$f_{T_j}(H'_{ij}) < \frac{2\ell}{k''}f(T_j) = 2\beta f(T_j)  $$
Lastly, it remains to show that if if the iteration is not $\gamma$-relevant because $e\cdot \texttt{OPT[G]}_{H} < \texttt{OPT}[G]$, we get a contradiction to our assumption that $f(T_{j}) < (1-1/e-2\epsilon/3)\texttt{OPT}$.  To see this, let $O \in \argmax_{T:|T|\leq k''}g(T)$, and notice that:
\begin{align*}
g(H \cup O_H) - g(H) < \frac{g(O)}{e}
\end{align*}   
hence:
\begin{align*}
f(R\cup H) - f(R) 
& = g(H) \\
& > g(H\cup O_H) - \frac{g(O)}{e}\\
& \geq  \left (1-\frac{1}{e}\right )g(O)\\
& \geq \left (1-\frac{1}{e} \right )\left ( f(R\cup O)\right ) - f(R)
\end{align*}
We therefore get that $f(T_j) \geq f(R\cup H) > (1-1/e)f(O)$.  Notice that since $|O| = k''$ and ${k''/k \geq (1-2\epsilon/3)}$, submodularity  implies $f(T_j) \geq (1-1/e -2\epsilon/3)\texttt{OPT}$, a contradiction.
\end{proof}

\begin{claim}\label{lem:tournament}
For $k \ge 96 \ell /\epsilon^2$ suppose that $f(T_i) \geq (1+\epsilon\delta/3)f(T_{j})$ 
and that ${f(T_j)\leq(1-1/e-2\epsilon/3)\texttt{OPT}}$.  Then, $T_i$ wins in the smooth comparison procedure w.p. $\geq 1-2/n^3$.
\end{claim}

\begin{proof}
Let $\beta = |H_{ij}|/k''$ where $k'' = k-(|H_{ij}|+|R_{i}|)$.  Since we assume that $k \ge 96 \ell$ and $\delta=\epsilon/6$ this implies that $2\beta < \epsilon^2/45$.  We therefore have:
$$f(T_i) > \left (1+\frac{\epsilon\delta}{3}\right )f(T_j) = \left (1+\frac{\epsilon^2}{18} \right )f(T_j)  > \left (1+\frac{\epsilon^2}{45} \right)^2 f(T_j) > \left (1+ 2\beta\right )^2f(T_j)$$
From Claim~\ref{clm:slick} this implies that for any $H'_{ij} \subseteq H_{ij}$ we have that with probability at least $1-1/n^3$:
$$f(T_j \cup H'_{ij}) \leq (1+2\beta) f(T_j \cup H'_{ij}) $$
We will condition on this event as well as the event that the maximal value obtained throughout the iterations of the algorithm is $\nu_{\max}$ and minimal value is $\nu_{\min}$, and that $\nu_{\max}/\nu_{\min}\leq n^{\tau}$ for some constant $\tau>0$.  
\begin{align*}
\Pr & \left [  \widetilde{f}(T_i \cup H'_{ij}) \geq \widetilde{f}(T_j \cup H'_{ij})  \Big | f(T_i) \geq \left(1+\frac{\epsilon\delta}{3}\right) f(T_j)    \right ]\\
=  \Pr & \left [  \xi_{i}{f}(T_i \cup H'_{ij}) \geq \xi_j {f}(T_j \cup H'_{ij})  \Big  | f(T_i) \geq \left (1+\frac{\epsilon\delta}{3}\right)   f(T_j)    \right ]\\
 >  \Pr & \left [  (1+2\beta) \cdot \frac{\xi_{i}}{\xi_j} \geq 1  \    \right ]\\
 \geq   \frac{1}{2} \ &+ \frac{1}{2\log_{1+2\beta} (\frac{\nu_{max}}{\nu_{min}})}
\end{align*}
The last inequality follows from a discretization argument:  Consider the $m \in O(\log n)$ intervals, where the $i$'th interval is $[\nu_{\min}(1 + 2\beta)^i , \nu_{\min}(1 + 2\beta)^{i+1}]$, and $i$ ranges from $0$ to $\log_{1 +2\beta}(\frac{\nu_{max}}{\nu_{\min}})$.  Due to symmetry of $\xi_i$ and $\xi_j$, the likelihood of $\xi_i$ falling in the same or higher interval than $\xi_j$ is: 
$$\frac{\sum_{i=1}^m i}{m^2} = \frac{1}{2}+\frac{1}{2m} = \frac{1}{2} + \frac{1}{2\log_{1+2\beta} (\frac{\nu_{max}}{\nu_{min}})} = \frac{1}{2}+\frac{1}{2\tau \log_{1+2\beta}n}$$
%
Applying a Chernoff bound, for any constants $\epsilon,\delta>0$, s.t. $\epsilon\delta/8 > 1+2\beta$, and $\nu_{\max}/\nu_{\min}\leq n^{\tau}$ for some constant $\tau>0$, we get that $T_i$ is chosen with probability at least $1-\exp(-\Omega(n/\log(n)))$, conditioned on $\nu_{\max}/\nu_{\min}<n^{\tau}$ which by Lemma~\ref{lem:upperBoundNoise} occurs with probability $1-\exp(-\Omega(n^\alpha))$ for some constant $\alpha>0$.  For sufficiently large $n$, $T_i$ therefore wins w.p. at least $1-2/n^3$.  
\end{proof}

\begin{proof}[\textbf{Proof of Lemma~\ref{lem:boosting}}]
Since $\forall i,j \in [1/\delta]$ $\textsc{Smooth-Compare}(\{T_i,T_j\},H_{ij})$ returns $T_i$ as long as $f(T_i) \geq (1-\epsilon\delta/3)f(T_j)$ and $f(T_j)<(1-1/e-2\epsilon/3)\texttt{OPT}$, and \textsc{Smooth-Compare} is called $1/\delta$ times we get: 
\begin{align*}
f(T_i) 
& \geq &\left (1-\frac{\epsilon\delta}{3} \right )&^{1/\delta} & \times & &\min&\left \{\left(1-\frac{1}{e}-\frac{2\epsilon}{3}\right)\texttt{OPT},\max_{j \in [1/\delta]}f(T_j) \right \} \hspace{0.6in}\\
& \geq &\left (1-\frac{\epsilon}{3} \right )& 	&\times&				   &\min&\left \{ \left(1-\frac{1}{e}-\frac{2\epsilon}{3}\right)\texttt{OPT},\max_{j \in [1/\delta]}f(T_j)\right \}.\hspace{0.6in}
\qedhere
\end{align*}
\end{proof}

\newpage

\section{Optimization for Small $k$}\label{sec:appendix_smallk}

\subsection*{Smoothing Guarantees}

\begin{lemma}[\ref{lem:avg}]
For any $\epsilon>0$ and any set $S \subset N$, let $A^\star \in \arg\max_{A:|A|=1/\epsilon}f_{S}(A)$.  Then:
$$ \left( 1- \epsilon\right)f_{S}(A^\star)\leq \score_{S}(A^\star)\leq f_{S}(A^{\star}).$$
\end{lemma}

\begin{proof}\label{lem:avg_proof}
By the maximality of $A^{\star}$ we have that $f(A^\star) \geq f(A^\star_{ij})$ for any $i,j$ since $A^{\star}_{ij}$ is generated by replacing $a_i \in A^{\star}$ with $a_{j} \notin A^{\star} \cup S$.   Therefore, the average of all $A_{ij}$s is upper bounded by $f_{S}(A^\star)$.

For the lower bound, let $c = 1/\epsilon$ and consider some arbitrary ordering on $a_{1},\ldots,a_{c} \in A^{\star}$.  Define $A_{\text{-}i} = A \setminus \{a_i\}$.  From the diminishing returns property we get that for any $i \in [c]$:
\begin{align*}
f_{S \cup A^{\star}_{\text{-}i}}(a_i) && = 	&&f(S \cup A^{\star}_{\text{-}i} \cup a_i) &&- &&	    f(S\cup A^{\star}_{\text{-}i})		&&\\
									&& \leq 	&&f(S \cup \{a_{1}\ldots,a_i\})			 && -  &&f(S \cup \{a_{1},\ldots,a_{i-1}\})&&
\end{align*}
Thus:
\begin{align}
\sum_{i=1}^{c} f_{S\cup A^\star_{\text{-}i}}(a_i)
\leq \sum_{i=1}^{c} \left ( f(S \cup \{a_{1}\ldots,a_i\}) - f(S \cup \{a_{1},\ldots,a_{i-1}\}) \right )
= f_{S}(A^{\star})\label{getlucky}
\end{align}
By summing over all $A^\star_{\text{-}i}$ we get the desired bound:
\begin{align*}
\score_{S}(A^\star)
& = \frac{1}{c(n-c-|S|)}\sum_{j=1}^{n-c-|S|}\sum_{i=1}^{c}f_S(A^{\star}_{ij})\\
& \geq \frac{1}{c}\sum_{i=1}^{c}f_S(A^{\star}_{\text{-}i})\rmk{monotonicity, since $A^\star_{\text{-}i}\subset A^\star_{ij}$}\\
& = \frac{1}{c}\sum_{i=1}^{c}\left ( f_S(A^{\star}_{\text{-}i} \cup a_i) - f_{S \cup A^\star_{\text{-}i}}(a_i) \right )\\
& =  \frac{1}{c}\sum_{i=1}^{c} f_S(A^{\star}) - \frac{1}{c}\sum_{i=1}^{c}f_{S \cup A^\star_{\text{-}i}}(a_i) \\
& \geq f_{S}(A^{\star})  - \frac{1}{c}f_{S}(A^{\star})\rmk{by~(\ref{getlucky})}\\
& =  \left (1- \frac{1}{c} \right)f_{S}(A^{\star})\rmk{}\\
& =  \left (1- \epsilon \right)f_{S}(A^{\star}).\rmk{}\qedhere
\end{align*}
\end{proof}

\paragraph{The smoothing lemma.}  The rest of this subsection is devoted to proving the following important lemma.  Intuitively, this lemma implies that at every iteration of \textsc{SM-Greedy} we identify the bundle which nearly maximizes the mean marginal contribution.

\begin{lemma*}[\ref{algoworks}]
Let $A \in \argmax_{B:|B|=c} \widetilde{\score}(S\cup B)$ where $c \geq \frac{16}{\epsilon}$, and assume that the iteration is $\frac{\epsilon}{4}$-significant.  Then, with probability at least $1-e^{-\Omega(n^{{1}/{10}})}$ we have that:
$$\score_{S}(A) \geq (1-\epsilon) \max_{B:|B|=c}F_S(B).$$
\end{lemma*}


\paragraph{Smoothing neighborhoods.}  The proof uses the smoothing arguments developed in Section~\ref{sec:appendix_smoothing}.  Recall that for a given set of elements $A\subseteq N$ a \emph{smoothing function} is a method which assigns $A$ a family of sets $\mathcal{H}(A)$ called the \emph{smoothing neighborhood}.  For a given function ${f:2^{N} \to \mathbb{R}}$, $A,S \subseteq N$, and smoothing neighborhood $\mathcal{H}(A)$ we define:
\begin{align*}
\hspace{1.6in}&&(1)\hspace{0.4in}&& \mathbf{F}_{S}(A)   			 	 &&:= &&\mathbb{E}_{X\in \mathcal{H}(A)}&& 	[ && f_{S}(X) && ];&&\hspace{2in}\\ 
\hspace{1.6in}&&(2)\hspace{0.4in}&& \mathbf{F}(S\cup A) 			 	 && := &&\mathbb{E}_{X\in \mathcal{H}(A)}&&   [ && f(S \cup X) &&  ];&&\hspace{2in}\\
\hspace{1.6in}&&(3)\hspace{0.4in}&& \mathbf{\widetilde{F}}(S\cup A) 	 &&:= &&\mathbb{E}_{X\in \mathcal{H}(A)} &&	 [ &&\widetilde{f}(S \cup X) &&  ].&&\hspace{2in}
\end{align*}
Note that $\mathbf{F}(A) \neq {F}(A)$.  In particular, as discussed above, we do not apply smoothing on the noisy version of $F$ directly, but rather on the noisy version of the function $\mathbf{F}$ which is applied on $\Ai := A \setminus \{a_i\}$, for all $i \in [c]$:
$$ \widetilde{\mathbf{F}}(S \cup \Ai) := \frac{1}{n - c - |S|}\sum_{j \notin S\cup A}\widetilde{f}(S\cup \Ai \cup \{a_j\})$$
Notice that the smoothing arguments then apply to $F$ since:
$$\widetilde{F}(S \cup A)  = \frac{1}{c}\sum_{i=1}^{c}\widetilde{\mathbf{F}}(S \cup \Ai)$$
In our case, for every $\Ai$, its smoothing neighborhood is:
$$\mathcal{H}(\Ai) = \left \{ \Ai \cup \{a_j\}  \ : \ j \notin S\cup A  \right \}$$
Throughout the rest of this section we will use $t$ to denote the number of sets in a smoothing neighborhood of $\mathcal{H}(\Ai)$.  Note that for every $i \in [c]$ the size of a smoothing neighborhood is:
$$t= |\mathcal{H}(\Ai)| = |N \cup (S \setminus A)| = n- c - |S|  \in O(n).$$

\paragraph{Smoothing in the sampled mean method.}  In order to apply Lemma~\ref{lem:noisysmoothing_lowerbound} in a meaningful way we need to bound the variation of the neighborhoods $\mathcal{H}(A^\star_{\text{-}i})$.  To do so, we use the next claim which essentially bounds the variation of the smoothing neighborhoods $\mathcal{H}(A^\star_{\text{-}i})$, of \emph{almost} all $A^\star_{\text{-}i}$.  

\begin{claim}\label{cl:bound_on_var}
Let $A^\star \in \argmax_{B:|B|=c} f_{S}(B)$, $c \geq 4/\epsilon$.  Then:
$$\frac{1}{c}\sum_{i=1}^c\max \left \{0, 1-2v_S(\mathcal{H}(A^\star_{\text{-}i}))\cdot t^{-1/4}  \right \} \mathbf{F}_{S}(A^\star_{\text{-}i})\geq \left (1- \epsilon \right)f_{S}(A^\star).$$
\end{claim}

\begin{proof}
To bound the average variation of the sets $\{A^{\star}_{\text{-}i}\}_{i=1}^c$ we argue that at most one set $A^\star_{\text{-}i}$ will be s.t. $f_{S}(\Ai^\star) < f_{S}(A^\star)/2$.  To see this, assume for purpose of contradiction there are $\Ai^\star$ and $A_{\text{-}j}^\star$ for which $f_S(\Ai^{\star})\leq f_S(A_{\text{-}j}^{\star})< f_{S}(A^{\star})/2$, then since $A^\star = A^\star_{\text{-}i} \cup A^\star_{\text{-}j}$ we get a contradiction:
$$
f_{S}(A^\star) =  f_{S}(\Ai^{\star} \cup A_{-j}^\star) \leq f_{S}(\Ai^{\star})+ f_{S}(A_{-j}^\star) < 2\cdot \frac{f_{S}(A^{\star})}{2} =  f_{S}(A^{\star}).$$
We therefore have at least $c-1$ sets s.t. each $\Ai^\star$ respects $f_{S}(A^\star_{\text{-}i}) \geq f_{S}(A^\star)/2$.  Call these sets \emph{bounded}.  For any such bounded set $A^\star_{\text{-}i}$, since $A^\star_{\text{-}i} \subset A^\star_{ij}$ for any $j \notin S \cup A^{\star}$, monotonicity implies:
$$\min_{A^{\star}_{ij} \in \mathcal{H}(A^\star_{\text{-}i})} f_{S}(A^{\star}_{ij}) \geq  \frac{f_{S}(A^{\star})}{2}$$
For a given set $A^\star_{\text{-}i}$ note that for every $j$, every set $A_{ij} \in \mathcal{H}(A^{\star}_i)$ respects $f_S(A^\star_{ij}) \leq f_S(A^{\star})$ due to the maximality of $A^\star$.  Thus for any bounded set $A^\star_{\text{-}i}$:
$$ v_S(\mathcal{H}(A^\star_{\text{-}i})) = \frac{\max_{A^{\star}_{ij} \in \mathcal{H}(A^\star_i)} f_{S}(A^{\star}_{ij})}{\min_{A^{\star}_{ij} \in \mathcal{H}(A^\star_i)} f_{S}(A^{\star}_{ij})} \leq \frac{f_{S}(A^\star)}{f_{S}(A^\star)/2} = 2$$
Let $l$ be the index of the set $A^\star_{\text{-}i}$ with the lowest value $f_{S}(A^\star_{\text{-}i})$.  Our discussion above implies that this is the only set whose variation may not be bounded from above by $2$.  Assume $n$ sufficiently large s.t. $t\geq 2^{12}/\epsilon^4$.  We therefore get:

\begin{align}
\frac{1}{c}\sum_{i=1}^c \left ( \max \{0,1-2v_S(\mathcal{H}(A^\star_{\text{-}i}))t^{-\frac{1}{4}}  \} \right )\mathbf{F}_{S}(A^\star_{\text{-}i})
& \geq \frac{1}{c}\sum_{i \neq l}\left ( \max \{ 0,1-2v_S(\mathcal{H}(A^\star_{\text{-}i})) t^{-\frac{1}{4}}  \} \right )\mathbf{F}_{S}(A^\star_{\text{-}i}) \label{eq:b1}\\
& \geq \frac{1}{c}\sum_{i \neq l} \left(1- 4t^{-\frac{1}{4}} \right)  \mathbf{F}_{S}(A^\star_{\text{-}i}) \label{eq:b2}\\
& \geq \frac{1}{c}\sum_{i \neq l} \left(1- 4t^{-\frac{1}{4}} \right) f_{S}(A^\star_{\text{-}i}) \label{eq:b3}\\
& \geq \left(1- 4t^{-\frac{1}{4}} \right)\frac{1}{c}\left( \sum_{i=1}^{c} f_{S}(A^\star_{\text{-}i}) - f_{S}(A^\star_{-l}) \right) \label{eq:b4}\\
& \geq \left(1- 4t^{-\frac{1}{4}} \right)\frac{1}{c}\left((c-1)f_{S}(A^\star) - f_{S}(A^\star_{-l}) \right) \label{eq:b5}\\
& \geq \left(1- 4t^{-\frac{1}{4}} \right)\frac{1}{c}\left((c-1)f_{S}(A^\star) - f_{S}(A^\star) \right) \label{eq:b6}\\
& \geq \left(1- 4t^{-\frac{1}{4}}\right)\left (\frac{c-2}{c} \right ) f_{S}(A^\star) \label{eq:b7}\\
& \geq \left(\frac{c-2}{c} - 4t^{-\frac{1}{4}} \right ) f_{S}(A^\star) \label{eq:b8}\\
& \geq \left ( 1-\epsilon   \right )f_{S}(A^\star)\label{eq:b9}
\end{align}
The inequality~(\ref{eq:b2}) is justified by the bound we established on bounded sets;~(\ref{eq:b3}) is due to monotonicity of $f_{S}$, since $F_{S}(A^\star_{\text{-}i})$ is an average of the marginal contribution over all possible $A^\star_{ij}$, which is a superset of $A^\star_{\text{-}i}$;~(\ref{eq:b5}) is due to an argument in the proof of Lemma~\ref{lem:avg};~(\ref{eq:b6}) is due to the optimality of $A^\star$;~(\ref{eq:b9}) is due to the assumption on the parameters in the statement of the claim.
\end{proof}

\begin{proof}[\textbf{Proof of Lemma~\ref{algoworks}}]
Let $A^{\star} = \arg \max_{A:|A|=c} f_S(A)$ and let $B:|B|=c$ be such that $\score_S(B) < (1-\epsilon) \score_{S}(A^\star)$.  We will apply the smoothing arguments and show that with high probability
$$\widetilde{\score}(S\cup A^\star) > \widetilde{\score}(S\cup B).$$
By taking a union bound over all possible $O(n^c)$ sets $B$ we will then conclude that the set whose smooth noisy contribution is largest must have smooth contribution at least factor of $(1-\epsilon)$ from that of $A^\star$, with high probability.

We will denote $\epsilon_{1} = \epsilon$ and $\epsilon_{2} = \epsilon/4$.  Notice that the conditions of Claim~\ref{cl:bound_on_var} are met with $\epsilon_2$ and that the iteration is $\epsilon_2$-significant, which from submodularity implies $f_{S}(A^\star) \geq \epsilon_2\cdot f(S)/k$.

For a set $B_{\text{-}i} \subset B$, using Lemma \ref{lem:boundonb}, for $t = n-c-|S|$, when $\omega$ denotes the highest realized value of a noise multiplier, we know that for $\lambda \in [0,1)$ with probability $1-\exp\left ({-\Omega(\lambda^2 t^{1/4}/\omega)}\right )$:
\begin{align*}
\widetilde{\score}(S \cup B) & = \frac{1}{c} \sum_{i} \widetilde{\mathbf{F}}(S \cup B_{\text{-}i}) \\
& <  \frac{1}{c} \sum_i(1+\lambda) \mu \cdot \left (f(S)+ \mathbf{F}_{S}(B_{\text{-}i}) + 3t^{-1/4} \max_{B_{ij} \in \{\mathcal{H}(B_{\text{-}i})\}}f_S(B_{ij}) \right )\\
& \leq (1+\lambda) \mu \cdot \left (f(S) + 3t^{-1/4}\max_{B_{ij} \in \{\cup_{i \in [c]}\mathcal{H}(B_{\text{-}i})\}}f_S(B_{ij}) + \frac{1}{c}\sum_{i=1}^{c} \mathbf{F}_{S}(B_{\text{-}i}) \right )\\
& \leq  (1+\lambda) \mu \cdot \left (f(S) + 3t^{-1/4} f_{S}(A^\star) + \frac{1}{c}\sum_{i=1}^{c} \mathbf{F}_{S}(B_{\text{-}i}) \right )\\
& \leq  (1+\lambda) \mu \cdot \left (f(S)+ 3t^{-1/4} f_{S}(A^\star) + \score(S \cup B) \right )\\
& \leq  (1+\lambda) \mu \cdot \left (f(S) + 3t^{-1/4} f_{S}(A^\star) + (1-\epsilon_1)\score(S \cup A^\star) \right )\\
& \leq  (1+\lambda) \mu \cdot \left (f(S) + 3t^{-1/4} f_{S}(A^\star) + (1-\epsilon_1)f_{S}(A^\star) \right )\\
& =  (1+\lambda)  \mu \cdot \left ( f(S) + f_S(A^\star) \left (3t^{-1/4} + (1-\epsilon_1) \right ) \right)
\end{align*}
We now need to argue that $\widetilde{\score}(S \cup A^\star)$ is sufficiently large to beat $\widetilde{\score}(S \cup B)$.  Assuming $n$ is sufficiently large s.t. $t\geq 2^{20}/\epsilon^4$, from lemmas~\ref{lem:noisysmoothing_lowerbound} and \ref{cl:bound_on_var} we know that for $\lambda \in [0,1)$ w.p. $1-e^{-\Omega(\lambda^2t^{1/4}/\omega)}$:
\begin{align*}
\widetilde{\score}(S \cup A^\star)
& = \frac{1}{c}\sum_{i=1}^{c}\widetilde{\mathbf{F}}(S \cup A^\star)\\
& > (1-\lambda)\mu \cdot \left ( f(S) +  \frac{1}{c}\sum_{i=1}^{c} \left (1-2v(\mathcal{H}(A^\star_i))\cdot t^{-1/4} \right) \cdot \mathbf{F}_{S}(A^\star) \right )\\
& > (1-\lambda)\mu \cdot \left ( f(S) + (1-\epsilon_2)f_{S}(A^\star) \right )
\end{align*}
We therefore get that:
\begin{align*}
\widetilde{\score}(S \cup A^\star) - \widetilde{\score}(S \cup B)
& \geq \mu \left ( (1-\lambda) \cdot \left ( f(S) + (1-\epsilon_2)f_{S}(A^\star) \right ) -  (1+\lambda)\cdot \left ( f(S) + f_S(A^\star) \left (3t^{-1/4} + (1-\epsilon_1) \right ) \right)  \right )\\
& \geq \mu \left( (1-\lambda)(1-\epsilon_2)f_{S}(A^\star) -  2\lambda f(S) - (1+\lambda)\left (3t^{-1/4} + (1-\epsilon_1) \right )f_S(A^\star) \right )\\
& \geq \mu \left( (1-\lambda)(1-\epsilon_2)f_{S}(A^\star) -  \frac{2\lambda k}{\epsilon_{2}}f_{S}(A^\star)  - (1+\lambda)\left (3t^{-1/4} + (1-\epsilon_1) \right )f_S(A^\star) \right )\\
& \geq \mu\cdot f_{S}(A^\star) \left ((1-\lambda)(1-\epsilon_2)-  \frac{2\lambda k}{\epsilon_{2}} - (1+\lambda)\left (3t^{-1/4} + (1-\epsilon_1) \right ) \right )\\
& \geq \mu\cdot f_{S}(A^\star) \left ((1-\lambda)(1-\epsilon_2)-  \frac{2\lambda k}{\epsilon_{2}}  - (1+\lambda)\left (\epsilon_{4} + (1-\epsilon_1) \right ) \right )\\
& \geq \mu\cdot f_{S}(A^\star) \left (1-\lambda -\epsilon_2 -   \frac{2\lambda k}{\epsilon_{2}}  - \epsilon_{2} - \lambda\epsilon_2 - 1 - \lambda + \epsilon_1 \right )\\
& > \mu\cdot f_{S}(A^\star) \left (\epsilon_1 - 3\epsilon_{2} - \lambda\left (\frac{2k}{\epsilon_2} \right) \right )\\
\end{align*}

For any $\lambda \leq \epsilon^2/2k$ the difference above is strictly positive.  Conditioning on $\omega$ being bounded from above by $t^{1/5}$ which happens with probability $1-e^{-\Omega(t^{1/5}/\log t)}$, since $k \in O(\log \log n)$ we that the result holds with probability at least $1-e^{-\Omega(t^{1/10})}$.
\end{proof}

\subsection*{Approximation Guarantee in Expectation}

\begin{lemma*}[\ref{lem:meanapx}]
Let $\delta>0$ and assume $k> 16 /\delta^2$, $c = 16/\delta$.  Suppose that in every $\delta/4$-significant iteration of \textsc{SM-Greedy} when $S$ are the elements selected in previous iterations, $A\in \argmax_{B:|B|=c} \widetilde{\score}(S \cup B)$, the bundle added $\hat{A}$ respects $f_{S}(\hat{A}) \geq (1-\delta)F_{S}(A)$.  Let $\bar{S}$ be the solution after $\lfloor k/c \rfloor$ iterations.  Then, w.p. $\geq 1-1/n^2$:
$$f(\bar{S}) = (1-1/e - 5\delta)\texttt{OPT}.$$
\end{lemma*}

\begin{proof}
We will analyze the solution only on iterations that are $\delta/4$ relevant since this is when we can apply the smoothing arguments.  Since $k> 16 /\delta^2$ and since each iteration is $\delta/4$-significant, by Lemma~\ref{algoworks} we know that in each iteration $A \in \argmax_{B:|B|=c}\widetilde{\score}(S\cup B)$ respects with overwhelming probability:
$$\score_{S}(A) \geq (1-\delta)\max_{B:|B|=c}\score_{S}(B)$$

We will condition on the success of this event in every one of the $\lfloor k/c \rfloor$ iterations.  By a union bound the result will hold w.p. at least $1-1/n^2$.   We assume that $n$ is sufficiently large s.t. $t\geq 2^{20}/\delta^4$.

To account for the fact that we are only analyzing $\delta/4$-significant iterations, we can compare against $(1-\delta/4)$ of the optimal value: let $\hat{k}$ be the last $\delta/4$-significant iteration and $\hat{O} \subseteq O$ be the subset of size $\hat{k}$ of the optimal solution whose value is largest.  By submodularity:
\begin{align}
f(\hat{O}) \geq (1-\delta/4)\texttt{OPT}\label{eq:alpha1}
\end{align}
Second, we argue that optimizing over sets of size $c$ rather than singletons is inconsequential when $k>c/\epsilon$.  To be convinced, notice that when the algorithm selects $c$ elements in every iteration the total number of elements selected will be $k' > k - c$.  Let $O' \in \arg\max_{T:|T|\leq k'}f(T)$.  As in previous arguments, from submodularity we have that: $(1-c/k)f(\hat{O})\leq f(O')$.  Since $k>c/\epsilon$ we have that:
\begin{align}
f(O') > (1-\delta)f(\hat{O}) > (1-2\delta)\texttt{OPT}\label{eq:alpha2}
\end{align}
We will henceforth analyze the algorithm against $O'$.  In a similar manner to the analysis of the greedy algorithm which selects singletons at every stage $i \in [k]$, we can analyze the greedy algorithm which selects sets of size $c$ at every stage $i \in [k'/c]$.  To ease notation assume $\lfloor k'/c \rfloor = k'/c$.

For a given stage of the algorithm, assume the set $S$ has been previously selected and that a set $\hat{A}$ is being added into the solution.  Let $B^{\star} = \arg\max_{B \subseteq O':|B|=c}f_{S}(B)$ and $A^\star = \arg\max_{B:|B|=c} f_{S}(B)$.
\begin{align*}
f_{S}(\hat{A})
& \geq  (1-\delta)\max_{B:|B|=c}\score_{S}(B) \rmk{assumption in the statement}\\
& >  (1-2\delta)\score_{S}(A^\star)  \rmk{Lemma~\ref{algoworks} applied with $\epsilon=\delta$}\\
& >  (1-3\delta)f_{S}(A^\star)\rmk{Lemma~\ref{lem:avg} and $c \geq 1/\delta$}\\
& >  (1-3\delta)f_{S}(B^\star) \rmk{maximality of $A^\star$}\\
& >  (1-3\delta)\frac{c}{k'} \cdot f_{S}(O')  \rmk{subadditivity.}\\
& =  (1-3\delta)\frac{c}{k'} \cdot \left ( f(O' \cup S) - f(S) \right )\\
& \geq  (1-3\delta)\frac{c}{k'} \cdot \left ( f(O') - f(S) \right )
\end{align*}
A standard inductive argument stating that at every iteration $i \in \lfloor k/c \rfloor$ we have that the value of the current solution is at least $\left (1 - (1-1/\lfloor k/c \rfloor)^i\right )\OPT$ implies that $f(\bar{S}) \geq \left(1-1/e - 3\delta \right )f(O')$.  Since we lose $2\delta$ from~(\ref{eq:alpha2}) this concludes our proof.
\end{proof}

\subsection*{From Expectation to High Probability}

\begin{definition}
For a given set $S$, let $A^\star \in \argmax_{B:|B|=c}f_{S}(B)$, $A \in \argmax_{B:|B|=c} \widetilde{\score}(S\cup B)$, and $\mathcal{A}=\{A_{ij}\}_{i\in A,j\notin A}$.  For a fixed $\epsilon>0$:
\begin{itemize}
\item $A_{ij} \in \mathcal{A}$ is $\epsilon$-\emph{good} if ${f_S(A_{ij}) \ge (1 - 2\epsilon) f_S(A^{\star})}$; let $\good$ denote all $\epsilon$-good $A_{ij} \in \mathcal{A}$;
\item $A_{ij} \in \mathcal{A}$ is  $\epsilon$-\emph{bad} if $f_S(A_{ij}) \le (1 - 3\epsilon) f_S(A^{\star})$; let $\bad$ denote all $\epsilon$-bad $A_{ij} \in \mathcal{A}$.
\end{itemize}
\end{definition}

\begin{claim}\label{clm:goodandbad}
For a set $S\subseteq N$ let $A \in \argmax_{B:|B|=c}\widetilde{\score}(S\cup B)$ and assume the iteration is $\epsilon/8$-significant and that $c\geq\epsilon/2$.  Then with probability at least $1-1/n^{10}$:
\begin{itemize}
\item $|\good| \geq  \frac{c(n-c-|S|)}{2}$;
\item $|\bad| \leq  \frac{c(n-c-|S|)}{2}$.
\end{itemize}
\end{claim}

\begin{proof}
Since the sets $A_{ij}$ are distinct both $\good$ and $\bad$ contain no repetitions and we can argue about their size.
To lower bound the size of $\good$, let $A^\star \in \argmax_{A:|A|=c}f_S(A)$.  When the iteration is $\epsilon/8$-significant, from Lemma \ref{algoworks} we know that with exponentially high probability:
$$\score_{S}(A) \geq (1-\epsilon/2)\score_{S}(A^\star)$$
When $c\geq 2/\epsilon$, from Lemma, we know that:
$$\score_{S}(A^\star) \geq (1-\epsilon/2)f_{S}(A^\star)$$
Denoting $m=c(n-c-|S|)$, we get with exponentially high probability:
\begin{align}
\score_{S}(A) = \frac{1}{m}\sum_{j=1}^{m}\sum_{i=1}^{c}f_S(A_{ij}) \ge (1 - \epsilon) f_S(A^{\star})\label{eq:sec5ineq1}
\end{align}
In addition, due to the maximality of $A^\star$ we have that $f_S(A_{ij}) \le f_S(A^{\star})$ for every $i,j$.  Therefore:
\begin{align}
\sum_{j=1}^{m}\sum_{i=1}^{c}f_S(A_{ij}) \le |\good|\cdot f_S(A^{\star}) + \left(m - |\good|\right)\cdot(1 - 2\epsilon) f_S(A^{\star})\label{eq:sec5ineq2}
\end{align}
Putting~(\ref{eq:sec5ineq1}) and ~(\ref{eq:sec5ineq2}) together we get that for sufficiently large $n$, with probability at least $1-1/n^{10}$:
$$m(1-\epsilon)f_{S}(A^\star) \leq \left ( |\good|+(m-|\good|)(1-2\epsilon) \right )f_{S}(A^\star)$$
Rearranging and using $m=c(n-c-|S|)$ we get that $|\good| \geq c(n-c-|S|)/2$.  Since there are a total of $c(n-c-|S|)$ it follows that $|\bad|\leq c(n-c-|S|)$ as required.
\end{proof}

\begin{definition}
Let $\pdf(x)$ denote the probability density function of the noise distribution.  For a set $S: |S| \in O(\log n)$, $c>0$, $\gamma>0$, we define $\mg$ and $\mb$ as:
\begin{itemize}
\item $\int_{\mb}^\infty \pdf(x) dx = \frac{2}{c (n-c-|S|) \log n}$;
\item $\int_{\mg}^\infty \pdf(x) dx = \frac{2\log n}{c (n-c-|S|)}$.
\end{itemize}
\end{definition}

The following claim immediately follows from the definition, yet it is still useful to specify explicitly.  The claim considers $c(n-c-|S|)/2$ samples since this is an upper and lower bound on $|\good|$ and $|\bad|$.  Therefore the claim gives us the likelihood that the largest noise multiplier of $\bad$ does not exceed $\mb$ and that at least one set from $\good$ exceeds $\mg$.

\begin{claim}\label{clm:probmamg}
For a fixed set $S$ and $A \in \argmax_{B:|B|=c}\widetilde{\score}(S\cup B)$, let $m=c(n-c-|S|)$ and consider $m/2$ independent samples from the noise distribution.  Then:
\begin{itemize}
\item $\Pr \left [\max \{\xi_1, \ldots \xi_{m/2}\} \le \mb\right ] > \left (1 - \frac{2}{\log n} \right)$;
\item $\Pr \left [\max \{ \xi_1 \ldots \xi_{m/2}\} \ge \mg \right ] > 1 - 2/ n$.
\end{itemize}
\end{claim}

\begin{proof}
For a single sample $\xi$ from $\mathcal{D}$, we have that:
\[\Pr[\xi \le \mb] = 1 - \frac{2}{m\log n}\]
If we take $m/2$ independent samples $\xi_{1}, \ldots \xi_{m/2}$, the probability they are all bounded by $\mb$ is:
\[\Pr \left [\max \{\xi_1, \ldots \xi_{|\bad|}\} \le \mb\right ] \geq \left (1 - \frac{2}{m \log n} \right )^{\frac{m}{2}} > \left (1 - \frac{2}{\log n} \right)\]
In the case of $\mg$, the probability that a single sample $\xi$ taken from $\mathcal{D}$ is at most $\mg$ is equal to:
$$\Pr \left [ \xi \le \mg \right ] = 1 - \frac{2\log n}{m}$$

If we take independent samples $\xi_{1}, \ldots \xi_{m/2}$, the probability they are all bounded by $\mg$ is:
$$\Pr \left [ \max \{ \xi_1, \ldots \xi_{c(n-c-|S|)} \} \le \mg \right ] = \left (1 - \frac{2\log n }{m}\right )^{\frac{m}{2}} < \frac{2}{2^{\log n}} =\frac{2}{n}$$
And accordingly the probability that at least one of these samples is greater than $\mg$ is:
\begin{equation*}
\Pr \left [\max \{ \xi_1 \ldots \xi_{m/2}\} \ge \mg \right ] > 1 - 2/ n.\qedhere
\end{equation*}
\end{proof}

\paragraph{Showing $\theta_{g}$ is arbitrarily close to $\theta_{b}$.}
Lemma~\ref{relationMbMg} below relates $\mg$ and $\mb$ assuming that $\mathcal{D}$ has a generalized exponential tail.  This lemma makes the result applicable for Exponential and Gaussian distributions, and it fully leverages the fact that $k \in O(\log \log n)$.  The lemma is quite technical, and we therefore first prove the much simpler case where the distribution is bounded.

\begin{lemma}\label{relationMbMg_bounded}
Assume $\mathcal{D}$ has a generalized exponential tail and that $\mathcal{D}$ is bounded, then for all $\eps \in \Omega(1 / \log \log n)$ we have that $\mg\geq (1 - \eps) \mb$.
\end{lemma}

\begin{proof}
Let $\chi$ be an upper bound on $\mathcal{D}$. If there is an atom at $\chi$ with some probability $\gamma > 0$, then we are done, as $\theta_g = \theta_b = \chi$. Otherwise, since $\mathcal{D}$ has a generalized exponential tail we know that $\pdf(\chi) = \gamma$ for some $\gamma > 0$, and that $\pdf$ is continuous at $\chi$. But then there is some $\delta > 0$ such that for any $\chi - \delta \le x \le \chi$ we have that $\pdf(x) \ge \gamma/2$. Choosing $n$ to be large enough that $(1 - \epsilon)\gamma > \gamma - \delta$, we have that

\[\int_{(1 - \epsilon)\gamma}^{\gamma} \pdf(x) \ge \gamma/2 \epsilon\]

Choosing $n$ large enough such that

\[\frac{2\log n}{c (n-c-|S|)} < \gamma/2 \epsilon\]

Gives that $\mg \ge (1 - \epsilon) \chi$. As $\mb \le \chi$ we are done.
\end{proof}

\begin{lemma}\label{relationMbMg}
If $\mathcal{D}$ has a generalized exponential tail then
${(1 - \eps) \mb \le \mg}$, $\forall \eps \in \Omega(1 / \log \log n)$.
\end{lemma}

\begin{proof}
  \noindent The proof follows three stages:
\begin{enumerate}
  \item We use properties of $\mathcal{D}$ to argue upper and lower bounds for $\pdf(x)$;
  \item We show an upper bound $M$ on $\mb$;
  \item We show that integrating a lower bound of $\pdf(X)$ from $(1 - \eps)M$ to $\infty$, yields a probability mass at least $\frac{\log n}{\eps c (n-c-|S|)}$. Now suppose for contradiction that $\mg < (1 - \eps)\mb$, we would get that $\int_{\mg}^{\infty}\pdf(x)$ is strictly greater than $\frac{\log n}{\eps c (n-c-|S|)}$, which contradicts the definition of $\mg$.
\end{enumerate}

We now elaborate each on stage. Recall that by definition of $\mathcal{D}$ for $x \ge x_0$, we have that $\pdf(x) = e^{-g(x)}$, where $g(x) = \sum_{i} a_i \alpha_i$ and that we do not assume that all the $\alpha_i$'s are integers, but only that $\alpha_0 \ge \alpha_1 \ge \ldots$, and that $\alpha_0 \ge 1$. We do not assume anything on the other $\alpha_i$ values.

For the first stage we will show that for every $g(x)$, there exists $n_0$ such that for any $n > n_0$ and $x \ge \left(\frac{\log n}{2a_0}\right)^{1/\alpha_0}$ we have that for $\beta = \eps / 100 < 1/100$:
\[(1 + \beta) a_0 x^{\alpha_0 - 1} e^{-(1 + \beta)a_0 x^{\alpha_0}} \le \pdf(x) \le (1 - \beta) a_0 x^{\alpha_0 - 1} e^{-(1 - \beta)a_0 x^{\alpha_0}} \]
We explain both directions of the inequality. To see
$a_0 x^{\alpha_0 - 1}(1 +\beta) e^{-(1 + \beta)a_0 x^{\alpha_0}} \le \pdf(x)$ we first show:
\[e^{-(1 + \beta/2)a_0 x^{\alpha_0}} \le \pdf(x)\]
This holds since for sufficiently large $n$, we have that:
\[x \ge \frac{(\log n)^{1/\alpha_0}}{2 a_0} \ge \left(\frac{2\sum_{i=1}|a_i|}{\beta a_0}\right)^{\alpha_0 - \alpha_1}\]
So the term $\frac{\beta}{2} x^{\alpha_0}$ dominates the rest of the terms. We now show that:
\[e^{-(1 + \beta/2)a_0 x^{\alpha_0}} \ge a_0 x^{\alpha_0 - 1}(1 +\beta) e^{-(1 + \beta)a_0 x^{\alpha_0}}\]
This is equivalent to:
\[e^{\beta a_0/2 x^{\alpha_0}} \ge a_0 x^{\alpha_0 - 1}(1 +\beta) \]
Which hold for $x = \log \log^3 n$ and large enough $n$.

The other side of the inequality is proved in a similar way. We want to show that:
\[\pdf(x) \le (1 - \beta) a_0 x^{\alpha_0 - 1} e^{-(1 - \beta)a_0 x^{\alpha_0}}\]
Clearly for $x > \log \log^3 n$ we have that $(1 - \beta) a_0 x^{\alpha_0 - 1} > 1$. Hence we just need to show that:
\[\pdf(x) \le e^{-(1 - \beta)a_0 x^{\alpha_0}}\]
But this holds for sufficiently large $n$ s.t.:
\[x \ge \frac{(\log n)^{1/\alpha_0}}{2 a_0} \ge \left(\frac{\sum_{i=1}|a_i|}{\beta a_0}\right)^{\alpha_0 - \alpha_1}\]

We now proceed to the second stage, and compute an upper bound on $\mb$. Note that if
$$\int_{\mb}^{\infty}  \pdf(x) = \int_{M}^{\infty}  g(x)$$
and for every $x \ge M$ we have $\pdf(x) \le g(x)$ then it must be that $M \ge \mb$. Applying this to our setting, we bound $\pdf(x) \le (1 - \beta) a_0 x^{\alpha_0 - 1} e^{-(1 - \beta)a_0 x^{\alpha_0}}$ to get:
\begin{align*}
\frac{1}{c (n-c-|S|) \log n}
& = \int_{M}^{\infty} (1 - \beta) a_0 x^{\alpha_0 - 1} e^{-(1 - \beta)a_0 x^{\alpha_0}} \\
& = -e^{-(1 - \beta)a_0 x^{\alpha_0}}|_{M}^{\infty} \\
& = e^{-(1 - \beta)a_0 M^{\alpha_0}}
\end{align*}
Taking the logarithm of both sides, we get:
\begin{align*}
-(1 - \beta)a_0 M^{\alpha_0}
& = \log \frac{1}{c (n-c-|S|) \log n}\\
& = - \log (c (n-c-|S|) \log n)
\end{align*}
Multiplying by $-1$, dividing by $(1 - \beta)a_0$ and taking the $1/\alpha_0$ root we get:
\[M = \left(\frac{\log (c (n-c-|S|) \log n)}{(1 - \beta)a_0}\right)^{\alpha_0}\]
Note that $(1 - \eps) M > \left(\frac{\log n}{2a_0}\right)^{1/\alpha_0}$ and hence our bounds on $\pdf(x)$ hold for this regime.

We move to the third stage, and bound $\int_{(1 -\eps)M}^{\infty} \pdf(x)$ from below. If we show that:
$\int_{(1 -\eps)M}^{\infty} \pdf(x)$ is greater than $\frac{\log n}{\eps c (n-c-|S|)}$, this implies that $\mg \ge (1 -\eps)M$, as $\mg$ is defined as the value such that when we integrate $\pdf(x)$ from $\mg$ to $\infty$ we get exactly $\frac{\log n}{\eps c (n-c-|S|)}$. We show:
\begin{align*}
\int_{(1 -\eps)M}^{\infty} \pdf(x)
& \ge (1 + \beta)a_0 \alpha_0 x^{\alpha_0-1}e^{-(1 + \beta)a_0 x^{\alpha_0}} \\
& = -e^{-(1 + \beta)a_0 x^{\alpha_0}}|_{(1 - \eps)M}^{\infty} \\
& = e^{-(1 + \beta)a_0 ((1 - \eps)M)^{\alpha_0}} \\
& = e^{-(1 + \beta)a_0 M^{\alpha_0}(1-\eps)^{\alpha_0}} \\
& \ge e^{-(1 + \beta)a_0 M^{\alpha_0}(1-\eps)}
\end{align*}

However $a_0 M^{\alpha_0} = \left(\frac{\log (c (n-c-|S|) \log n)}{(1 - \beta)}\right)$. Since $\beta < 0.1$ we have that $\frac{1+\beta}{1 - \beta} < 1 + 3 \beta$. Substituting both expressions we get:
\begin{align*}
e^{-(1 + \beta)a_0 M^{\alpha_0}(1-\eps)}
& \ge e^{-(1 + 3\beta)(1 - \eps) \log (c (n-c-|S|) \log n)} \\
& = \left(\frac{1}{c (n-c-|S|) \log n}\right)^{(1 - \eps)(1 + 3 \beta)} \\
& \ge \left(\frac{1}{c (n-c-|S|) \log n}\right)^{(1 - \eps/2)}
\end{align*}

Where we used that $\beta = \eps/100$ and hence $(1 - \eps)(1 + 3 \beta) < 1 - \eps/2$. We now need to compare this to $\frac{\sqrt{\log n}}{\eps c (n-c-|S|)}$. To
do this, note that:
\begin{align*}
\left(\frac{1}{c (n-c-|S|) \log n}\right)^{(1 - \eps/2)}
& \ge \frac{1}{c (n-c-|S|)^{1 - \eps/2} \log n}\\
& \ge \frac{2^{\sqrt{\log n}}}{c (n-c-|S|) \log n} \\
& \ge \frac{\log n}{\eps c (n-c-|S|)}
\end{align*}

Where $n$ is large enough that $\frac{\eps}{2} \log (n - c - |S|)  > \sqrt{\log n}$.  This completes the proof, since ${\mg \ge (1 - \eps)M \ge (1 - \eps)\mb}$ as required.
\end{proof}

\begin{lemma*}[\ref{single}]
For any $\epsilon>0$, suppose we run \textsc{SM-Greedy} where in each iteration we add a bundle of elements of size $c= 16/\epsilon$.  For any $\epsilon/8$-significant iteration where the set previously selected is ${S: |S| \in O(\log\log n)}$, let $A \in \argmax \widetilde{\score}(S\cup A)$ and $\MS = \argmax_{(i,j)\in A \times N\setminus S\cup A} \widetilde{f}(S \cup A_{ij})$.
Then, with probability at least $1-3/\log n$ we have that:
$$f_S(\MS) \ge (1 - 3\epsilon)\score_S(A).$$
\end{lemma*}

\begin{proof}
We will use the above claims to argue that with probability at least $1-4/\log n$ the noisy mean value of any set in $\bad$ is smaller than the largest noisy mean value of a set in $\good$.  Since a bad set is defined as a set $B$ for which $f_S(B) \leq (1-3\epsilon)f_{S}(A^\star)$ this implies that the set returned by the algorithm has value at least $(1-3\epsilon)f_{S}(A^\star)$.  Since for any set $A:|A|=c$ we have that $f_{S}(A^\star)$ is an upper bound on $\score_{S}(A)$ will complete the proof.

We will condition on the event that $|\good| \geq c(n-c-|S|)/2$ which happens with probability at least $1-1/n^{10}$ from Claim~\ref{clm:goodandbad}.  Under this assumption, from Claim~\ref{clm:probmamg} we know that with probability at least $1-2/n$ at least one of the noise multipliers of sets in $\good$ has value at least $\mg$, and from Lemma~\ref{relationMbMg} we know that $\mg \geq (1-\gamma)\mb$ for any $\gamma \in \Theta(1/\log\log n)$.  Thus:
\begin{align*}
\max_{A_{ij} \in \good}\tilde{f}(S\cup A_{ij})
&& = 	  && \max_{A_{ij} \in \good}\xi_{A_{ij}} &&\times&& &&   [ && f(S) && + && f_{S}(A_{ij})  && ] &&\\
&& \geq && \mg  &&\times&& 							  && [ && f(S) && + && (1-2\epsilon)f_{S}(A^\star) && ] &&\\
&& \geq && (1-\gamma) \mb   &&\times&& && [ && f(S)&& + &&(1-2\epsilon)f_{S}(A^\star)&&]&&
\end{align*}
Let $B \in \argmax_{C \in \bad} \widetilde{f}(S \cup C)$.  From Claim~\ref{clm:probmamg} we know that w.p. at least $1-2/\log n$ all noise multipliers of sets in $\bad$ are at most $\mb$.  Thus:
$$
\widetilde{f}(S \cup B) = \max_{A_{ij} \in \bad}\tilde{f}(S \cup A_{ij}) = \max_{A_{ij} \in \bad}\xi_{A_{ij}}f(S \cup A_{ij}) \leq \mb\cdot [f(S)+(1-3\epsilon)f_{S}(A^\star)]
$$
Let $d$ be some constant such that $|S| \leq d\log\log n$.  Note that the iteration is $\epsilon$-significant, and therefore due to the maximality of $A^\star$ and since $f(S)\leq \texttt{OPT}$ and the optimal solution has at most $d\cdot \log\log n$ elements we have that:
$$f_{S}(A^\star) \geq \frac{\epsilon}{d\log\log n}f(S).$$
Since Lemma~\ref{relationMbMg} applies to any $\gamma \in \Theta(1/\log\log n)$, we know that for any constant $d$ there is a large enough value of $n$ such that $\gamma< \epsilon^2/3d\log\log n$.
Putting it all together and conditioning on all events we have with probability at least $1-3/\log n$:
\begin{align*}
\widetilde{f}(S \cup \MS) - \widetilde{f}(S \cup B)
& \geq && \Big ( (1-\gamma) \mb\cdot [f(S)+(1-2\epsilon)f_{S}(A^\star)] \Big ) -  \Big (\mb \cdot [f(S)+(1-3\epsilon)f_{S}(A^\star)] \Big )  &&
\end{align*}
\begin{align*}
& \geq &&\mb &&\Big (  && \epsilon f_{S}(A^\star) &&-&& \gamma &\times& &&[ &&(1-2\epsilon) f_{S}(A^\star) &&+ &&f(S)       &&    ] && \Big )&&\\
& \geq &&\mb && \Big ( &&\epsilon f_{S}(A^\star) &&-&& \gamma  &\times&  &&\big [ &&(1-2\epsilon) f_{S}(A^\star) &&+&& \frac{d\log\log n}{\epsilon} f_S(A^\star)      &&   \big ] && \Big )&&\\
& = &&\mb f_{S}(A^\star) && \Big ( &&\epsilon  &&-&& \gamma &\times&  &&\big [ &&(1-2\epsilon) &&+&& \frac{d\log\log n}{\epsilon}      &&   \big ] && \Big )&&\\
& > &&\mb f_{S}(A^\star) && \Big ( &&\epsilon  &&-&& \frac{\epsilon^2}{3d\log \log n} &\times& &&\big [ &&(1-2\epsilon) &&+&& \frac{d\log\log n}{\epsilon}      &&   \big ] && \Big )&&\\
& > &&\mb f_{S}(A^\star) && \Big ( &&\epsilon  &&-&& \frac{2\epsilon}{3} \Big ) && && && && && && &&
\\
& > && 0
%
%
%
\end{align*}
Since the difference is strictly positive this implies that with probability at least $1-3/\log n$ a bad set will not be selected by the algorithm which concludes our proof.
\end{proof}

\subsection*{Approximation Guarantee of SM-Greedy}

\begin{theorem}\label{final-thm}
For any monotone submodular function $f:2^N \to \mathbb{R}$ and $\epsilon>0$, when ${k  \in \Omega({1}/{\epsilon}) \cap O(\log \log n)}$, there is a $(1 - 1/e - \epsilon)$ approximation for $\max_{S:|S|\leq k} f(S)$, with probability $1-4/\log n$ given access to a noisy oracle whose distribution has a generalized exponential tail.
\end{theorem}

\begin{proof}\label{final-thm_proof}
First, for the case in which $k \in \Omega({1}/{\epsilon^2})$, we can apply \textsc{SM-Greedy} as described in the main body of the paper.  Let $\delta = \epsilon/5$ and set $c = 16/\delta$.  At any given $\delta/8$-significant iteration of \textsc{SM-Greedy} from Lemma~\ref{single} we know that with probability at least $1-3/\log n$ we have that $f(\MS)\geq (1-\delta) \score_S(A)$, where $A \in \argmax_{B:|B|=c}\widetilde{F}(B) $.  We can then apply Lemma~\ref{lem:meanapx} which implies that with probability at least $1-\frac{4}{\log n}$ we have a $1-1/e-5\delta = (1-1/e-\epsilon)$ approximation.

In the case $k \in \Omega(1/\epsilon) \cap O(1/\epsilon^2)$ note that taking bundles of size $c\in O(1/\epsilon)$ in each iteration may result in a $1/2$ approximation.  In this case, we therefore enumerate over all possible sets of size $c=k$ and output ${\MS = \argmax \widetilde{f}(A_{ij})}$ where ${A = \argmax_{B:|B|=k} \widetilde{\score}(B)}$.
By Lemma~\ref{single} we know that w.p. $1-3\log n$:
\begin{align}
f(\MS) \geq (1-48/c)F(A) = (1-48/k)F(A) \geq (1-\epsilon/2)F(A)\label{eq:smallk1}
\end{align}
By the smoothing lemma (Lemma \ref{algoworks}) we know that for any fixed $\epsilon$ and sufficiently large $n$ with overwhelming probability $\score (A) \ge (1-\epsilon/2)\score(A^\star)$ for ${A^\star \in \argmax_{B:|B|=k}f(B)}$.  By the sampled mean method (Lemma~\ref{lem:avg}) we know that $F(A^\star) \geq (1-1/k)f(A^\star)$, thus:
\begin{align}
F(A) \geq (1-1/k - \epsilon/2)f(A^\star) \label{eq:smallk2}
\end{align}
Putting (\ref{eq:smallk1}) and (\ref{eq:smallk2}) together and taking a union bound we get our result.
\end{proof}

\newpage

\section{Optimization for Very Small $k$}\label{sec:very_smallk_appendix}

\subsection*{Smoothing Guarantees}

\begin{lemma*}[\ref{smoothing_verysmall}]
Let $A \in \argmax_{B:|B|=k}\widetilde{F}(B)$.  Then, for any fixed $\epsilon>0$ w.p. $1-e^{-\Omega(\epsilon^2(n-k))}$:
$$F(A) \geq (1-\epsilon)\max_{B:|B|=k} F(B)$$
\end{lemma*}
\begin{proof}
The proof follows the same reasoning as those from previous sections.  Let $A^\star = \argmax_{B:|B|=k}F(B)$.  
We will show that w.h.p. no set $B$ for which $F(B)<(1-\epsilon)F(A^\star)$ beats $A$.  The size of the smoothing set is $t=n-k$, and $\omega$ is an upper bound on the noise multiplier. 

Note that the optimality of $A^\star$ and submodularity imply that $f(A^\star \cup x) \leq 2f(A^\star)$, for all $x \in N\setminus A^\star$.  Hence from monotonicity the variation is bounded by 2: 
$$v(A^\star) = \frac{\max_{x \in N\setminus A}f(A^\star\cup x)}{\min_{x \in N\setminus A}f(A^\star\cup x) } \leq \frac{2f(A^\star)}{f(A^\star)} = 2$$
We can therefore apply Lemma~\ref{lem:boundonb} and get that with probability at least $1-e^{\Omega(\lambda^2t^{1/4}/\omega)}$:
$$\widetilde{F}(A^\star) \geq (1-\lambda)\mu \left (1- 4t^{-1/4} F(A^\star)  \right )$$
To upper bound $\widetilde{F}(B)$ for a set $B$ s.t. $F(B) < (1-\epsilon)F(A^\star)$, note that the value of largest set in the smoothing neighborhood is $\max_{x \in N\setminus B}f(B\cup x)\leq 2f(A^\star)$.  Hence, from Lemma~\ref{lem:noisysmoothing_lowerbound} we get that with probability at least 
$1-e^{\Omega(\lambda^2t^{1/4}/\omega)}$:
$$F(B) \leq (1+\lambda)\mu \left (F(B) +6t^{-1/4}F(A^\star) \right )$$
Therefore when $n$ is sufficiently large s.t. $t^{-1/4} \leq \epsilon/100$ and $\lambda<1$ we get that:
\begin{align*}
F(A^\star) - F(B)
& \geq (1-\lambda)\mu (1-4t^{-1/4})F(A^\star) - (1+\lambda)\mu \left( F(B) + 6t^{-1/4}F(A^\star)  \right )\\
& \geq \mu \left ( (1-\lambda)(1-\frac{4\epsilon}{100})F(A^\star) - (1+\lambda)(1-\epsilon)F(A^\star) - (1+\lambda)\frac{6\epsilon}{100}F(A^\star)\right ) \\
& \geq \mu \left ( (1-\lambda)(1-\frac{4\epsilon}{100})F(A^\star) - (1+\lambda)(1-\epsilon)F(A^\star) - (1+\lambda)\frac{6\epsilon}{100}F(A^\star)\right ) \\
& > \mu \cdot F(A^\star)\left ( \epsilon -2\lambda - \epsilon/5\right ) \\
\end{align*}
Using $\lambda<\epsilon/10$ the above inequality is strictly positive.  Conditioning on the event of $\omega$ being sufficiently small completes the proof.
\end{proof}

\subsection*{An Approximation Algorithm for Very Small $k$}

\paragraph{Approximation guarantee in expectation.}  We first present the algorithm whose approximation guarantee is arbitrarily close to $k/(k+1)$, in expectation. 

\begin{algorithm}
\caption{\textsc{Exp-Small-Greedy}}\label{alg:Small-Greedy(exp)}
\label{c}
\begin{algorithmic}[1]
\INPUT budget $k$
	\STATE $A \leftarrow \arg\max_{B\ : |B|=k}\widetilde{F}(B)$
	\STATE $x \leftarrow \textrm{select random element from $N\setminus A$}$
	\STATE $ \MS \leftarrow \textrm{random set of size $k$ from $A \cup x$}$
\RETURN $\MS$ 
\end{algorithmic}
\end{algorithm}

\begin{theorem}\label{thm:tinyk}
For any submodular function $f:2^N \to \mathbb{R}$, the algorithm \textsc{Exp-Small-Greedy} obtains returns a $(k/(k+1) - \epsilon)$ approximation for $\max_{S:|S|\leq k}f(S)$, in expectation, for any fixed $\epsilon>0$.
\end{theorem}

\begin{proof}
From Lemma~\ref{lem:avg} we know that $f(\MS) \geq (k/(k+1))F(A)$.  Let $A^\star = \argmax_{B:|B|=k}f(B)$.  From monotonicity we know that $f(A^\star) \leq F(A^\star)$.   Applying Lemma~\ref{smoothing_verysmall} we get that for the set $F(A) \geq (1-\epsilon)F(A^\star)$.  Hence:
\begin{equation*}
f(\MS) \geq \left (\frac{k}{k+1}\right)F(A) \geq (1-\epsilon)\left (\frac{k}{k+1} \right)F(A^\star) \geq (1-\epsilon)\left (\frac{k}{k+1} \right) f(A^\star) 
> \left ( \left (\frac{k}{k+1} \right)-\epsilon \right ) \texttt{OPT}.\qedhere
\end{equation*}
\end{proof}

\paragraph{High probability.} To obtain a result w.h.p. we modify the algorithm above.  The algorithm enumerates all possible subsets of size $k-1$, and then select the set $A\in \argmax_{B:|B|=k-1}\widetilde{F}(B)$.  The algorithm then selects $\MS \in \argmax_{X \in \mathcal{H}(A)} \widetilde{f}(X)$.  A formal description is added below.

\begin{algorithm}
\caption{\textsc{WHP-Small-Greedy}}\label{alg:Small-Greedy(hp)}
\label{c}
\begin{algorithmic}[1]
\INPUT budget $k$
	\STATE $A \leftarrow \arg\max_{B\ : |B|=k-1}\widetilde{F}(B)$
	\STATE $ \MS \leftarrow \argmax_{x \in N\setminus A}\widetilde{f}(A \cup x)$
\RETURN $\MS$
\end{algorithmic}
\end{algorithm}
The analysis of the algorithm is similar to the high probability proof from Section~\ref{sec:smallk}.

\begin{theorem*}[\ref{thm:verysmallkwhp}]
For any submodular function $f:2^N \to \mathbb{R}$ and any fixed $\epsilon>0$ and constant $k$, there is a $\left(1-1/k-\epsilon\right)$-approximation algorithm for $\max_{S:|S|\leq k}f(S)$ which only uses a generalized exponential tail noisy oracle, and succeeds with probability at least $1-6/\log n$.
\end{theorem*}
\begin{proof}
Let $A \in \argmax_{B:|B|=k-1}\widetilde{F}(B)$, and let $A^\star \in \argmax_{B:|B|=k-1} f(B)$.  Since $A^\star$ is the optimal solution over $k-1$ elements, from submodularity we know that $f(A^\star) \geq (1-1/k)\texttt{OPT}$. 
%
%
%
%
What now remains to show is that $\MS \in \argmax_{x \in N\setminus A} \widetilde{f}(A \cup x)$ is a $(1-\epsilon)$ approximation to $F(A)$.  To do so recall the definitions of good and bad sets from the previous section.  Let $\delta = \epsilon/3$.  Suppose that a set $X$ is in $\delta$\texttt{-good}(A) if $f(X) \geq (1-2\delta)f(A^\star)$ and in $\delta$\texttt{-bad}(A) if $f(X)\leq (1-3\delta)f(A^\star)$.  We will show that the set selected has value at least as high as that of a bad set, i.e. $(1-3\delta)f(A^\star)$ which will complete the proof.  

We first show that with probability at least $1-6/\log n$ the noise multiplier of some good set is at least $\mg$ and of a bad set is at most $\mb$.  To do so we will first argue about the size of $\delta$\texttt{-good}(A) and $\delta$\texttt{-bad}(A).  From Lemma~\ref{smoothing_verysmall} and the maximality of $A$ we know that with exponentially high probability  $F(A) \geq (1-\delta)F(A^\star)$.  Therefore for $m = n-k$:
$$F(A) = \frac{1}{m}\sum_{x \notin A}f(A \cup x) \geq (1-\delta)\frac{1}{m}\sum_{x \notin A^\star}f(A^\star \cup x) \geq (1-\delta)f(A^\star) $$  
Due to the maximality of $A^\star$ and submodularity we know that $f(A\cup x)\leq 2f(A^\star)$ for all $x \notin A$:
$$\sum_{x \notin A}f(A\cup x) \leq |\delta\texttt{-good}(A)|2f(A^\star) + (m- |\delta\texttt{-good}(A)|)(1-2\delta)f(A^\star)$$
Putting the these bounds on $F(A)$ together and rearranging we get that:
$$ |\delta\texttt{-good}(A)| \geq \frac{\delta \cdot m}{1+2\epsilon} \geq \frac{\delta m}{3}$$
Therefore, for sufficiently large $n$ the likelihood of at least one set achieving value at least $\mg$ is: 
$$\Pr[ \max \{\xi_{1},\ldots,\xi_{\delta \cdot m/3}\} \geq \mg] \geq 1- \left (1-\frac{2\log n}{m} \right)^\frac{\delta  m}{3} \geq 1 - \frac{2}{n^{\delta/3}} \geq 1-\frac{1}{\log n}$$
To bound $\delta$\texttt{-bad}(A) we will simply note that it is trivial that $\delta\texttt{-bad}(A)< m$.  Thus, the likelihood that all noise multipliers of bad sets are bounded from above by $\mb$ is:  
\[\Pr \left [\max \{\xi_1, \ldots \xi_{m}\} \le \mb\right ] \geq \left (1 - \frac{2}{m \log n} \right )^{m} > \left (1 - \frac{4}{\log n} \right)\]
Thus, by a union bound and conditioning on the event in Lemma~\ref{smoothing_verysmall} we get that $\mb$ is an upper bound on the value of the noise multiplier of bad sets and $\mg$ is with lower bound on the value of the noise multiplier of a good stem all with probability at least $1-6/\log n$.  From Lemma~\ref{relationMbMg} we know that for any $\gamma \in \Theta(1/\log\log n)$ we have that $\mg \geq (1-\gamma)\mb$.  Thus: 
\begin{align*}
\max_{X \in  \delta\texttt{-good}(A)}\tilde{f}(X)
= \max_{X \in \delta\texttt{-good}(A)}\xi_{X }f(X) \geq \mg\cdot (1-2\delta)f(A^\star)
 \geq (1-\gamma) M_{b}\cdot (1-2\delta)f(A^\star)
\end{align*}
Let $B \in \argmax_{C \in \delta\texttt{-bad}} \widetilde{f}(S \cup C)$.  From Claim~\ref{clm:probmamg} we know that with probability at least $1-2/\log n$ all noise multipliers of sets in $\bad$ are at most $\mb$.  Thus:
$$
\widetilde{f}(S \cup B) = \max_{X \in \delta\texttt{-bad}}\tilde{f}(X) = \max_{X \in \bad}\xi_{X}f(X) \leq M_{b}\cdot(1-3\delta)f(X) 
$$
Putting it all together we have with probability at least $1-6/\log n$:  
\begin{align*}
\widetilde{f}(\MS) - \widetilde{f}(B) 
 \geq M_{b}f(A^\star)\cdot \left((1-\gamma)(1-2\delta) - (1-3\delta)\right)
 > \mb f(A^\star)\left ( \delta -\gamma \right )
\end{align*}
Since Lemma~\ref{relationMbMg} applies to any $\gamma \in \Theta(1/\log\log n)$, and $\delta$ is fixed it applies to $\gamma<\delta$ and the difference is positive.  Since $\delta = \epsilon/6$ this completes our proof.
\end{proof}

\subsection*{Information Theoretic Lower Bounds for Constant $k$}
Surprisingly, even for $k=1$ no algorithm can obtain an approximation better than $1/2$, which proves a separation between large and small $k$.\footnote{We note that if the algorithm is not allowed to query the oracle on sets of size greater than $k$, Claim~\ref{clm:lower_constant} can be extended to show a$O(n)$ inapproximability, so choosing a random element is almost the best possible course of action.} The following is a tight bound for $k=1$.

\begin{claim}\label{clm:lower_constant}
There exists a submodular function and noise distribution for which w.h.p. no randomized algorithm with a noisy oracle can obtain an approximation better than $1/2 +O(1/\sqrt{n})$ for $\max_{a \in N} f(a)$.
\end{claim}

\begin{proof}
We will construct two functions that are identical except that one function attributes a value of $2$ for a special element $x^\star$ and $1$ for all other elements, whereas the other is assigns a value of $1$ for each element.  In addition, these functions will be bounded from above by 2 so that the only queries that gives any information are those of singletons.  More formally, consider the functions $f_{1}(S) = \min\{|S|,2\}$ and $f_{2}(S)=\min\{g(S),2\}$ where $g:2^{N}\to \mathbb{R}$ is defined for some $x^\star \in N$ as:
\[
    g(S)=
\begin{cases}
    2,& \text{if } S = x^\star\\
    |S|,              & \text{otherwise}
\end{cases}
\]
The noise distribution will return $2$ with probability $1/\sqrt{n}$ and 1 otherwise.

We claim that no algorithm can distinguish between the two functions with success probability greater than $1/2+O(1/\sqrt{n})$.
For all sets with two or more elements, both functions return 2, and so no information is gained when querying such sets.  Hence, the only information the algorithm has to work with is the number of 1, 2, and 4 values observed on singletons. If it sees the value 4 on such a set, it concludes that the underlying function is $f_2$. This happens with probability $1/\sqrt{n}$.

Conditioned on the event that the value 4 is not realized, the only input that the algorithm has is the number of 1s and 2s it sees. The optimal policy is to choose a threshold, such if a number of 2s observed is or above this threshold, the algorithm returns $f_2$ and otherwise it reruns $f_1$.  In this case, the optimal threshold is $\sqrt{n}+1$. 

The probability that $f_2$ has at most $\sqrt{n}$ twos is $1/2 - 1/\sqrt{n}$, and so is the probability that $f_1$ has at least $\sqrt{n}+1$ twos, and hence the advantage over a random guess is $O(1/\sqrt{n})$ again.

An algorithm which approximates the maximal set on $f_2$ with ratio better than $1/2 + \omega(1/\sqrt{n})$ can be used to distinguish the two functions with advantage $\omega(1/\sqrt{n})$. Having ruled this out, the best approximation one can get is $1/2 +O(1/\sqrt{n})$ as required.
\end{proof}

We generalize the construction to general $k$. The lower for general $k$ behaves like $2k/(2k-1)$, where our upper bound is $(k-1)/k$.

\begin{claim}\label{clm:evenlower_constant}
There exists a submodular function and noise distribution for which w.h.p. no randomized algorithm with a noisy oracle can obtain an approximation better than $(2k-1)/2k +O(1/\sqrt{n})$ for the optimal set of size $k$.
\end{claim}

\begin{proof}
Consider the function:
\[
    f_1(S)=
\begin{cases}
    2|S|,& \text{if } |S| < k \\
    2k-1, & \text{if } |S| = k \\
    2k, & \text{if } |S| > k \\
\end{cases}
\]
and the function $f_2$, which is dependent on the identity of some random set of size $k$, denoted $S^{\star}:$
\[
    f_2(S;S^{\star})=
\begin{cases}
    2|S|,& \text{if } |S| < k \\
    2k-1, & \text{if } |S| = k, S \neq S^{\star}\\
    2k, & \text{if }  S = S^{\star}\\
    2k, & \text{if } |S| > k \\
\end{cases}
\]
Note that both functions are submodular.

The noise distribution will return $2k/(2k-1)$ with probability $n^{-{1}/{2}}$ and 1 otherwise.  Again we claim that no algorithm can distinguish between the functions with probability greater than $1/2$. Indeed, since $f_1,f_2$ are identical on sets of size different than $k$, and their value only depends on the set size, querying these sets doesn't help the algorithm (the oracle calls on these sets can be simulated). As for sets of size $k$, the algorithm will see a mix of $2k-1$, $2k$, and at most one value of $4k^2/(k-1)$. If the algorithm sees the value $4k^2/(k-1)$ then it was given access to $f_2$. However, the algorithm will see this value only with probability $1/\sqrt{n}$. Conditioning on not seeing this value, the best policy the algorithm can adopt is to guess $f_2$ if the number of $2k$ values is at least $1 + \frac{{n \choose{k}}}{\sqrt{n}}$, and guess $f_1$ otherwise. The probability of success with this test is $1/2 + O(1/\sqrt{n})$ (regardless of whether the underlying function is $f_1$ or $f-2$). Any algorithm which would approximate the best set of size $k$ to an expected ratio better than $(2k-1)/2k + \omega(1/\sqrt{n})$ could be used to distinguish between the function with an advantage greater than $1/\sqrt{n}$, and this puts a bound of $(2k-1)/2k + O(1/\sqrt{n})$ on the expected approximation ratio.
\end{proof}

\newpage
\section{Noise Distributions}\label{sec:distributions}
As discussed in the Introduction, our goal was to allow noise distribution in the model to potentially be Gaussian, Exponential, uniform and generally bounded.  It was important for us that algorithm to be oblivious to the specific noise distribution, and rely on its properties only in the analysis. For achieve this we introduced the class of \emph{generalized exponential tail} distributions.  We recall the definition from the Introduction.

\begin{definition*}
A noise distribution $\mathcal{D}$ has a \textbf{generalized exponential tail} if there exists some $x_0$ such that for $x > x_0$ the probability density function $\pdf(x) = e^{-g(x)}$, where $g(x) = \sum_{i} a_ix^{\alpha_i}$. We do not assume that all the $\alpha_i$'s are integers, but only that $\alpha_0 \ge \alpha_1 \ge \ldots$, and that $\alpha_0 \ge 1$. If $\mathcal{D}$ has bounded support we only require that either it has an atom at its supremum, or that $\pdf$ is continuous and non zero at the supremum.
\end{definition*}

%

Note that the definition includes Gaussian and Exponential distributions.  For $i >0$ it is possible that $\alpha_i < 1$ which implies that a generalized exponential tail also includes cases where the probability density function denoted $\pdf$ respects $\pdf(x) = \pdf(x_0) e^{-g'(x-x_0)}$ (we can simply add $\pdf(x_0)$ to $g$ using $\alpha_i = 0$ for some $i$, and move from $g'(x - x_0)$ to an equivalent $g(x)$ via a coordinate change).  

The most important property of the noise distribution is that all of its moments are constant, independent of $n$. In fact, $\mathcal{D}$ describes how the noise affects a single evaluation, and does not depend on the number of elements. This means (for example) that if we could get $h(n)$ independent samples from $\mathcal{D}$, we would be arbitrarily close to the mean, as long as $h(n)$ is monotone in $n$.

\paragraph{Impossibility for distributions that depend on $n$.}  We note that if the adversary would have been allowed to choose the noise distribution as a function of $n$, then no approximation would be possible, even if the noise distribution had mean $1$. For example, a noise distribution which returns 0 with probability $1 - 1/2^{2n}$ and $2^{2n}$ with probability $1/2^{2n}$ has an expected value of $1$, is not always $0$, but does not enable any approximation.

\paragraph{Impossibility for two distributions.}
One can consider having multiple noise distributions which act on different sets. A noise distribution can be assigned to a set either in adversarial manner, or at random.
If sets are assigned to noise distributions in an adversarial manner, it is possible to construct the bad example of the correlated case from Section~\ref{sec:extensions} with just  two noise distributions.  If sets are assigned to a noise distribution in an i.i.d manner, this reduces to the i.i.d case when there is a single distribution.

\paragraph{The relation between $n$ and the distribution}
As we have explained above, if the distribution depends on $n$, then approximation is not possible. In particular, this means that if the universe is too small, optimization is not possible. For example, suppose that $\mathcal{D}$ returns $0$ with probability $1 - 2^{-100}$, and otherwise returns $2^{100}$. Then $\mathcal{D}$ is bounded away from zero, has expectancy $1$, but approximation is not possible if $n = 50$. Hence we need to assume some minimal value $n_0$ that depends on the distribution, and assert an approximation ratio of $1 - 1/e - \epsilon$ only for $n > n_0$. We note that $n_0$ is constant, and hence if $n \le n_0$ we can run the ``optimal'' algorithm of evaluating the noisy oracle over all subsets of $n$, but the approximation ratio might still be arbitrarily bad.

We note that the problem is not ``just'' an atom at zero. Suppose that $f$ is additive, and bounded between $1$ and $100$. if $\mathcal{D}$ is uniform over the set $2^{100^{i}}$ for $1 \le i \le 2^{100}$ and $n = 50$ then approximation is not possible; if $\tilde f(A)$ turns out to be larger than $\tilde f(B)$ this says very little about $f(A), f(B)$ - it's more likely happen due to the noise.

\newpage
\section{Additional Examples}\label{sec:examples}
In this section we show some examples of how greedy and its variants fail under error and noise.

\paragraph{Greedy fails with error.}  In the maximum-coverage problem we are given a family of sets that cover a universe of items, and the goal is to select a fixed number of sets whose union is maximal.  This classic problem is an example of maximizing a monotone submodular function under a cardinality constraint. For a concrete example showing how greedy fails with error, consider the instance illustrated in Figure~\ref{fig:noise}.  In this instance there is one family of sets $\mathcal{A}$ depicted on the left where all sets cover the same two items, and another family of disjoint sets $\mathcal{B}$ that each cover a single unique item.  Consider an oracle which evaluates sets as follows.
For any combination of sets the oracle evaluates the cardinality of the union of the subsets exactly, except for a few special cases: For $S = A \cup b \ \ \forall A \subseteq \mathcal{A},b\in \mathcal{B}$  the oracle returns $\widetilde{f}(S) =2$,
and for $S \subseteq \mathcal{A}$ the oracle returns $\widetilde{f}(S)=2+\delta$ for some arbitrarily small $\delta>0$.
With access to this oracle, the greedy algorithm will only select sets in $\mathcal{A}$ which may be as bad as linear in the size of the input.  In this example we tricked the greedy algorithm with a $1/3$-erroneous oracle, but same consequences apply to an $\epsilon$-erroneous oracle for any $\epsilon>0$ by planting $(1-\epsilon)/\epsilon$ items in $\mathcal{A}$.

\paragraph{Greedy fails with random noise.}  In practice, the greedy algorithm is often used although we know the data may be noisy.  Hence, a different direction for research could be to analyze the effect of noise on the existing greedy algorithm. Unfortunately, it turns out that the greedy algorithm fails even on very simple examples.

\begin{theorem}
Given a noise distribution that is either uniformly distributed in ${[1 - \epsilon, 1 + \epsilon]}$ for any $\epsilon>0$, a Gaussian, or an Exponential, the greedy algorithm cannot obtain a constant factor approximation ratio even in the case of maximizing additive functions under a cardinality constraint.
\end{theorem}
\begin{proof}[Proof sketch]
Consider an additive function, which has two types of elements: $k = \sqrt{n}$ good elements, each worth $n^{1/4}$, and $n - k$ bad elements, each worth $1$. Suppose that the noise is uniform in $[1 - \epsilon, 1 + \epsilon]$. Then after taking $k^{2/3}$ good elements greedy is much more likely to take bad elements, which leads to an approximation ratio of $O(1/n^{1/6})$.  Similar examples hold for Gaussian and Exponential noise.
\end{proof}

\paragraph{Greedy fails when taking maximal sampled mean bundle.}  In Section~\ref{sec:smallk} we discuss a greedy algorithm which iteratively takes bundles of $O(1/\epsilon)$ elements that maximize $\widetilde{F}(S\cup B)$, where $\widetilde{F}(S \cup A) = \sum_{i\in A,j\notin S\cup A} \widetilde{f}(S \cup A_{ij})$.  To see this can be arbitrarily bad, even when $\widetilde{F}\approx F$, consider an instance with $n-2$ elements $N'$ s.t. for any $S\subseteq N'$ the function evaluates to $f(S)=M$ for some arbitrarily large value $M>0$, and an additional subset of elements $A=\{a_1,a_2\}$ s.t. $f(A)= f(a_1) = f(a_2) = \epsilon$, for some arbitrarily small $\epsilon>0$.  Now assume that for any $S\subseteq N'$ and $i \in [2]$ we have $f(S\cup a_i)=M+\epsilon$.  The sampled mean of $A$ is maximal, its value is arbitrarily small.

\begin{figure*}[t]
\begin{centering}
                \includegraphics[trim = 0mm 0mm 0mm 0mm, height=40mm]{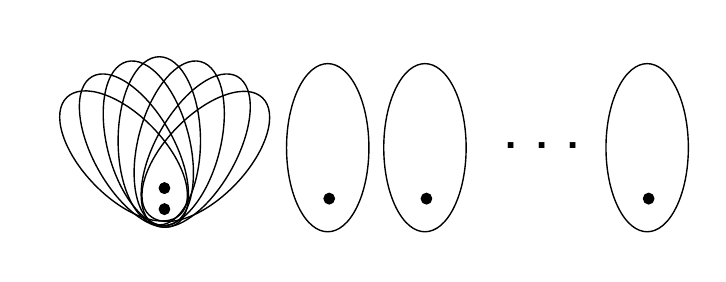}
                 \caption{\footnotesize{An instance of max-cover for which the greedy algorithm fails with access to an oracle with error.}}\label{fig:noise}
                 \end{centering}
\end{figure*}

\end{document}